\documentclass{revtex4-2}

\usepackage{amsthm}
\usepackage{ulem}
\usepackage{stix}

 \usepackage{amsfonts}
\usepackage{amsmath}
\usepackage{amssymb}
\usepackage{graphicx} 
\usepackage{dcolumn}
\usepackage{dsfont}
\usepackage{times} 
\usepackage{epsfig}
\usepackage[]{units}
\usepackage{float}

\usepackage{lipsum}
 
\usepackage[colorlinks,citecolor=red,urlcolor=blue,bookmarks=false,hypertexnames=true]{hyperref}

\usepackage{color}
\definecolor{colorRed}{rgb}{0.75,0.,.0}
\definecolor{colorGreen}{rgb}{0.,0.75,.0}
\definecolor{colorBlue}{rgb}{0.,0,1}
\definecolor{colormh}{rgb}{0.78,0.63,0.43}

\newcommand{\R}{{\mathbb R}}

\newcommand{\Z}{{\mathbb Z}}

\newcommand{\pair}[2]{\left(#1,#2\right)}

\newcommand{\rmd}{\mathrm{d}}

\def\sech{\mathop{\mathrm{sech}}}

\newcommand{\eps}{\varepsilon}

\def\epsilon{\varepsilon}

\def\beq{\begin{equation}}
\def\eeq{\end{equation}}

\def\dn{\mathrm{dn}}

\newcommand{\at}[1]{\left(#1\right)}

\usepackage{amsthm}
\newtheorem{lemma}{Lemma}

\begin{document}

\title{Hydrodynamics of a Discrete Conservation Law}


\author{Patrick Sprenger}
\affiliation{Department of Applied Mathematics, University of
  California Merced , Merced, California, 95343}

\author{Christopher Chong}
\affiliation{Department of Mathematics, Bowdoin College, Brunswick, Maine 04011}

\author{Emmanuel Okyere}
\affiliation{Department of Mathematics, Bowdoin College, Brunswick, Maine 04011}

\author{Michael Herrmann}
\affiliation{Institute of Partial Differential Equations, Technische Universit\"at Braunschweig, 38106 Braunschweig, Germany}

\author{P. G. Kevrekidis}
\affiliation{Department of Mathematics and Statistics, University of Massachusetts Amherst, Amherst, Massachusetts 01003-4515}

\author{Mark A. Hoefer}
\affiliation{Department of Applied Mathematics, University of
  Colorado, Boulder, Colorado 80309}

\begin{abstract}
  The Riemann problem for the discrete conservation law
  $2 \dot{u}_n + u^2_{n+1} - u^2_{n-1} = 0$ is classified using
  Whitham modulation theory, a quasi-continuum approximation, and
  numerical simulations.  A surprisingly elaborate set of solutions to
  this simple discrete regularization of the inviscid Burgers'
  equation is obtained.  In addition to discrete analogues of
  well-known dispersive hydrodynamic solutions---rarefaction waves
  (RWs) and dispersive shock waves (DSWs)---additional unsteady
  solution families and finite time blow-up are observed.  Two
  solution types exhibit no known conservative continuum correlates:
  (i) a
  counterpropagating DSW and RW solution separated by a symmetric,
  stationary shock and (ii) an unsteady shock emitting two
  counter-propagating periodic wavetrains with the same frequency
  connected to a partial DSW or a RW.  Another class of solutions
  called traveling DSWs, (iii), consists of a partial DSW connected to
  a traveling wave comprised of a periodic wavetrain with a rapid
  transition to a constant. Portions of solutions (ii) and (iii) are
  interpreted as shock solutions of the Whitham modulation
  equations. 
\end{abstract}

\date{\today}

\maketitle

\section{Introduction}

The hydrodynamics of the conservation law (inviscid Burgers' equation)
\begin{equation}
  \label{eq:1}
  u_t + \left ( u^2 \right )_x = 0, 
\end{equation}
with $u(x,t) \in \R$, $x,t \in \R$ is succinctly expressed by solutions of the
Riemann problem that consists of Eq.~\eqref{eq:1} for $t > 0$ subject
to the initial condition
\begin{equation}
  \label{eq:2}
  u(x,0) =
  \begin{cases}
    u_- & x \le 0 , \\
    u_+ & x > 0,
  \end{cases}
  \quad x \in \R ,
\end{equation}
for $u_\pm \in \R$.  Solutions must be interpreted in a weak sense and
depend intimately upon the regularization applied.  For the
\textit{viscous regularization} in which equation \eqref{eq:1} is
modified to Burgers' equation $u_t + (u^2)_x = \nu u_{xx}$ where
$\nu > 0$, the weak solution of \eqref{eq:1}---either a moving
discontinuity (shock) or rarefaction wave (RW)---is uniquely
determined by considering the strong vanishing viscosity limit
$\nu \to 0^+$ of Burgers' equation \cite{Whitham74}.  This results in
the well-known Rankine-Hugoniot jump condition for the speed
$V = (u_- + u_+)/2$ and Lax entropy condition $u_+ < V < u_-$ of
admissible shock solutions.  Regularization by more complex viscous
terms (nonlinear, higher order, and viscous-dispersive) generally
results in the same weak solution \cite{LeF02}.  An alternative
\textit{dispersive regularization} is the Korteweg-de Vries (KdV)
equation $u_t + (u^2)_x + \epsilon^2 u_{xxx} = 0$.  In this case, the
zero dispersion limit $\epsilon \to 0$ converges weakly in the sense
that it satisfies the KdV-Whitham modulation equations corresponding
to averaged conservation laws \cite{Whitham74,lax_small_1983}.
Gurevich and Pitaevskii recognized the physical importance of the
asymptotic approximation obtained by considering small but nonzero
$\epsilon$ and obtained the dispersive shock wave (DSW) solution of
the KdV equation for \eqref{eq:2} when $u_- > u_+$ as a self-similar
solution of the KdV-Whitham modulation equations \cite{GP73,Whitham74}.  They
also observed that when $u_- < u_+$, the KdV equation with small but
nonzero $\epsilon$ is well-approximated by the same RW obtained by
viscous regularization.  Dispersive shock waves are unsteady,
modulated nonlinear wavetrains connecting two distinct levels
\cite{Mark2016}.  In contrast to viscous regularization,
alternative dispersive regularizations can result in drastically
different Riemann problem solution behavior in the small dispersion
regime, particularly when higher order \cite{Sprenger_2020} or
nonlocal \cite{congy_dispersive_2021} dispersive terms are considered.
The multiscale dynamics of conservative nonlinear wave equations in
the small dispersion regime are generally referred to as dispersive
hydrodynamics \cite{Mark2016}.

In this paper, we study the dispersive hydrodynamics of the
\textit{discrete regularization} of Eq.~\eqref{eq:1}
\begin{equation}
  \label{eq:3}
  \frac{\rmd}{\rmd t} u_n + \tfrac12 \left ( u_{n+1}^2 - u_{n-1}^2
    \right ) = 0,
\end{equation}
by solving the Riemann problem
\begin{equation}\label{step}
  u_n(0) = \begin{cases}
             u_-, \quad & n \leq 0 \\
             u_+, \quad & n > 0 ,
           \end{cases}
\end{equation}
for \eqref{eq:3} where $n \in \Z, t\in \R, u=u_n(t) \in \R$.  Equation
\eqref{eq:3} is the simplest centered differencing scheme for the
hydrodynamic flux.  As recognized in the early days of computational
fluid dynamics by Von Neumann, and later clarified by Lax,
differencing schemes like \eqref{eq:3} introduce oscillations that
require mitigation if one wishes to converge strongly to viscously
regularized solutions that satisfy the conservation laws of fluid
dynamics \cite{hou_dispersive_1991}.  In this paper, we consider
Eq.~\eqref{eq:3}, subject to \eqref{step}, in its own right, divorced
from the aim of approximating solutions of Eq.~\eqref{eq:1}.  The
semi-discrete equation \eqref{eq:3} can be interpreted as the
dispersive regularization
\begin{equation}
  \label{eq:4}
  u_T + \frac{i}{\epsilon} \sin(-i \epsilon \partial_X) \left (u^2
  \right ) = 0 ,
\end{equation}
by introducing $T = \epsilon t$, $X = \epsilon n$ into
Eq.~\eqref{eq:3}, where $\epsilon > 0$ is the lattice spacing.  Here,
$i\sin(-i\partial_x)=\tfrac12\big(\exp(\partial_x)-\exp(-\partial_x)\big)$
is the pseudo-differential operator for the centered discrete
derivative and acts on a function $f$ in the variable $x$ via
\begin{align}
  \label{eq:55}
  \big(i\sin(-i\partial_x)f\big)(x)=\tfrac12\big(f(x+1)-f(x-1)\big).
\end{align}
Equation \eqref{eq:4} is similar to Whitham type evolutionary
equations
\cite{ehrnstrom_existence_2012,hur_modulational_2015,binswanger_whitham_2021}
except that it is further constrained to be band-limited.  For the
lattice equation \eqref{eq:3} at time $t$, the support of the
discrete-space Fourier transform of $u_n(t)$ is $[-\pi,\pi]$ due to
the smallest length scale set by the lattice spacing, whereas the
Fourier transform of quasi-continuum approximations is 
not generally
compactly supported.  As we will demonstrate, this fundamental
property of lattice equations introduces new hydrodynamic solution
features that do not appear within certain continuum limits of the
model.

Equation \eqref{eq:4} can be used to formally derive quasi-continuum
approximations by using Pad\'e approximants of
$\frac{i}{\epsilon}\sin(-i\epsilon \partial_X)$ for
$0 < \epsilon \ll 1$.  For example, the (1,3) Pad\'e approximant
\begin{equation}
  \label{eq:29}
  \frac{i}{\epsilon} \sin(-i \epsilon \partial_X) \approx
  \big ( 1- \frac{\epsilon^2}{6} \partial^2_{X} \big )^{-1}\partial_X ,
\end{equation}
leads to the Benjamin-Bona-Mahoney (BBM) equation absent the linear
convective term \cite{benjamin_model_1972}
\begin{equation}
  \label{eq:BBM}
  U_T + (U^2)_X - \frac{\epsilon^2}{6}
  U_{XXT} = 0.
\end{equation}
This quasi-continuum approximation, inspired by the work of Rosenau on
mass-spring chains~\cite{rosenau2,rosenau1}---see also the discussion
in \cite{Nester2001}---is expected to faithfully represent the
long-wavelength behavior of the lattice model \eqref{eq:3}. But it is
no longer band-limited.  In fact, outside of RWs and DSWs, the Riemann
problem solutions we obtain for the lattice equation \eqref{eq:3} bear
no resemblance to the corresponding Riemann problem solutions of
\eqref{eq:BBM} obtained in \cite{congy_dispersive_2021}.
Nevertheless, the quasi-continuum model \eqref{eq:BBM} admits exact
solitary and periodic traveling wave solutions that can be used to
approximate corresponding solutions of the lattice model \eqref{eq:3}.

A particular feature of the band-limited lattice equation
\eqref{eq:3}, and others with centered differences, is the existence
of stationary, period two (binary) oscillation solutions
\begin{equation}
  \label{eq:8}
  u_{2n} = \alpha, \quad u_{2n+1} = \beta, \quad n \in \Z,
\end{equation}
for any $\alpha,\beta \in \R$.  Equation \eqref{eq:3} was studied in
\cite{wilma} using extensive numerical simulations for certain types
of odd initial data and binary oscillations were found to play an
important role.  Allowing for slow spatio-temporal modulations of this
solution---$\alpha = \alpha(X,T)$, $\beta = \beta(X,T)$,
$T = \epsilon t$, $X = \epsilon n$, $0 < \epsilon \ll 1$---Turner and
Rosales \cite{wilma} obtained the modulation equations
\begin{equation}
  \label{eq:9} 
  \alpha_T + (\beta^2)_X = 0, \quad \beta_T + (\alpha^2)_X = 0 .
\end{equation}
Two important implications of the hydrodynamic-type equations
\eqref{eq:9} are: (i) the equations \eqref{eq:9} are elliptic whenever
$\alpha\beta < 0$ and (ii) the existence of discontinuous, shock solutions
satisfying Rankine-Hugoniot jump conditions.  The ellipticity of
\eqref{eq:9} was shown to provide a route through modulated binary
oscillations to blow up of solutions of \eqref{eq:3}.  

While the topic of DSWs and dispersive hydrodynamics has been more
extensively explored in the realm of continuum
media~\cite{dsw,Mark2016}, as already implied by the above discussion,
studies of the discrete realm have the potential to offer new and
intriguing wave features.  In addition, the motivation for such
explorations has significantly increased on account of a diverse range
of corresponding applications. A central topic is the study of
granular crystals and associated nonlinear
metamaterials~\cite{Nester2001}, consisting typically of elastically
interacting bead chains. There, a sequence of experimental efforts in
simpler~\cite{Hascoet2000,Herbold07}, as well as in progressively more
complex media, including dimers~\cite{Molinari2009} and the more
recent setup of hollow elliptic cylinders~\cite{HEC_DSW} have
manifested the spontaneous emergence of DSWs under suitable loading
conditions. However, this has not been the only setting where
``effectively discrete'' DSWs have experimentally emerged. Another
example is in nonlinear optics where such structures have appeared in
optical waveguide arrays \cite{fleischer2}.  Finally, and quite
recently, yet another setup has emerged, that of tunable magnetic
lattices~\cite{talcohen} in which ultraslow shock waves can arise and
be experimentally imaged.

Earlier interest in lattice shocks include the heyday of conservation
laws and shock waves in the fifties and sixties when material
scientists were interested in the compression of a solid by passage of
a very strong shock wave through materials
\cite{holian_atomistic_1995}.  Early numerical studies (molecular
dynamics simulations) depicted what we now call a lattice DSW in the
material's stress profile and recognized its unsteady character in a
one-dimensional anharmonic chain \cite{first_DSW}.  This contradicted
the basic assumption of steadiness underlying the Rankine-Hugoniot
jump conditions and led to some controversy in the field.  The DSW's
leading edge was then identified with a homoclinic traveling wave
solution (solitary wave) of a continuum approximation in
\cite{duvall_steady_1969}, later identified as a generic feature of
DSWs in continuum media \cite{GP73,gurevich_expanding_1984}.  The
controversy continued for about 15 years until, in 1979,
three-dimensional lattice simulations were shown to exhibit a
transition from unsteady to steady shock fronts due to transverse
strains \cite{holian_molecular_1979}.  It is worth noting that the
same transition from one-dimensional, unsteady (dispersive) to
multi-dimensional, steady (effectively viscous) shock propagation, was
recently observed in a completely different, ultracold atomic
superfluid \cite{mossman_dissipative_2018}.

\begin{figure}[t]
    \centering
    \begin{minipage}[c]{0.4\textwidth}
\centering
    \includegraphics[scale = 0.75]{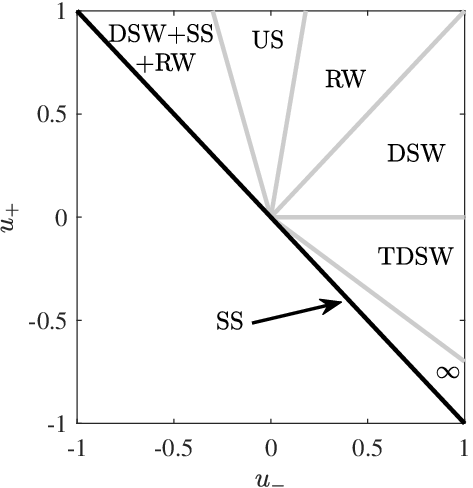} 
\end{minipage}
    \begin{minipage}[c]{0.45\textwidth}
\centering
 \begin{tabular}{l|l}
   \hline
   SS & Stationary shock \\ \hline
 RW & Rarefaction wave  \\ \hline
 DSW & Dispersive shock wave                               \\ \hline
 DSW+SS+RW & DSW + stationary shock + RW                          \\ \hline
 TDSW & traveling DSW     \\ \hline
  $\infty$ & Blow up         \\  \hline  
   US & Unsteady shock\\ \hline
\end{tabular}
\end{minipage}\vspace{0.25 cm}

    \includegraphics{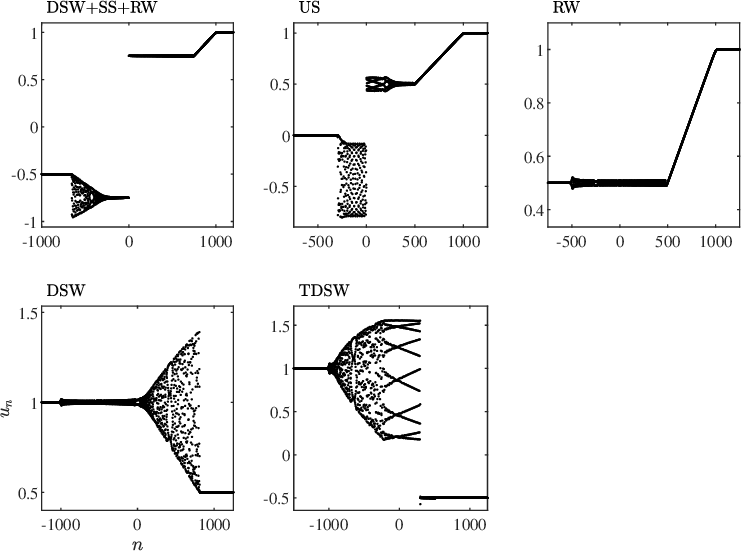}
    \caption{Classification of the Riemann problem \eqref{step} of the
      discrete conservation law \eqref{eq:3}. The numerical
      computations are shown at $t = 500$ and have been
      performed with the initial data in \eqref{eq:35}.}
    \label{fig:example_sols}
\end{figure}
These works have motivated the present authors to revisit ``lattice
hydrodynamics'' and the prototypical settings where DSW structures can
arise in nonlinear dynamical lattices. Canonical examples of first
order in time, quadratically nonlinear lattice nonlinear ordinary
differential equations (ODEs) in a conservation law form were
discussed extensively in the work of~\cite{wilma}. A subset of the
present authors has recently revisited this class of models
in~\cite{CHONG2022133533} attempting to incorporate tools from Whitham
modulation theory, bringing to bear both data-driven, as well as more
theoretically-inspired quasi-continuum approaches to obtain an
effective dimensional reduction, through an ODE description of the DSW states.
Our aim here is to expand on this work, offering a more systematic
classification of the possible solutions of such models by using
Whitham modulation theory and, where appropriate,
quasi-continuum approximation considerations supplemented by
direct numerical simulations.

We select the scale invariant, representative nonlinear example
\eqref{eq:3} within the class of models of~\cite{wilma} and set up the
corresponding discrete Riemann problem \eqref{step}.  Figure
\ref{fig:example_sols} depicts our phase diagram as a partitioning of
the parameter space $(u_-$,$u_+)$ and identifies seven distinct
solution behaviors. Some of these, such as the possibility of a RW or
a DSW as well as that of blow up are to a certain degree expected or
have been argued to be present previously~\cite{wilma}.  They are
labeled RW, DSW, and $\infty$ in Fig.~\ref{fig:example_sols},
respectively.  However, there are choices of initial conditions that
yield less common dynamical outputs, some of which are genuinely
discrete in nature with labels in Fig.~\ref{fig:example_sols}
identified parenthetically.  These include, for instance, a
stationary, symmetric shock on its own (SS) or separating a RW and a
DSW (DSW+SS+RW).  Another example is a traveling DSW (TDSW), which
consists of a partial DSW connected to a heteroclinic
periodic-to-equilibrium traveling wave.  Arguably the most complex
structure encountered is an unsteady shock evolving between two
distinct traveling waves featuring the same temporal frequency (US).
In what follows, we explain our partitioning of the phase diagram
\ref{fig:example_sols} into the regions pertaining to these different
dynamical behaviors and we offer a set of tools that can be used to
understand each one, as well as unveil some open directions for future
exploration.

There exists a family of Riemann problems that can be considered in the
discrete setting by modifying the value $u_0(0)$ in \eqref{step}.  For
example, Turner and Rosales set $u_0(0) = 0$ and $u_- = -u_+$
\cite{wilma}. We primarily focus on the data \eqref{step} in which
$u_0(0) = u_-$ resulting in the phase diagram of
Fig.~\ref{fig:example_sols}.  While changing $u_0(0)$ does not affect
the observed solution phases, it does change the phase boundaries.  We
interpret this microscopic modification of the initial data impacting
the macroscopic properties of solutions as an indication of
non-uniqueness of the Riemann problem for the dispersive
regularization \eqref{eq:29}.

It is important to distinguish our use of the term ``shock'' or
``shock wave'' from the classical notion of discontinuous weak
solutions of inviscid Burgers' equation \eqref{eq:1}.  We identify
four classes of shock solutions to the discrete equation \eqref{eq:3}
by prefacing each with a descriptor in Fig.~\ref{fig:example_sols}.
The simplest is the symmetric, \textit{stationary} shock (SS) solution
of the lattice
\begin{equation}
  \label{eq:35}
  u_n(t) =
  \begin{cases}
    -u_0 & n \le 0, \\
    u_0 & n > 0,
  \end{cases}
\end{equation}
where $u_0 > 0$.  The other shock solutions can be understood as
special solutions of the first-order, quasi-linear Whitham modulation
equations corresponding to Eq.~\eqref{eq:3} that are described in
Sec.~\ref{sec:whitham-theory}.  The unsteady dispersive shock wave
(DSW) is approximated by a nonlinear, periodic wavetrain modulated by
a rarefaction wave solution of the Whitham modulation equations.  Note
that the DSW does not satisfy the Rankine-Hugoniot jump conditions of
the Whitham modulation equations.  On the other hand, the unsteady
shock (US) is approximated by two periodic traveling waves that
satisfy the jump conditions for the Whitham modulation equations.  The
traveling dispersive shock wave (TDSW) consists of both an unsteady
partial DSW---approximated by a rarefaction solution of the Whitham
modulation equations---and a steady traveling wave---approximated by a
periodic traveling wave and a solitary wave that satisfy the jump
conditions for the Whitham modulation equations.  There is an
important distinction between the traveling wave and US as
discontinuous solutions of the Whitham modulation equations.  The
phase speeds of the two periodic traveling waves in the US solution
differ from one another and from the shock speed, which is zero.  On
the other hand, the traveling wave solution consists of a single
periodic traveling wave whose phase speed is the same as the shock
speed.  For clarity, we summarize the four distinct uses of the term
``shock'' in this paper:
\begin{enumerate}
\item the stationary lattice shock (SS) \eqref{eq:35};
\item the unsteady dispersive shock wave (DSW) that is approximated by
  a rarefaction solution of the Whitham modulation equations;
\item the unsteady shock (US) that is approximated by a discontinuous
  shock solution of the Whitham modulation equations;
\item the traveling dispersive shock wave (TDSW) that is approximated
  by a shock-rarefaction solution of the Whitham modulation equations.
\end{enumerate}

Our presentation will be structured as follows.  In section
\ref{sec:model-equations}, we present the model equations, as well as
the principal setup and notation for our study.  In section
\ref{sec:whitham-theory}, we focus on the Whitham modulation equation
formulation for the discrete problem. We discuss the corresponding
conservation laws and how their averaging can provide information for
the DSW features of our model.  Section \ref{sec:class} is dedicated
to the systematic classification of our solutions in the different
parametric regimes, accompanied by illustrative numerical computations
of the different identified waveforms.  In section
\ref{sec:non-uniq-riem}, we show how modification of the Riemann data
\eqref{step} at a single site can lead to drastically different
solution behaviors.  Finally, in section \ref{sec:concl-future-chall},
we summarize our findings and present our conclusions, as well as a
number of open questions for further research into this budding theme.

\section{Model Equations}
\label{sec:model-equations}

It will be beneficial to generalize Eq.~\eqref{eq:3} and consider the
discrete scalar conservation law~\cite{wilma}
\begin{equation} \label{deq}
2 \frac{du_n}{dt} +  \Phi'(u_{n+1} )-  \Phi'(u_{n-1} ) =0,
\end{equation}
a discretization of the more general conservation law
$u_t + \Phi'(u)_x = 0$, where $n\in \Z, t\in \R, u=u_n(t) \in \R$ and
the potential $\Phi(u)$ is assumed to be smooth with $\Phi'(u)$ a
convex function of its argument $\Phi'''(u) \ne
0$. Equation~\eqref{deq} possesses a Lagrangian and Hamiltonian
structure \cite{Herrmann_Scalar}, yet it is first order only, making
its analysis slightly more convenient when compared to classical
nonlinear oscillators, such as those of the Fermi-Pasta-Ulam-Tsingou
(FPUT) type \cite{FPU55}.  Besides serving as a prototype model for
lattice DSWs, Eq.~\eqref{deq} is also of interest for applications,
such as in the description of traffic flow \cite{Whitham90}; for a
discussion of relevant models and their continuum limits see also
Ref.~\cite{wilma}.
 
In this paper, we primarily focus on the potential
\begin{equation}\label{pot}
\Phi(u) = \frac{u^3}{3} .
\end{equation}
For this choice, the ``mass''
$$M(t) = \sum_n u_n(t) $$
and ``energy''
$$ E(t) = \sum_n \Phi(u_n(t)), $$
when well-defined, are conserved in the infinite lattice and in a
finite lattice with periodic boundary conditions.  The linear
dispersion relation for Eq.~\eqref{deq} is
\begin{equation}
  \label{eq:18}
  \omega_0(k,\bar{u}) = \Phi''(\bar{u}) \sin(k), \quad k,\bar{u} \in \R,
\end{equation}
for linearized wave solutions of the form
$u_n(t) = \bar{u} + a e^{i(kn-\omega_0 t)}$, $|a| \ll 1$.  Throughout
the manuscript we consider the Riemann, step initial data
\eqref{step}.
For numerical simulations, the infinite lattice is truncated by
introducing $N > 0$ (even) to represent the number of lattice
sites. The corresponding spatial domain is $-N/2+1 < n \leq N/2$ and
the simulation temporal domain is $[0, T_f / \eps]$, where $T_f$ is a
fixed constant independent of $\epsilon= 1/N$. We use free boundary
conditions $u_{-N/2} = u_{-N/2+1}$ and $u_{N/2} = u_{N/2-1}$ in
conjunction with the initial data Eq.~\eqref{step}, and we choose
domain sizes large enough that interactions with the boundary are
negligible.  When investigating finite time blow up, we employ
periodic conditions. This allows us to monitor if the rescaled
quantities $E(t) \to E(t)/N$ and $M(t) \to M(t)/N$ are conserved
(details in sec.~\ref{sec:blow-up}).  A variational integrator is used
for simulations, see \cite{Herrmann_Scalar}. Simulations were also
carried out with a Runge-Kutta method to check for consistency which
yielded negligible differences on the time scales considered in this
paper (for cases that did not involve blow up features).   Due to the scaling symmetry
$t \to a t$, $u_n \to a u_n$, for any nonzero $a \in \R$ of equation
\eqref{deq} subject to \eqref{pot}, we can set either $u_+ = 1$ or
$u_- = 1$ without loss of generality. Figure \ref{fig:example_sols}
shows a classification of the zoology of solutions that arise from the
Riemann problem. They include
rarefaction waves (RWs) for $u_+= 1$ and $u_-\in(0.18,1)$, 
dispersive shock waves (DSWs) for $u_+ \in (0,1)$ and $u_- = 1$, 
solutions consisting of dispersive shock waves, stationary shocks and rarefaction waves (DSW + SS + RW) for $u_+= 1$ and $u_-\in(-1,-0.26)$, 
solutions consisting of traveling dispersive shock waves (TDSW) for $u_+ \in (-0.724,0)$ and $u_- = 1$ ,
unsteady shocks (US) for $u_+= 1$ and $u_-\in(-0.26,0.18)$, 
and blow up ($\infty$) for $u_+ \in (-1,-0.724)$ and $u_- = 1$.
The region boundaries are approximate. In sec.~\ref{sec:class}, we
provide a detailed analysis for each of the five solution types just
described, starting first with the simplest, and moving through them
gradually in terms of their complexity according to the table in
Fig.~\ref{fig:example_sols}.  We employ a number of tools for the
study of these solutions, including direct numerical simulation,
fixed-point iteration schemes, modulation theory, weak solutions, DSW
fitting, and quasi-continuum modeling.  The details of these
approaches will be given in the sections they are employed, with the
exception of modulation theory. This analysis is slightly more
involved, and thus has a dedicated section. Our intention in
presenting these tools is to leverage this specific, but interesting
in its own right, example in order to utilize a variety of techniques
that may be of broader relevance to applications in other Hamiltonian
nonlinear dynamical lattices. It would be of particular
interest to identify similar phenomena or/and to leverage
the techniques utilized herein in other dispersive,
nonlinear lattice models.

\subsection{An alternative, integrable discretization}
\label{sec:an-altern-integr}

Prior to describing the solutions of Eq.~\eqref{eq:3} depicted in
Fig.~\ref{fig:example_sols}, we briefly comment on the alternative
discretization
\begin{equation}
  \label{eq:5}
  \frac{\rmd u_n}{\rmd t} + u_n\left ( u_{n+1} - u_{n-1} \right ) = 0 ,
\end{equation}
of Eq.~\eqref{eq:1} subject to \eqref{pot}.  This equation was studied
in \cite{LAX86,goodman_dispersive_1988} where it was shown to exhibit
DSWs and, for positive data, to be completely integrable by a
transformation \cite{KAC1975160} to an equation related to the Toda
lattice \cite{toda}.  We have performed numerical simulations of
Eq.~\eqref{eq:5} subject to the Riemann data \eqref{step} and observe
DSWs when $u_- > u_+ \ge 0$, RWs when $u_+ > u_- \ge 0$ and blow-up
when $u_+$ and $u_-$ exhibit opposite signs. Examples of numerical
simulations of DSWs and RWs that emerge from strictly positive initial
Riemann data are shown in Fig. \ref{fig:kvm}

\begin{figure}[H]
    \centering
    \includegraphics{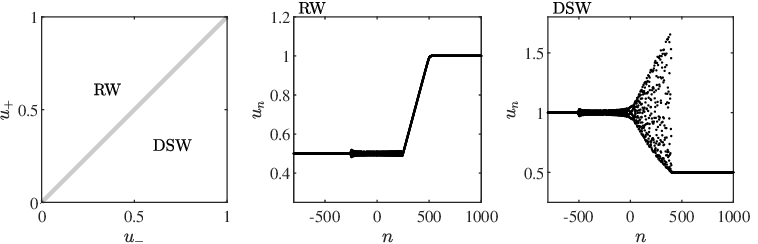}
    \caption{Classification of the Riemann problem of \eqref{step}  (left panel) and typical time evolutions of the Riemann data for a RW (middle) and a DSW (right panel) 
    within Eq. \eqref{eq:5}. }
    \label{fig:kvm}
\end{figure}

Of course, complete integrability confers a great deal of mathematical
structure.  Whitham modulation theory for the Toda lattice was
developed in \cite{Holian81,kodama_solutions_1990,Bloch_Toda92} while
the inverse scattering transform for the Toda lattice with step-type
initial data was developed in \cite{venakides_toda_1991}; see also the
recent discussion of Whitham theory applied to DSWs in the Toda
lattice \cite{biondini_whitham_2023}.  Collectively, these works
support our numerical observation that, for positive Riemann data,
Eq.~\eqref{eq:5} exhibits only RW and DSW solutions.  These Riemann
problem solution behaviors are to be contrasted with those depicted in
Fig.~\ref{fig:example_sols}.  Although integrability of
Eq.~\eqref{eq:5} is lost for sign indefinite initial data, the only
dynamics we numerically observe are indicative of blow up.  Thus, the
discretization \eqref{eq:3} we focus on in this paper, admits a wider
variety of dynamics than the integrable alternative of \eqref{eq:5}.

\section{Whitham Theory}
\label{sec:whitham-theory}

\subsection{Modulation equations for a continuum system} 

In this section, we consider a continuum model system by introducing
the interpolating function $u(x,t)$ such that $u(n,t) = u_n(t)$ for
all $n \in \mathbb{Z}$. This allows us to represent the advance-delay
operator in the discrete system \eqref{deq} as a pseudo-differential
operator. The resulting continuum model is
\begin{align}\label{eq:cont}
u_t + i\sin(-i\partial_x) \left(\Phi'(u)\right) = 0 \,,
\end{align}
Equation \eqref{eq:cont} can be written in the Hamiltonian form 
\begin{align}
u_t = J \frac{\delta H}{\delta u}, 
\end{align}
where $J = i \sin(-i\partial_x)$ is the antisymmetric operator and
$H = \int \Phi(u) \ dx$ is the Hamiltonian. Periodic solutions of
\eqref{eq:cont} are of the form $u(x,t) = \varphi(\theta;\mathbf{q})$
with phase $\theta = kx - \omega t$ and parameters
$\mathbf{q} \in \R^3$ (e.g., wavenumber, amplitude, mean).  They satisfy
\begin{align}
\label{eqn:TW}
  - \omega \varphi_\theta  + i\sin(-ik\partial_\theta)
  \Phi'(\varphi) = 0\,,
\end{align}
which is equivalent to the nonlinear advance-delay differential equation
\begin{align}
\label{eqn:TW2}
     2\,\omega\, \varphi_\theta\at\theta=\Phi'\big(\varphi(\theta+k)\big)-\Phi'\big(\varphi(\theta-k)\big)
\end{align}
The solution theory of such nonlocal equations is rather intricate but
the existence of a three-parameter family of traveling waves has been
establshed in \cite{Herrmann_Scalar} by variational
techniques. Integrating \eqref{eqn:TW} once with respect to $\theta$
\begin{align}
  \label{eq:6}
-c_{\rm p} \varphi  + {\rm sinc}(-ik\partial_\theta) \Phi'(\varphi)= A, 
\end{align}
where $c_{\rm p} = \frac{\omega}{k}$ is the phase speed and $A$ is a
real constant. The pseudodifferential operator is then interpreted as
a multiplier on the Fourier coefficients of $\varphi$,
\begin{align}
\sin(-ik\partial_\theta) \Phi'(\varphi) &= \sum_{n} \sin(n k) \hat{p}_n e^{in \theta}, \\
\hat{p}_n & = \frac{1}{2\pi} \int_0^{2\pi} \Phi'\left(\varphi(\theta)\right)e^{-in\theta} d\theta. 
\end{align}
Equation \eqref{eq:cont} possesses the two conserved quantities
\begin{align}
M(t) &= \int u dx \\
E(t) &= \int \Phi(u) dx,
\end{align}
where the domain of integration is determined by the decay or
periodicity of $u$.  We now seek the modulation equations for a
periodic wave with the slowly varying ansatz
\begin{align}
  \label{eq:7}
  u(x,t)= \varphi(\theta;\mathbf{q}(X,T)) + \epsilon
  \varphi_1(\theta,X,T) + \cdots , \quad X =\epsilon x, \quad T
  =\epsilon t, \quad
  0<\epsilon \ll 1,
\end{align}
in which the leading order term $\varphi(\theta;\mathbf{q})$ is the
periodic traveling wave solution satisfying \eqref{eq:6} with vector
of parameters $\mathbf{q}$ that varies on the slow scales $X$ and $T$
while $\varphi$ is 2$\pi$-periodic in $\theta$. We impose the
generalized wavenumber and frequency relationships $\theta_x = k$ and
$\theta_t = -\omega$ along with their compatibility
\begin{align}
k_T + \omega_X = 0 . 
\end{align}
\begin{lemma} \label{lemma:operator}
The nonlocal operator acting on a modulated periodic function $g(\theta,X,T) \in \mathcal{C}^1$ has the multiple scale expansion 
\begin{align}\label{eq:mult_scales}
\begin{split}
\sin(-i\partial_x) g &= \sin(-i k\partial_\theta - i\epsilon \partial_X) g\\
& \sim \sin(-i k \partial_\theta) g - i \frac{\epsilon}{2}\bigg(\cos(-i k \partial_\theta) g_X + \big(\cos(-i k \partial_\theta) g\big)_X\bigg) + \mathcal{O}(\epsilon^2)\\
\end{split}
\end{align}
\end{lemma}
\begin{proof}
  The proof follows from the analyticity of $\sin(\cdot)$. A detailed
  proof follows all of the ideas in \cite{binswanger_whitham_2021}.
\end{proof}

We now average Eq.~\eqref{eq:cont} and its higher order conserved
densities by introducing the averaging operator
\begin{equation}
  \label{eq:15}
  \overline{F[\varphi]}(X,T) = \frac{1}{2\pi} \int_{0}^{2\pi}
  F[\varphi(\theta;\mathbf{q}(X,T))]\,\mathrm{d}\theta ,
\end{equation}
where $F[u(x,t)] = F(u,u_x,u_t,u_{xt}, \ldots)$ is a local function of
$u$ and its derivatives.  If $u$ is a multiscale
function of the form $u = g(\theta,X,T)$, then
\begin{equation}
  \label{eq:16}
  \begin{split}
    \overline{\partial_t F} &= -\omega
    \overline{\partial_\theta F} + \epsilon
    \partial_T \overline{F} = \epsilon
    \partial_T \overline{F}, \\
    \overline{\partial_x F} &= k \overline{\partial_\theta F} + \epsilon
    \partial_X \overline{F} = \epsilon
    \partial_X \overline{F},
  \end{split}
\end{equation}
by virtue of the fact that $F = F[g]$ is periodic in $\theta$ so the
period average of $\partial_\theta F[g]$ is zero. 
We use Lemma
\ref{lemma:operator} to compute averages of, for example,
$\sin(-i k \partial_\theta)g$ for any $g \in L_2([0,2\pi])$ with the
Fourier series $g = \sum_n \hat{g}_n e^{in\theta}$:x
\begin{align}
\begin{split}
\overline{\sin(-i k \partial_\theta)g} & = \overline{\sum_n \sin(n k) \hat{g}_n e^{in\theta}}
 = 0 .
\end{split}
\end{align}
We now insert the multiple scales ansatz \eqref{eq:7} into the two
conservation laws associated with \eqref{eq:cont}, and average. This
procedure results in the system of conservation laws
\begin{subequations}\label{eq:cont_mod}
  \begin{align}
    \label{eq:21}
    \overline{\varphi}_T + \overline{\Phi'(\varphi)}_X & = 0 \\
    \label{eq:22}
    \overline{\Phi(\varphi)}_T + \frac{1}{2}\left(\sum_m \cos(mk)
    \hat{p}_{m}^2\right)_X & = 0 \\
    \label{eq:23}
    k_T + \omega_X & = 0 
\end{align}
\end{subequations}

In the vanishing amplitude $a \to 0$ limit, Eqs.~\eqref{eq:21} and
\eqref{eq:22} become the Hopf equation
\begin{equation}
  \label{eq:20}
  \bar{u}_T + \Phi''(\bar{u})\bar{u}_X = 0,
\end{equation}
for the mean $\bar{u} = \overline{\varphi}$, and the conservation of
waves equation \eqref{eq:23} corresponds to linear wave modulation
theory with frequency $\omega = \omega_0$ given by the linear
dispersion relation \eqref{eq:18}.

The nonlinear modulation equations can alternatively be derived by
employing Whitham's other method of an averaged Lagrangian functional,
see for instance \cite[chapter 14]{Whitham74} and
\cite{benzoni-gavage_slow_2014} for symplectic PDEs,
\cite{Venakides99, DHM06, DH08, HR10} for an application to FPUT chains as the
most prominent example of Hamiltonian lattices, and
\cite{CHONG2022133533} for the discrete conservation law
\eqref{deq}. In this setting, the modulation equations take the form
\begin{subequations}\label{eqn:WhithamAbstract}
  \begin{align}
    \label{eqn:WhithamAbstractV}
   \overline{u}_T+(E_ {\overline{u}})_X &=0\\
    \label{eqn:WhithamAbstractK}
   k_T + (E_S)_X&=0
    \\%
    \label{eqn:WhithamAbstractS}
    S_T + (E_k)_X & = 0 \,.
\end{align}
This is a system of Hamiltonian PDEs with density variables $\overline{u}$, $k$,
and $S$, which represent the wave mean, the nonlinear wave number, and a
nonlocal auxiliary variable that might be regarded as a generalized
wave momentum. Moreover, the equation of state $E=E(\overline{u},k,S)$ describes
the energy of a traveling wave and its partial derivatives provide the
fluxes in \eqref{eqn:WhithamAbstract}. The energy is also conserved
according to the extra conservation law
\begin{align}
\label{eqn:WhithamAbstractE}
E_T+\big(\tfrac12 E_{\overline{u}}^2+E_kE_S \big)_X=0\,,
\end{align}
\end{subequations}
which is implied by \eqref{eqn:WhithamAbstract} thanks to the chain rule.
A closer look to the derivation of \eqref{eqn:WhithamAbstract} in \cite{CHONG2022133533} reveals that \eqref{eq:21} and \eqref{eq:23} correspond to \eqref{eqn:WhithamAbstractV} and 
\eqref{eqn:WhithamAbstractK}, respectively, while \eqref{eq:23} is the analogue to \eqref{eqn:WhithamAbstractE}.
A complete 
understanding of \eqref{eq:cont_mod} and \eqref{eqn:WhithamAbstract} is currently out of reach because we are not able to characterize the analytical properties of the constitutive relations since these depend in a very implicit and not tractable way on the three-dimensional solution sets of the nonlinear advance-delay-differential equation \eqref{eqn:TW}. For instance, it is not even clear for which values of the parameters the Whitham system \eqref{eqn:WhithamAbstract} is hyperbolic or genuinely nonlinear. For this reason we do not work with the full lattice modulation equations directly but combine different approximation procedures with a careful evaluation of numerical data.
\subsection{Relation to the lattice dynamics}

Although neither analytical nor numerical solutions to the nonlinear modulation systems 
\eqref{eq:cont_mod} or \eqref{eqn:WhithamAbstract} are available, we
can extract important partial information from numerical simulations
of initial value problems to \eqref{deq}. The key observation is that
the lattice ODE as well as an implied energy equation represent
discrete counterparts of local conservation laws and transform under
the hyperbolic scaling of space and time into first order PDEs. To see
this, we fix a \textit{window function} $\chi$ that depends smoothly
on the macroscopic variables $(X,T)$, decays sufficiently fast, and
has normalized integral. Using a shifted copy of $\chi$, we are able
to quantify the local moments of any microscopic observable near a fixed macroscopic point. For instance, the average
\begin{align}
\label{eq:MesAv}
    \langle u_n\rangle\pair{X}{T}=\eps^2\sum_{\breve{n}}\int u_{\breve{n}}(\breve{t}\,)\,\chi(\eps \breve{t}-T,\eps \breve{n}-X)\,d{\breve{t}}\,,\qquad T=\eps\, t\,,\qquad X= \eps\,n
\end{align}
represents the \textit{mesoscopic space-time averages} of $u_n(t)$ near the macroscopic point $(X,T)$.

\begin{lemma} Any bounded solution to \eqref{deq} satisfies in the hyperbolic scaling limit $\eps\to0$ the conservation laws
\begin{subequations}\label{eq:disc_mod}
  \begin{align}
    \label{eq:a_disc_mod}
    \partial_T \big\langle u_n \big\rangle  + \partial_X \big\langle \Phi'(u_n) \big\rangle =
    0, \\
    \label{eq:b_disc_mod}
    \partial_T \big\langle \Phi(u_n)\big\rangle + \partial_X
    \big\langle\tfrac{1}{2} \,\Phi'(u_n)\, \Phi'(u_{n+1})\big\rangle =0,
  \end{align}
\end{subequations}
provided that these are interpreted in a distributional sense.
\end{lemma}

\begin{proof} We only give an informal derivation but mention that an alternative and more elegant framework is provided by the theory of Young measures. The latter can also be applied to non-smooth window functions $\chi$ and reveals that the mesoscopic averages $\langle\cdot\rangle$ can be expected to be independent of the particular choice for $\chi$.
Using the abbreviation $p_n = \Phi'(u_n)$, discrete integration by parts as well as the smoothness of $\chi$ we verify
\begin{align*}
    \big\langle\tfrac12\, p_{n+1}-\tfrac12\, p_{n-1}\big\rangle\pair{X}{T} &=
    \eps^3\sum_{\breve{n}}\int p_{\breve{n}}(\breve{t}\,)\,\Big(\tfrac12\, \chi(\eps \breve{t}-T,\eps \breve{n}-\eps-X)-
    \tfrac12\, \chi(\eps \breve{t}-T,\eps \breve{n}+\eps-X) \Big)\,d{\breve{t}}
    \\&=-
    \eps^3\sum_{\breve{n}}\int p_{\breve{n}}(\breve{t}\,)\,\partial_X \chi(\eps \breve{t}-T,\eps \breve{n}-X)\,d{\breve{t}}+\text{h.o.t.}
    \\&=
    -\eps\,\partial_X \big\langle p_n\big\rangle\pair{X}{T}+\text{h.o.t.}
\end{align*}
and by similar computations we obtain
\begin{align*}
    \big\langle\dot{u}_{n+1}\big\rangle\pair{X}{T} =
    -\eps\,\partial_T\big\langle u_n\big\rangle\pair{X}{T}+\text{h.o.t.}\,.
\end{align*}
The asymptotic validity of \eqref{eq:a_disc_mod} is thus a direct consequence of the microscopic dynamics \eqref{deq}, the definition of the bracket $\langle\cdot\rangle$ in \eqref{eq:MesAv}, and the hyperbolic scaling. The lattice ODE \eqref{deq} implies
with
\begin{align*}
    \frac{d}{d t} \Phi(u_n)+ \big( \tfrac12\, \Phi^\prime(u_{n})\,\Phi^\prime(u_{n+1})\big)-\big( \tfrac12\, \Phi^\prime(u_{n-1})\,\Phi^\prime(u_{n})\,\big)=0\,,
\end{align*}
another discrete conservation law (in which the time derivative of a
density is balanced by the discrete divergence of a flux quantity), so
the second claim \eqref{eq:a_disc_mod} can be justified along the same
lines.
\end{proof}

There is an important difference between the conservation laws in
\eqref{eq:cont_mod} and \eqref{eq:disc_mod}.  The PDEs in
\eqref{eq:cont_mod} (and likewise those in
\eqref{eqn:WhithamAbstract}) are derived under the \textit{hypothesis}
that the lattice solution can be approximated by a modulated traveling
wave, see \eqref{eq:7}, and the closure relations involve the
(unknown) profile functions for traveling lattice waves as well as
averages with respect to the scalar phase variable $\theta$.
Numerical simulations with well-prepared initial data (e.g., the
Riemann initial data \eqref{step}) indicate that the approximation
assumption concerning the microscopic data is indeed satisfied but no
rigorous proof is available, neither for the lattice \eqref{deq} nor
for FPUT chains with convex interaction potential. The only exceptions are
the few completely integrable cases but the details are still
complicated and involve special coordinates related to the Lax
structure. In particular, even for the lattice of Eq.~\eqref{eq:5}
and the Toda chain
it is not easy to compute how the phase averages in the modulation
equations depend on the traveling wave parameters.

The status of \eqref{eq:disc_mod} is completely different. The two
PDEs can be established under very mild assumptions (boundedness of
lattice solutions) and by means of fundamental mathematical principles
(such as integration by parts and compactness in the sense of Young
measures). They reflect universal constraints for the
macroscopic dynamics, do not require any a-priori knowledge on the
fine structure of the microscopic oscillations, and hold for a large
class of initial data (which might even be oscillatory or random).  
Moreover, the
mesoscopic space-time averages can easily by extracted from numerical
data. In the simplest case, we use a straightforward box counting with
space-time windows of microscopic length $1/\sqrt{\eps}$ (or 
macroscopic length $\sqrt{\eps}$).  Of course, \eqref{eq:disc_mod}
does not provide a complete set of macroscopic equations and without
further information it is not clear whether or how the fluxes can be
computed in a pointwise manner from the densities. The equations are
nevertheless very useful since they allow us to derive and check
partial information on the solution of the modulation equations from
numerical data. In particular, in the context of modulated traveling
waves, the PDEs \eqref{eq:a_disc_mod} and \eqref{eq:b_disc_mod}
correspond to \eqref{eq:21} and \eqref{eq:22}, respectively.

\subsection{Self-similar solutions}

The Whitham modulation equations \eqref{eq:cont_mod} are a system of
conservation laws that can be compactly expressed in the form
\begin{align}
  \label{eq:11}
  \mathbf{P}(\mathbf{q})_T + \mathbf{Q}(\mathbf{q})_X = 0, \quad
  \mathbf{q} = [\bar{u},a,k]^{\rm T}, 
\end{align}
where the vectorial density $\mathbf{P}$ and flux $\mathbf{Q}$ depend
on the slowly varying parameters $\mathbf{q}$ through integrals of the
periodic orbit $\varphi$.  Equation \eqref{eq:11} can also be
expressed in the form
\begin{equation}
  \label{eq:14}
  \mathbf{q}_T + \mathcal{A} \mathbf{q}_X = 0, \quad \mathcal{A} =
  \left ( \frac{\partial 
      \mathbf{P}}{\partial \mathbf{q}} \right )^{-1} \frac{\partial
    \mathbf{Q}}{\partial \mathbf{q}},
\end{equation}
provided the inverse is nonsingular.  We will use solutions of the
Whitham equations to approximate the long time dynamics of solutions
to the Riemann problem \eqref{eq:3}, \eqref{step}.  Consequently, it
is natural to consider the Riemann problem
\begin{align}
  \label{eq:12}
  \mathbf{q}(X,0) = \begin{cases}
             \mathbf{q}_- & X < 0\\ 
             \mathbf{q}_+ & X > 0 
           \end{cases}, 
\end{align}
for the Whitham equations \eqref{eq:11} themselves.  Rarefaction
(simple) wave solutions and discontinuous shock solutions of the
binary oscillation modulation system \eqref{eq:9} were used in
\cite{wilma} to interpret various features of the numerical solutions.
In this work, we will make use of rarefaction wave and discontinuous
shock solutions of the more general Whitham modulation equations
\eqref{eq:cont_mod}.

The invariance of the Riemann problem \eqref{eq:11}, \eqref{eq:12}
with respect to the hydrodynamic scaling $X \to \sigma X'$,
$T \to \sigma T'$ for real $\sigma \ne 0$ suggests seeking
self-similar solutions in the form $\mathbf{q} = \mathbf{q}(\xi)$,
$\xi = X/T$.  Equation \eqref{eq:14} possesses rarefaction waves
 satisfying \cite{dafermos_hyperbolic_2016}
\begin{equation}
  \label{eq:17}
  \frac{\rmd \mathbf{q}}{\rmd \xi} = \frac{\mathbf{r}_i}{\nabla \lambda_i
    \cdot \mathbf{r}_i}, \quad \mathbf{q}(\xi_\pm) = \mathbf{q}_\pm,
  \quad \xi_- < \xi_+,  
\end{equation}
where $\mathcal{A} \mathbf{r}_i = \lambda_i \mathbf{r}_i$ and
$\lambda_i = \xi$, provided the characteristic field is genuinely
nonlinear $\nabla \lambda_i \cdot \mathbf{r}_i \ne 0$.  Since
$\mathbf{q}_\pm$ lie on the same, one-dimensional integral curve, they
are constrained by two integral relations resulting from integration
of the third order ODEs \eqref{eq:17}.  Admissibility requires
$\xi_- < \xi_+$.  The eigenvalues $\lambda_i$ can be interpreted as
speeds. For example, in the context of DSWs, the trailing edge speed
is $c_-=\xi_-$ and the leading edge speed $c_{+} = \xi_+$.

Another class of self-similar solutions are discontinuous shock
solutions to the Whitham system \eqref{eq:cont_mod}
\begin{align}
  \mathbf{q}(\xi) = \begin{cases}
             \mathbf{q}_- & \xi < V \\ 
             \mathbf{q}_+ & \xi > V 
           \end{cases}, 
\end{align}
where $V$ is the velocity of the shock solution that satisfies the
Rankine-Hugoniot jump conditions
\begin{align}
  \label{eq:19}
  - V [\![ \mathbf{P} ]\!] + [\![ \mathbf{Q} ]\!] = 0 .
\end{align}
The brackets $ [\![ \cdot ]\!]$ denote the jump in its argument evaluated on
the left and right triple $\mathbf{q}_\pm$ that parameterize distinct,
steady periodic orbits $\varphi_\pm$. 

For strictly hyperbolic, genuinely nonlinear Whitham modulation
equations with negative linear dispersion
($\partial_k^2 \omega_{0} < 0$), classical DSW solutions connecting
the two constant states $u_\pm$ are described by a rarefaction
solution of \eqref{eq:17} in which $\lambda = \lambda_2$ is the middle
characteristic speed and $\mathbf{q}_- = [u_-,0,k_-]$,
$\mathbf{q}_+ = [u_+,a_+,0]$.  The two constraints that result from
integrating \eqref{eq:17} determine the trailing edge wavenumber $k_-$
and speed $\xi_-$ as well as the leading edge amplitude $a_+$ and
speed $\xi_+$ \cite{Mark2016}.  Therefore, a classical DSW corresponds
to a rarefaction wave solution of the modulation equations,
\textit{not} a shock solution.  For DSW construction, we will use the
DSW fitting method, which leverages certain structural properties of
the Whitham modulation equations under the assumptions of strict
hyperbolicity and genuine nonlinearity in order to obtain $k_-$,
$a_+$, and $\xi_\pm$ by integrating a scalar ODE
\cite{El2005,Mark2016}.

Whitham himself pondered the notion of discontinuous shock solutions
to his eponymous equations \cite{Whitham74}. But their utility was
only recently discovered in \cite{Sprenger_2020} where shock solutions
of the Whitham modulation equations for a fifth order KdV (KdV5)
equation were deemed admissible if there exists a heteroclinic
traveling wave solution connecting the corresponding left and right
periodic orbits, each moving with the same speed as the shock.  Such
traveling wave solutions are possible in higher order equations such
as KdV5.  These \textit{Whitham shocks} were used to solve the Riemann
problem for KdV5 and, later, were investigated in the Kawahara
equation \cite{sprenger_traveling_2023}.  In this paper, we will show
that similar Whitham shocks emerge as the traveling wave portion of
the TDSW solution in Fig.~\ref{fig:example_sols}.  We also provide
analytical and numerical evidence of the existence of a new class of
Whitham shocks, i.e., shock solutions of the Whitham modulation
equations \eqref{eq:cont_mod} whose corresponding left and right
periodic orbits possess the same frequency but different speeds than
one another and the shock itself (see US in
Fig.~\ref{fig:example_sols}).

\subsection{Weakly nonlinear regime}

In the previous sections, we derived the modulation equations supposing
the existence of a family of nonlinear periodic solutions. In the
case where no known explicit periodic solution is available,
it is useful to approximate the periodic solution with a truncated
cosine series. The approximation via the Poincar\'{e}-Lindstedt method
utilizes an asymptotic expansion of both the profile of the periodic
solution and its frequency in the small amplitude parameter
$0 < a \ll 1$. The approximate periodic solution and its frequency are
given, for a generic potential $\Phi$ by
\begin{align} 
u &\sim \bar{u} + \frac{a}{2} \cos(kn - \omega t) +  a^2 \frac{\sin (2 k) \Phi^{(3)}(\bar{u})}{16 \Phi''(\bar{u})\left(2 \sin (k) -\sin (2 k) \right)} \cos(2(kn - \omega t)) + o(a^2),\label{eq:stokes_u}
\\
\omega &\sim \Phi''(\bar{u}) \sin(k) + a^2 \omega_2 + o(a^2), \quad \omega_2 = \frac{1}{32} \sin (k) \left(\frac{\Phi^{(3)}(\bar{u})^2}{(\sec (k)-1) \Phi''(\bar{u})}+\Phi^{(4)}(\bar{u})\right) ,\label{eq:stokes_omega}
\end{align}
which maintain their asymptotic ordering so long as
\begin{equation}
  \label{eq:10}
  a^2/|k| \ll 1 \quad \mathrm{and} \quad
  |a \Phi^{(3)}(\bar{u})/\Phi''(\bar{u})| \ll 1, 
\end{equation}
i.e., for $\Phi(u) = u^3/3$, neither $|k|$ nor $|\bar{u}|$ are too
small.  Inserting \eqref{eq:stokes_u}, \eqref{eq:stokes_omega} into
the modulation equations \eqref{eq:cont_mod}, we obtain the weakly
nonlinear Whitham modulation equations in conservative form by
retaining terms up to $O(a^2)$
\begin{align}
\bar{u}_T + \left(\Phi'(\bar{u}) + \frac{a^2}{16}\Phi'''(\bar{u})\right)_X & = 0, \\
\left(\Phi(\bar{u}) + \frac{a^2}{16}\Phi''(\bar{u})\right)_T 
+ \left(\frac{1}{2} \Phi'(\bar{u})^2 + \frac{a^2}{16} \left(\Phi''(\bar{u})^2 \cos(k) + \Phi'(\bar{u})\Phi'''(\bar{u})\right)\right)_X & = 0, \\
k_T + \left(\Phi''(\bar{u}) \sin(k) + \frac{a^2}{32} \sin (k)
  \left(\frac{\Phi^{(3)}(\bar{u})^2}{(\sec (k)-1)
  \Phi''(\bar{u})}+\Phi^{(4)}(\bar{u})\right)\right)_X & = 0 .
\end{align}
In the case of the cubic potential \eqref{pot}, our focus here, the
modulation equations are
\begin{subequations}
  \label{eq:weakly_nl_mod}
  \begin{align}
    \label{eq:a_weakly_nl_mod}
    \bar{u}_T + \left(\bar{u}^2 + \frac{a^2}{8}\right)_X & = 0, \\
    \label{eq:b_weakly_nl_mod}
    \left(\frac{\bar{u}^3}{3} + \frac{a^2}{8}\bar{u}\right)_T 
    + \left(\frac{1}{2} \bar{u}^4+ \frac{a^2}{8} \bar{u}^2 \left(2
    \cos(k) + 1\right)\right)_X & = 0, \\ 
    \label{eq:c_weakly_nl_mod}
    k_T + \left(2\bar{u}\sin(k) + \frac{a^2}{16}
    \left(\frac{\sin(k)}{(\sec (k)-1) \bar{u}}\right)\right)_X & = 0. 
\end{align}
\end{subequations}
Properties of the modulation equations can be elucidated by casting
them in quasi-linear form
$\tilde{\mathbf{q}}_t + \tilde{\mathcal{A}} \tilde{\mathbf{q}}_x = 0$,
where
\begin{align}
  \tilde{\mathcal{A}} =
  \begin{bmatrix}
    2 \bar{u} & \frac{1}{8} & 0 \\
    4 a^2 \cos (k) & 2 \bar{u} \cos (k) 
                            & -2 a^2 \bar{u} \sin (k) \\
    2 \sin (k) + a^2 \omega_{2,\bar{u}} & \omega_2
                            & 2\bar{u} \cos(k) + a^2
                              \omega_{2,k} 
  \end{bmatrix},
                              \qquad \tilde{\mathbf{q}} =
                              \begin{bmatrix}
                                \bar{u} \\ a^2 \\ k
                              \end{bmatrix}.
\end{align}
A perturbation calculation gives the
eigenvalues of the flux matrix $\mathcal{A}$ to $\mathcal{O}(a)$
\begin{subequations}
  \label{eq:48}
  \begin{align}
    \lambda_3 & = 2 \bar{u} +O(a^2), \\ 
    \lambda_2 & = 2\bar{u} \cos(k) +
                \frac{a}{2}\cos\left(\frac{k}{2}\right)\sqrt{2-\cos(k)} +O(a^2),
    \\ 
    \lambda_1 & = 2\bar{u} \cos(k) -
                \frac{a}{2}\cos\left(\frac{k}{2}\right)\sqrt{2-\cos(k)}  +O(a^2)
                , 
  \end{align}
\end{subequations}
with the corresponding right eigenvectors
\begin{subequations}
  \begin{align}
    \mathbf{r}_3&= \left [
         \bar{u}\tan\left(\frac{k}{2}\right), 0, 1 \right ]^T
         +O(a^2),\\
    \mathbf{r}_2&=
         \left [0, 0, \sqrt{2-\cos (k)} \right ]^T
         + \frac{a}{4} \left [ \csc
                  \left(\frac{k}{2}\right), - 32 \bar{u} \sin
         \left(\frac{k}{2}\right), 0  \right ]^T
         + O(a^2), \label{eq:r2} \\
    \mathbf{r}_1&= [0, 0, \sqrt{2-\cos (k)} ]^T - \frac{a}{4}
         \left [ \csc \left(\frac{k}{2}\right), -32 \bar{u} \sin
         \left(\frac{k}{2}\right),
         0  \right ]^T
         + O(a^2). 
  \end{align}
\end{subequations}
    
The quasilinear system is strictly hyperbolic if all of the
eigenvalues are distinct, and real valued. To the order of the
approximation given, the weakly nonlinear system is strictly
hyperbolic provided $a \ne 0$, $k \ne \pi$, and
\begin{equation}
  \label{eq:54}
  |\overline{u}| \ne \overline{u}_{\rm cr}, \quad u_{\rm cr} =
  \frac{a\cos \left ( \frac{k}{2} \right ) 
    \sqrt{2-\cos(k)}}{8 \sin^2\left ( \frac{k}{2} \right )} .
\end{equation}
When $\overline{u} > u_{\rm cr}$ and $a > 0$, the eigenvalues are
ordered $\lambda_1 < \lambda_2 < \lambda_3$.  When $a = 0$,
Eqs.~\eqref{eq:a_weakly_nl_mod} and \eqref{eq:b_weakly_nl_mod}
coincide with the Hopf equation for the mean $\bar{u}$ and the
remaining equation corresponds to the conservation of waves from
linear wave modulation theory.  When $k = \pi$,
Eq.~\eqref{eq:c_weakly_nl_mod} is identically satisfied.  While
intuition might suggest that \eqref{eq:a_weakly_nl_mod} and
\eqref{eq:b_weakly_nl_mod} are somehow related to the modulation
equations for binary oscillations \eqref{eq:9}, in fact, the
asymptotic derivation breaks down.  For example, the period average of
the weakly nonlinear solution \eqref{eq:stokes_u} is no longer
$\bar{u}$ but rather $\bar{u} - \frac{a^2}{8\bar{u}}$ so that the
density in \eqref{eq:a_weakly_nl_mod} does not correspond to the
density in \eqref{eq:a_disc_mod}.  When $k = \pi$, one should discard
Eq.~\eqref{eq:weakly_nl_mod} altogether in favor of the modulation
equations for binary oscillations \eqref{eq:9}, which apply beyond the
weakly nonlinear regime considered here.

\section{Classification of Solutions} \label{sec:class}

From now onwards, we focus solely on the discrete equation
\eqref{eq:3} (Eq.~\eqref{deq} subject to \eqref{pot}).  Figure
\ref{fig:example_sols} depicts seven qualitatively distinct solution
families to the Riemann problem \eqref{eq:3}, \eqref{step} depending
upon the parameters $u_\pm$ in the initial data.  We now proceed to
describe each of these solution families using a combination of
numerical simulation, Whitham modulation theory, and quasi-continuum
approximation.  The straight line boundaries between each solution
family in Fig.~\ref{fig:example_sols} are determined empirically (to
two decimal digits accuracy) and some are explained by analytical
considerations. By a possible reflection of the lattice $n \to -n$ and
a rescaling of time, we can, without loss of generality, set either
$u_+ = 1$ while varying $u_- \in [-1,1]$ or set $u_- = 1$ while
varying $u_+ \in [-1,1]$.  Therefore, we can map out the phase diagram
in the $(u_-,u_+)$  plane by traversing the top and right edges of the
square $[-1,1]^2$.

The special case in which $u_+ = u_-$ is trivial but the case in which
$u_+ = -u_- \ne 0$ is the stationary shock (SS) solution
\eqref{eq:35}.  Otherwise, the solutions exhibit more complexity,
which we now explore. We start with the simplest case first, and then
work toward the richest, most complex scenario.

\subsection{Rarefaction waves (RWs)}
\label{sec:RW}

The simplest observed dynamical structure is the rarefaction wave
shown as RW in Fig.~\ref{fig:example_sols}.  Empirically, we find that
they form when $u_+= 1$ and $u_-\in(0.18,1)$.  The bifurcation at
$u_- = 0.18$ will be described in section~\ref{sec:unsteady}. The leading
order RW behavior is given by the self similar solution
($\xi = n/t = X/T$)
\begin{align}
  \label{eq:27}
  u_n(t) \sim \bar{u}(\xi) =
  \begin{cases}
    u_- & \xi \leq 2u_- \\ 
    \xi/2 & 2u_-  < \xi \leq 2u_+  \\ 
    u_+ & 2u_+ < \xi
  \end{cases}, 
\end{align}
of the dispersionless equation \eqref{eq:20}.  A favorable comparison
of this profile with a numerical simulation is shown in
Fig.~\ref{fig:RW}. Because the data is expansive, the effect of
dispersion manifests at higher order where a small amplitude,
dispersive wavetrain is emitted from the lower, left edge of the RW.
\begin{figure}
    \centering
    \includegraphics{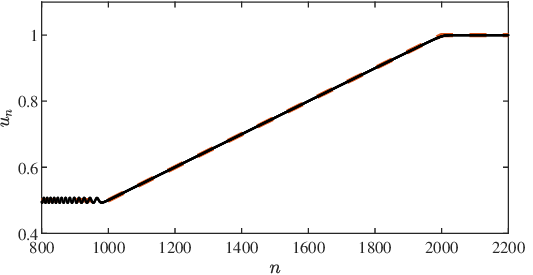}
    \caption{Comparison of the self-similar solution \eqref{eq:27}
      (red dashed) with numerical simulation of the initial data
      \eqref{step} with $u_- = 0.5$ and $u_+ = 1$ (black dots).}
    \label{fig:RW}
\end{figure}
The slowest (most negative) group velocity is
$\partial_k\omega_0(\pi,u_-) = -2 u_-$, which corresponds to an inflection
point of the linear dispersion relation \eqref{eq:18}. 
Consequently,
the leftmost edge of these small amplitude waves is expected to have
an Airy profile whose decay estimate is proportional to $t^{-1/3}$,
similar to the Fourier analysis carried out for linear FPUT chains
\cite{MielkePatz2017}. The details of the linear wavetrain
accompanying RWs and DSWs for the BBM equation were studied in
\cite{congy_dispersive_2021}.  We follow a similar procedure by
linearizing about the left initial state $u_n = u_- + v_n$ to obtain 
\begin{align}
  \label{eq:linearized_eq}
  \frac{\rm d}{\rm{d} t}v_n + u_-(v_{n+1} - v_{n-1}) = 0. 
\end{align}
The initial data \eqref{step} then becomes
\begin{equation}
  \label{eq:36}
  f_n =
  \begin{cases}
    0 & n < 0, \\ u_+ - u_- & n \geq 0 ,
  \end{cases}
\end{equation}
whose discrete-space Fourier transform is the distribution
\begin{align}
  \label{eq:37}
  \hat{f}(k) = \sum_{n=-\infty}^\infty f_n
  e^{-ink} = (u_+ - u_-) \left (\frac{1}{1 - e^{-ik}} +
  \pi \delta(k) \right ), \quad k \in (-\pi,\pi],
\end{align}
where $\delta(k)$ is the Dirac delta.  To approximate the nonlinear
equation \eqref{eq:3} by the linear equation \eqref{eq:linearized_eq},
one could seek solutions in which $0 < u_+ - u_- \ll |u_+|+|u_-|$.
Alternatively, we follow \cite{congy_dispersive_2021} and consider
scale separation in which the highest frequency components of
\eqref{eq:37} are assumed to separate from the RW so that the initial
data becomes
\begin{equation}
  \label{eq:38}
  \hat{v}(k,0) =
  \begin{cases}
    \frac{u_+ - u_-}{1 - e^{-ik}} & k_0 < |k| \le \pi, \\
    0 & \mathrm{else} ,
  \end{cases}
\end{equation}
for some $0< k_0 < \pi$ that is sufficiently far from the zero dispersion points, $k = 0,\pi$.  Then, the solution of the linear equation
\eqref{eq:linearized_eq} can be determined by taking the discrete-space
Fourier transform $\hat{v}(k,t) = \sum_n v_n(t) e^{-ink}$.  The
solution of Eq.~\eqref{eq:linearized_eq} subject to \eqref{eq:38} is
\begin{align}
  \label{eq:stat_phase}
  v_n(t) = \frac{1}{2\pi} \int_{-\pi}^\pi \hat{v}(k,0) e^{i\theta(k)t}
  \ {\rm d}k \qquad \theta(k;n,t,u_-) = k n/t- 2u_-\sin(k) .
\end{align}
Quantitative information regarding the solution can be determined
asymptotically for $t \to \infty$ with $n/t$ fixed using the method of
stationary phase \cite{Whitham74}. The leading order behavior is
determined by analyzing the integral \eqref{eq:stat_phase} near the
stationary points, $k_s$ where $\theta_k(k_s) = 0$. Stationary points
are therefore given by $\pm k_s$ where
\begin{align}
    n/t = 2 u_- \cos(k_s) ,
\end{align}
for $-2u_- < n/t < 2u_-$. The leading order behavior in the vicinity
of the stationary points is determined by expanding the integrand in
\eqref{eq:stat_phase} about the stationary points $k = \pm k_s$. When
$k_s \ne \pi$ and $|k_s| > k_0$, the leading order behavior is
\begin{equation}\label{eq:nondeg_stat_phase}
  \begin{split}
  v_n(t) &\sim \frac{1}{\sqrt{2\pi t |\partial_{k}^2 \omega_0(k_s,u_-)|}}
           \left(\hat{v}(k_s,0) e^{i k_s n - i \omega_0(k_s,u_-) t + i
           \pi/4} + {\rm c.c.} \right), \\
         & = \frac{1}{\sqrt{2\pi t |\partial_{k}^2 \omega_0(k_s,u_-)|}} (u_+ - u_-)
           \csc\left(\tfrac{k_s}{2}\right)\sin\left(\theta(k_s)t +
           \pi/4 + \tfrac{k_s}{2}\right) .
  \end{split}
\end{equation}
The profile \eqref{eq:nondeg_stat_phase} is compared with numerical
simulations of the initial value problem in
Fig.~\ref{fig:linear_wave_comparison_RW} on the interval
$[-u_- t_f, u_-t_f]$ at a final simulation time of $t = t_f =
1000$. The interval is chosen so that $k_s \in (\pi/3,2\pi/3)$, i.e.,
the truncation parameter $k_0 = \pi/3$ and we avoid the degenerate
stationary points $k_s = 0, \pi$.  We observe that the linear profile
\eqref{eq:nondeg_stat_phase} is in good agreement with the numerical
simulation. However, for larger initial jumps, the linear wave begins
to deviate from the simulation.  This may be attributed to the
emergence of stronger nonlinear effects not captured by the leading
order asymptotics which require a larger truncation parameter $k_0$.

\begin{figure}
    \centering
    \includegraphics{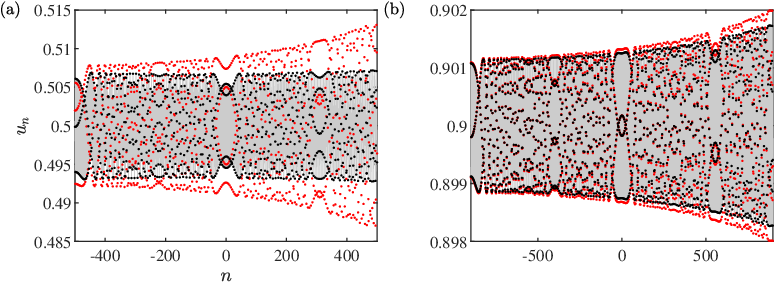}
    \caption{Comparison of the linear wave from the stationary phase
      analysis \eqref{eq:nondeg_stat_phase} (red dots) and numerical
      simulation (black dots) resulting in a rarefaction wave for (a)
      $u_- = 0.5$ and (b) $u_- = 0.9$.}
    \label{fig:linear_wave_comparison_RW}
\end{figure}
To investigate the leftmost edge of the linear wave emitted from the
RW, we modify our previous analysis and expand the phase in the
integral \eqref{eq:stat_phase} about the inflection point
$k = \pi$ of the linear dispersion relation
\begin{align}\label{eq:phase_pi}
\theta(k) \sim \pi n/t - (n/t + 2u_-)(\pi-k) + u_-(\pi-k)^3/3 +
  \cdots
\end{align}
The expansion \eqref{eq:phase_pi} is inserted into the
integral \eqref{eq:stat_phase}. A calculation reveals that the leading
order asymptotics in the vicinity of the ray $n/t = -2u_-$ are given
by
\begin{align}\label{eq:airy}
  v_n(t) \sim - \frac{u_+-u_-}{2(t u_-)^{1/3}}\cos(\pi n)
  {\rm Ai}\left( - (n + 2 t u_-) (t u_-)^{-1/3}\right),
\end{align}
where ${\rm Ai}(\cdot)$ is the Airy function
${\rm Ai}(z) = \frac{1}{2\pi} \int_\mathbb{R} e^{i \kappa z + i
  \kappa^3/3} \ {\rm d} \kappa .$ The Airy profile \eqref{eq:airy}
favorably compares with the two Riemann problem simulations depicted
in Fig.~\ref{fig:enter-label}, even for large $u_+ - u_-$.

\begin{figure}
    \centering
    \includegraphics{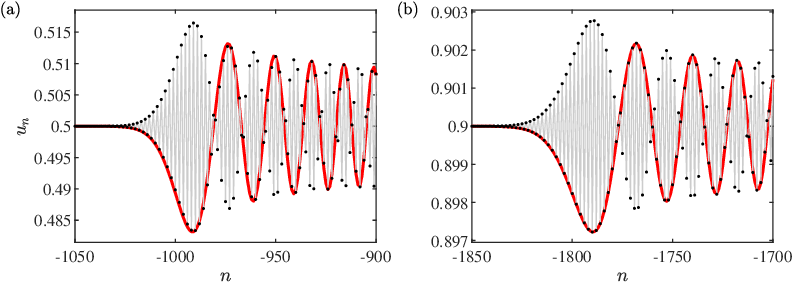}
    \caption{Comparison of the Airy profile \eqref{eq:airy} (red
      curve) with numerical simulations (black dots) resulting in a
      rarefaction wave with initial data $u_+ = 1$ and (a) $u_- = 0.5$
      and (b) $u_- = 0.9$. }
    \label{fig:enter-label}
\end{figure}

It is worth contrasting the observed RW dynamics with those of the
quasi-continuum approximation in the BBM equation \eqref{eq:BBM} that
was studied in \cite{congy_dispersive_2021}.  Qualitatively, the
dynamics exhibited by the two models in overlapping regimes of the
$(u_-,u_+)$ plane of Riemann data are very similar.  Both equations
exhibit large-scale dynamics that are well-approximated by the
self-similar solution \eqref{eq:27} and its analogue for the BBM
equation.  The details of the short-scale, emitted dispersive
wavetrains are quantitatively different but, since both equations
admit non-convex linear dispersion relations, they both exhibit Airy
profiles with amplitude decay proportional to $t^{-1/3}$.

The long-time dynamics produced by the lattice model \eqref{eq:3}
significantly differs from those generated by its quasi-continuum BBM
counterpart \eqref{eq:BBM} when either equation is strongly influenced by
small scale effects.  The actual Riemann problems for the BBM equation
studied in \cite{congy_dispersive_2021} were $\mathrm{tanh}$-smoothed,
monotone transitions between $u_-$ and $u_+$, 
a feature which introduces an
external length scale characterizing the width of the initial
transition.  When this width is larger than the $\mathcal{O}(1)$
oscillatory length scale (or $\mathcal{O}(\epsilon)$ in
Eq.~\eqref{eq:BBM}), the BBM equation exhibits a RW for all
$|u_-| < u_+$.  As shown in Fig.~\ref{fig:example_sols}, RW generation
on the lattice is limited to the region $0.18 u_- < u_+ < u_-$,
$u_- > 0$, with short-scale oscillatory dynamics occurring when
$-u_- < u_+ < 0.18 u_-$, $u_- > 0$.  When the BBM initial transition
width is sufficiently small, RW generation can be accompanied by the
spontaneous generation of solitary waves and/or an expansion shock,
features not observed in the lattice model.

\subsection{Dispersive shock waves (DSWs)}

For $u_- = 1$ and 
$u_+\in(0,1)$
the data is compressive and results
in an expanding, modulated oscillatory wave train between the states
$u_-$ and $u_+$. This structure is called a dispersive shock wave; see
the panel labeled DSW in Fig.~\ref{fig:example_sols}. In
\cite{CHONG2022133533}, DSWs were studied in Eq.~\eqref{eq:disc_mod}
with $\Phi(u) = u^2/2 + u^4/4$ and $u_- = 1$, $u_+ =0$ using numerical
simulations and a dimension reduction approach. In the following, we
study DSWs in the system \eqref{eq:3} as the step value $u_+$ varies
using two semi-analytical approaches, DSW fitting and a continuum
model.  

\subsubsection{Approximation of the DSW harmonic and soliton edge
  speeds via DSW fitting}

In this section, we outline the method for fitting the macroscopic DSW
properties (edge speeds, amplitudes, and wavenumbers) by applying the
fitting method first introduced by El \cite{El2005}; see also
\cite{Mark2016}.  This method was originally developed and justified for continuum PDEs where it has been extensively applied.  Since it only requires knowledge of the linear dispersion relation and the solitary wave amplitude-speed relation, it is straightforward to extend the method to the semi-discrete lattice equation \eqref{deq}.

An example DSW from a numerical simulation of the
Riemann problem with $u_+ = 0.5$ and $u_-=1$ is shown in
Fig.~\ref{fig:DSW_speeds}.  The DSW is comprised of a modulated,
nonlinear wavetrain that terminates in two distinct limits: vanishing
amplitude (called the harmonic edge) and vanishing wavenumber (called
the soliton edge).  The modulation solution describing the DSW is the
rarefaction solution \eqref{eq:17} of the Whitham modulation equations
\eqref{eq:cont_mod} with $\mathbf{q}_- = [u_-,0,k_-]^T$,
$\mathbf{q}_+ = [u_+,a_+,0]^T$.  The harmonic edge wavenumber $k_-$
and the soliton edge amplitude $a_+$, as well as their corresponding
edge speeds $\xi_-$ and $\xi_+$ are determined by integrating the ODE
\eqref{eq:17}, thus relating these macroscopic DSW properties to the
initial data $u_\pm$. In numerical simulations, the amplitude of the
DSW does not vanish exactly at the harmonic edge, so we define the
location of the trailing edge by the intersection of the oscillatory
envelope curves shown in Fig.~\ref{fig:DSW_speeds}.

\begin{figure}[H]
\centering
    \includegraphics{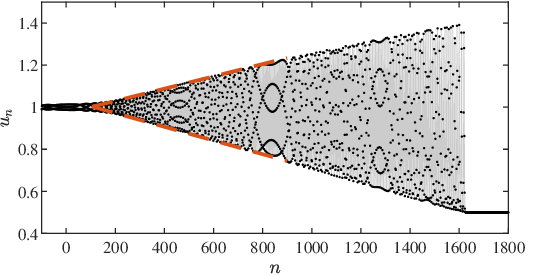}
    \caption{Discrete DSW for $u_- = 1$, $u_+ = 0.5$ at $t = 1000$.  The dashed red
      oscillatory envelope curves intersect at the harmonic edge.
      }
    \label{fig:DSW_speeds}
\end{figure}
Although the DSW modulation is determined, in principle, by the
aforementioned rarefaction solution of the Whitham modulation
equations (assuming strict hyperbolicity and genuine nonlinearity of
the second characteristic field), we do not have explicit 
expressions relating
the integrals in \eqref{eq:cont_mod} to the periodic orbit's
parameters $\mathbf{q} = [\bar{u},a,k]^T$.  An alternative technique
that allows one to obtain the DSW's edge properties is the DSW fitting
method \cite{Mark2016}.  This method assumes the existence of the
rarefaction solution. For the sake of completeness, we will carry out
the DSW fitting procedure for a generic potential $\Phi(u)$. The
wavenumber of the DSW at the harmonic edge can be determined by solving
the initial value problem
\begin{equation}
  \label{eq:24}
  \frac{\rmd k}{\rmd \bar{u}} = \frac{\partial_{\bar{u}}
    \omega_0(k,\bar{u})}{\Phi''(\bar{u}) -
    \partial_k\omega_0(k,\bar{u})} = \frac{\Phi'''(\bar{u})\sin
    k}{\Phi''(\bar{u})(1-\cos k)}, \quad k(u_+) = 0 ,
\end{equation}
where $\omega_0$ is the linear dispersion relation \eqref{eq:18}.
This ODE can be integrated by separation of variables to obtain
\begin{align}
  k(\bar{u}) = \cos^{-1} \left (\frac{ 2
  \Phi''(u_+)-\Phi''(\bar{u})}{\Phi''(\bar{u})} \right ).
\end{align}
The wavenumber at the DSW harmonic edge is $k_- = k(u_-)$. The
velocity of the harmonic edge is given by evaluating the linear group
velocity at $k_-$
\begin{align}
  \label{eq:32}
  c_- = \partial_k\omega_0(k_-,u_-) = 2\Phi''(u_+) - \Phi''(u_-).
\end{align}

The velocity of the DSW at the soliton edge is calculated in a similar
way. We begin by introducing the conjugate variables
$\Tilde{\omega}_0\left(\Tilde{k},\bar{u}\right) = - i \omega_0(i
\Tilde{k},\bar{u}) = \Phi''(\bar{u}) \sinh(\Tilde{k})$, where
$\Tilde{k}$ acts as a soliton amplitude parameter. The velocity of the
DSW soliton edge is deduced by evaluating the solitary wave dispersion
relation $c_+ = \Tilde{\omega}(\Tilde{k}_+,u_+)/\Tilde{k}_+$, where we
find $\tilde{k}_+$ by solving the initial value problem
\begin{align}
  \frac{d \Tilde{k}}{d \bar{u}} =
  \frac{\partial_{\bar{u}} \tilde{\omega}_0}{\Phi''(\bar{u}) -
  \partial_{\Tilde{k}} \Tilde{\omega}_0} = \frac{\Phi'''(\bar{u})
  \sinh \Tilde{k}}{\Phi''(\bar{u})(1-\cosh \Tilde{k})} ,  \qquad
  \Tilde{k}(u_-) = 0.  
\end{align}
Integration results in
\begin{equation}
  \label{eq:25}
  \tilde{k}(\bar{u}) = \cosh^{-1} \left ( \frac{2 \Phi''(u_-) -
      \Phi''(\bar{u})}{\Phi''(\bar{u})} \right ) ,
\end{equation}
and the soliton edge conjugate wavenumber
$\tilde{k}_+ = \tilde{k}(u_+)$.  The soliton edge velocity of the
discrete DSW is given by
\begin{align} \label{eq:cplus}
  c_+ =  \frac{\tilde{\omega}(\tilde{k}_+,u_+)}{\tilde{k}_+} =
  \frac{2}{\tilde{k}_+}\sqrt{\Phi ''({u}_-)(\Phi''(u_-)-\Phi''(u_+))} 
\end{align}
Comparisons with numerical simulations of the Riemann problem are
given in Fig. \ref{fig:DSW_fitting_quadratic} with the potential
$\Phi(u) = u^3/3$ and the initial data normalized so that $u_- =
1$. 
To estimate the velocity of the leading edge, we track the position of
the right-most lattice site that is above the far-field initial data
$u_+$ at the integer valued times in our numerical simulation. This
time series data is fit with a line whose slope is the approximate
velocity of the soliton edge of the DSW. To estimate the harmonic edge
velocity, we produce a linear fit of the modulated wavetrain amplitude
near the location of the harmonic edge, which is found by extracting
peaks of the solution at output times. The intersection of this linear
approximation with the constant level $u_-$ is the approximate
location of the DSW harmonic edge. The time interval of our numerical
computations varies depending on $u_+$. For instance, when we take
$u_+ = 0.1$, we approximate the solution at $t \approx 1000$, while
taking $u_+ = 0.9$ requires the longer time $t \approx 4000$ for the
solution to asymptotically develop. Upon varying $u_+$, we observe
good agreement between the predictions of DSW fitting for the DSW
harmonic and soliton edge velocities and numerical simulations for
$u_+ \gtrsim 0.5$, while the predictions begin to deviate from what is
observed in numerical simulations below this threshold.  The DSW
fitting method is subject to the convexity conditions
\begin{equation}
  \label{eq:26}
  \frac{\partial c_\pm}{\partial u_-} \ne 0, \quad \frac{\partial
    c_\pm}{\partial u_+} \ne 0 .
\end{equation}
A direct calculation for the potential $\Phi(\bar{u}) = \bar{u}^3/3$
demonstrates that these convexity conditions are indeed satisfied.
The numerical results suggest that the DSW fitting method provides an
adequate approximate prediction of the discrete DSW edge properties.
\begin{figure}[H]
    \centering
    \includegraphics{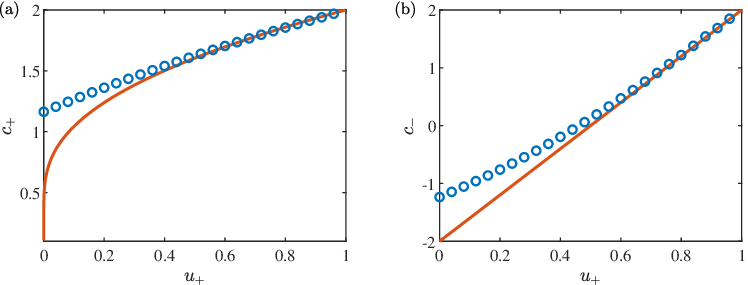}
    \caption{Comparison of DSW fitting predictions (solid, red line)
      and numerical simulations of DSWs with $u_- = 1$ (blue
      circles). The leading edge, solitary wave velocity $c_+$ is
      shown in panel (a) and the trailing, linear wave edge velocity
      comparisons are shown in panel (b).}
    \label{fig:DSW_fitting_quadratic}
\end{figure}

The harmonic edge of the DSW is accompanied by linear radiation much
like the left edge of the RW in the previous section. To describe
this, we apply the same approach as in Sec. \ref{sec:RW} to
approximating the small amplitude linear waves that emanate from the
DSW's harmonic edge. The only change is that, for the DSW,
$u_+ - u_- < 0$ in Eq.~\eqref{eq:36}. A comparison is shown in
Fig. \ref{fig:DSW_radiation}.

\begin{figure}
    \centering
    \includegraphics{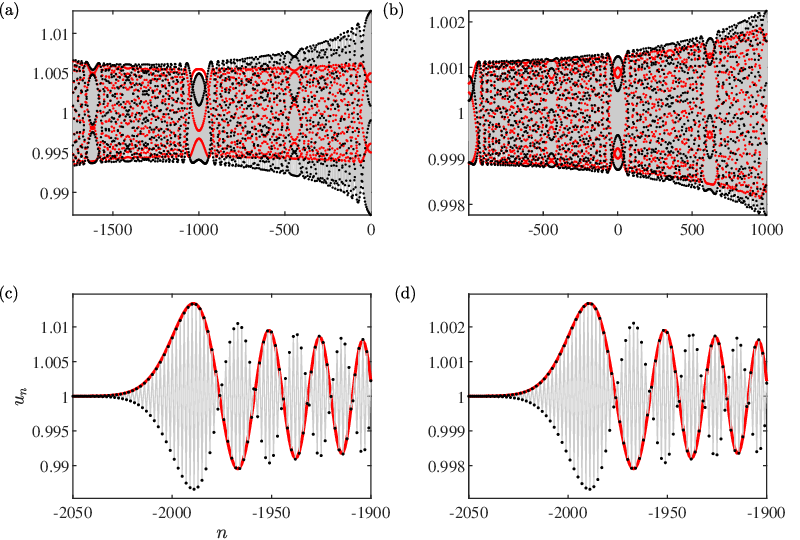}
    \caption{Comparison of stationary phase analysis
      \eqref{eq:nondeg_stat_phase} and the Airy profile
      \eqref{eq:airy} (both in red) with numerical simulations (black
      dots) resulting in a DSW at $t = 1000$ with initial data
      $u_- = 1$ and (a,c) $u_+ = 0.5$ and (b,d) $u_+ = 0.9$. The
      band-limited interval of the stationary points in the formula
      \eqref{eq:nondeg_stat_phase} are (a)
      $k_{\rm s} \in (2\pi/3,5\pi/6)$ and (b)
      $k_{rm s} \in (\pi/3,2\pi/3)$. These intervals are chosen to
      avoid the linear waves at the left edge of the DSW.}
    \label{fig:DSW_radiation}
\end{figure}

\subsubsection{Approximations of the leading edge amplitude via a
  quasi-continuum model} \label{sec:quasi}

In this subsection, we go a bit deeper into the description of the
solitary wave at the soliton edge of the DSW.  Like in the previous
subsection, for $u_+ \in (0,1)$ (where $u_-=1$ is fixed), we
numerically solve the Riemann problem, generate a DSW and extract the
amplitude of the soliton edge, and its speed. This is done by
inspecting the time series of a node sufficiently far from the center
of the lattice (we chose $n=300$, in which case we have observed the leading edge is developed) 
and
simply computing the amplitude as $a_+ = \max{u_n(t)} - u_+$. The
speed is estimated by computing $c_+ = 1/(t_{n+1} - t_n)$ where $t_n$
and $t_{n+1}$ are the time values where $u_n$ and $u_{n+1}$ attain
their maxima.  The blue open circles in
Fig.~\ref{fig:compare_quasi}(a) show the amplitude of the DSW soliton
edge as a function of $u_+$. Since for each value of $u_+$, we compute
both the speed $c_+$ and amplitude $a_+$, we also show a parametric
plot of $(c_+, a_+)$ parameterized by $u_+$ in
Fig.~\ref{fig:compare_quasi}(b).

To confirm that the DSW soliton edge is indeed described by a solitary
wave, we compute a numerical solitary wave solution of the lattice
equation \eqref{deq} by using a fixed-point iteration scheme
\cite{Herrmann_Scalar} to solve the advance delay equation
\begin{equation} \label{adv1}
  2 c \mathcal{V}^{\prime}(\eta)=\Phi^{\prime}(u_+ + \mathcal{V}(\eta+1))-\Phi^{\prime}(u_+ + \mathcal{V}(\eta-1)),
\end{equation}
which is a rescaled copy of \eqref{eqn:TW} and obtained by
substituting
\begin{equation}\label{tw}
u_{n}(t)=u_+ + \mathcal{V}(\eta), \qquad \eta = n-  c t
\end{equation}
into Eq.~\eqref{deq}.  While we are free to select values of $c$ and
$u_+$ when solving Eq.~\eqref{adv1}, we select combinations of them
according to the relationship extracted from the DSW soliton edge
(i.e., the blue circles in Fig.~\ref{fig:compare_quasi}(a)). Upon
convergence of the scheme, we compute the amplitude of the resulting
solitary wave, which is the maximum of the wave minus the background
$u_+$.  The amplitudes of the solitary waves are shown as the solid
red dots in Fig.~\ref{fig:compare_quasi}(a,b). Note that these red
dots fall almost exactly within the blue circles, indicating that the
soliton edge of the DSW is indeed well-approximated by a solitary wave
solution of the lattice equation.

We can obtain an analytical approximation of discrete solitary waves
by considering the BBM quasi-continuum approximation \eqref{eq:BBM}
of the lattice dynamics \eqref{eq:3}.
Entering the moving frame $U(X,T) = \phi(X - c T)$ and integrating
once, this PDE becomes the second order ODE
\begin{equation}
  \label{eq:39}
  c \frac{ \epsilon^2}{6} \phi'' = B +  c \phi -  2\phi^2
\end{equation}
where $B$ is an arbitrary integration constant. This ODE can be solved
using quadrature \cite{Kamchatnov}.  In particular, the solitary wave
with tails decaying to the background state $u_+$ has the form
\begin{equation} \label{qc} u_n(t) = \phi(X - c T) = u_+ +
  \left(\frac{3c}{2} - 3 u_+ \right){\sech}^2\left(\sqrt{\frac{
        \frac{3}{2}c- 3 u_+}{c}} \, (n - c t)\right)
\end{equation} 
which assumes $c > 2u_+$. Note the maximum speed of linear waves on a
background $u_+$ is $2u_+$, implying the solitary waves travel faster
than all linear waves, as expected.  In Eq.~\eqref{qc}, $c$ and $u_+$
can be chosen independently of each other but we once again select
combinations of them according to the relationship extracting from the
DSW soliton edge (the blue circles in
Fig.~\ref{fig:compare_quasi}(a)). The solid blue dots of
Figure~\ref{fig:compare_quasi}(a) show the quasi-continuum prediction
of the amplitude $a_+ = \frac{3c_+}{2} - 3 u_+$, where it is seen that
the approximation becomes better as the jump height decreases (i.e.,
as $u_+ \rightarrow 1$).  The quasi-continuum prediction of the
amplitude is only semi-analytical as it relies on the numerically
obtained relationship of $u_+$ and $c_+$ from the DSW soliton edge
data. An analytical prediction can be derived by using the DSW soliton
edge speed in Eq.~\eqref{eq:cplus} of the previous subsection, which
allows us to express the amplitude, $a_+$, of the DSW soliton edge in
terms of just $u_+$ (or $c_+$):
\begin{equation} \label{amplitude_uplus}
  a_+ = \frac{6\sqrt{1-u_+}}{\cosh^{-1}\left(\frac{2-u_+}{u_+} \right)} - 3 u_+
\end{equation}
See the solid red line of Fig.~\ref{fig:compare_quasi}(a,b) for a plot of
this formula.

Finally, the solitary wave profile given by Eq.~\eqref{qc} matches the
numerically computed solitary wave solution of Eq.~\eqref{adv1} quite
well, especially for longer wavelength solutions.  See
Fig.~\ref{fig:compare_quasi}(c) for a comparison of the actual
solitary wave (solid red dots) and quasi-continuum approximation
(solid red line) for two example parameter sets.

 \begin{figure}[H]
\centerline{
  \begin{tabular}{@{}p{0.33\linewidth}@{}p{0.33\linewidth}@{}@{}p{0.33\linewidth}@{}  }
  \rlap{\hspace*{5pt}\raisebox{\dimexpr\ht1-.1\baselineskip}{\bf (a)}}
 \includegraphics[height=4.5cm]{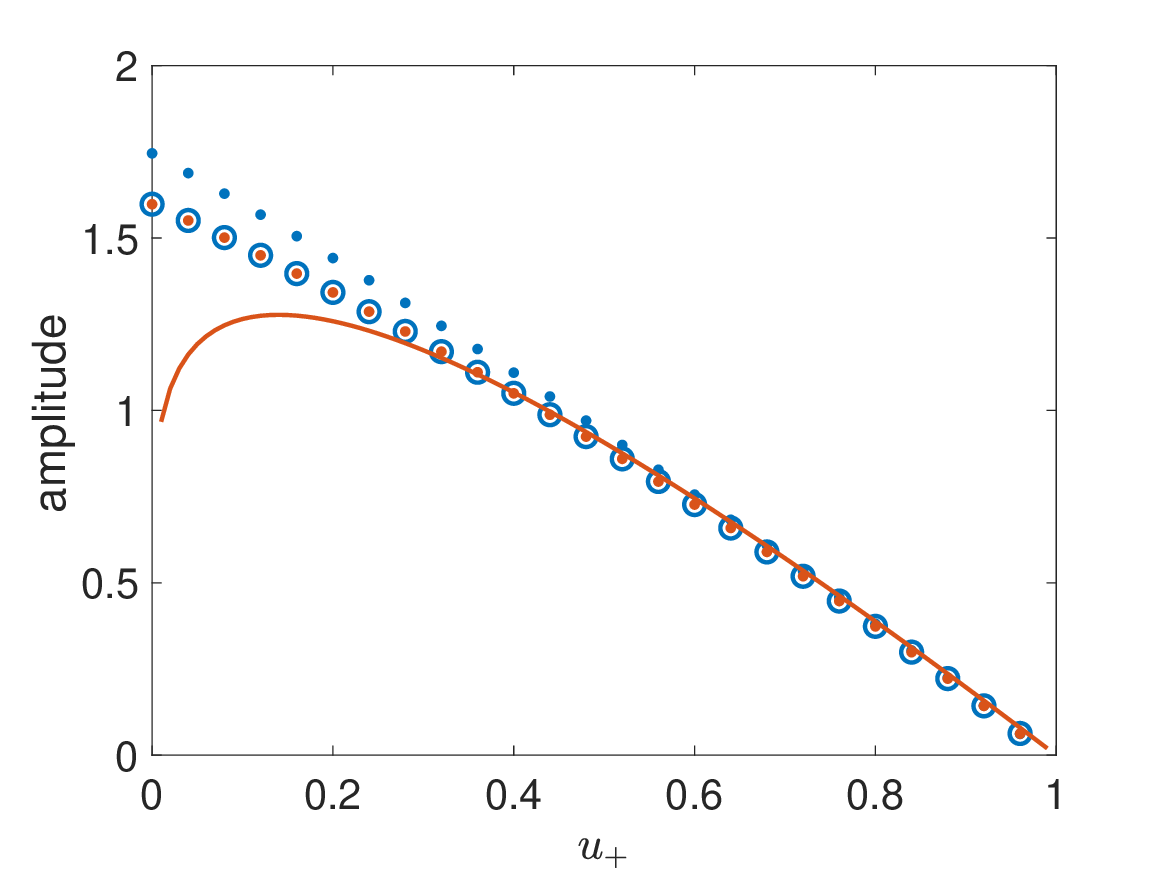} &
  \rlap{\hspace*{5pt}\raisebox{\dimexpr\ht1-.1\baselineskip}{\bf (b)}}
\includegraphics[height=4.5cm]{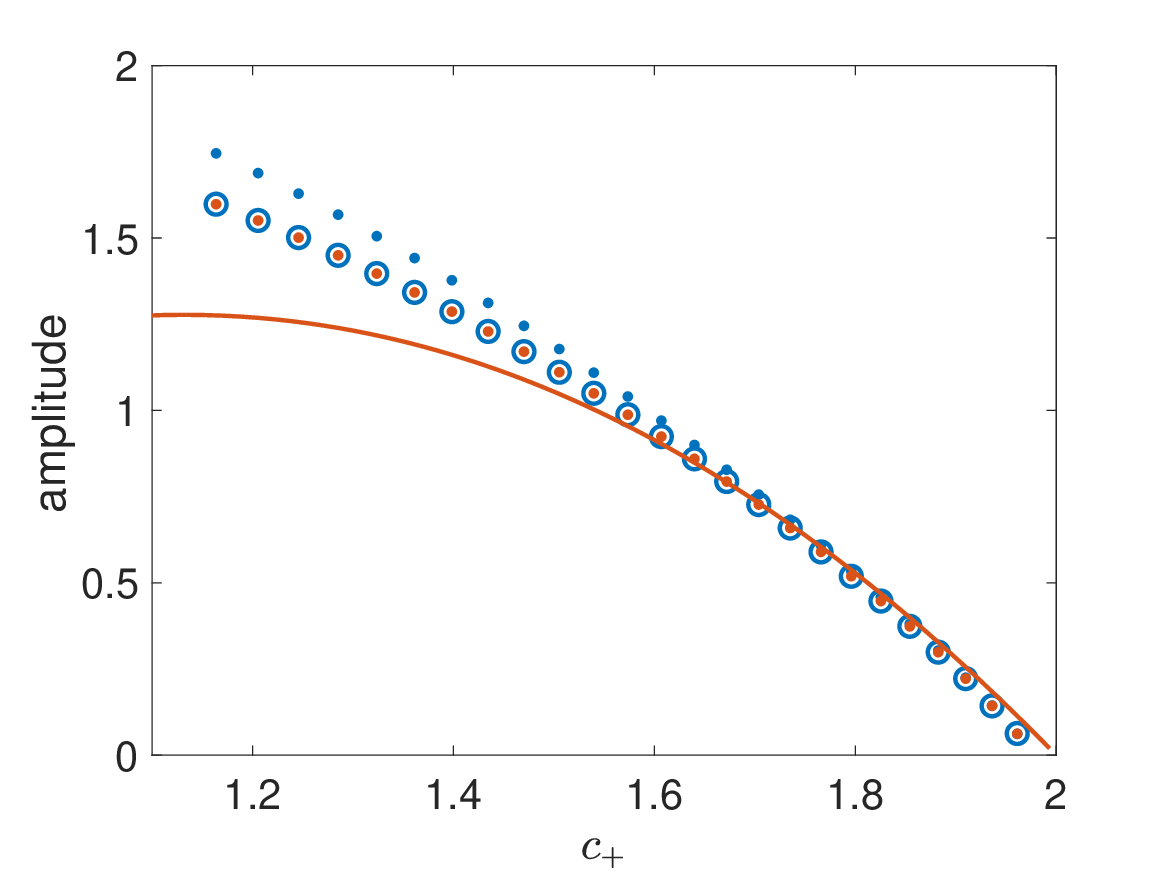} &
  \rlap{\hspace*{5pt}\raisebox{\dimexpr\ht1-.1\baselineskip}{\bf (c)}}
 \includegraphics[height=4.5cm]{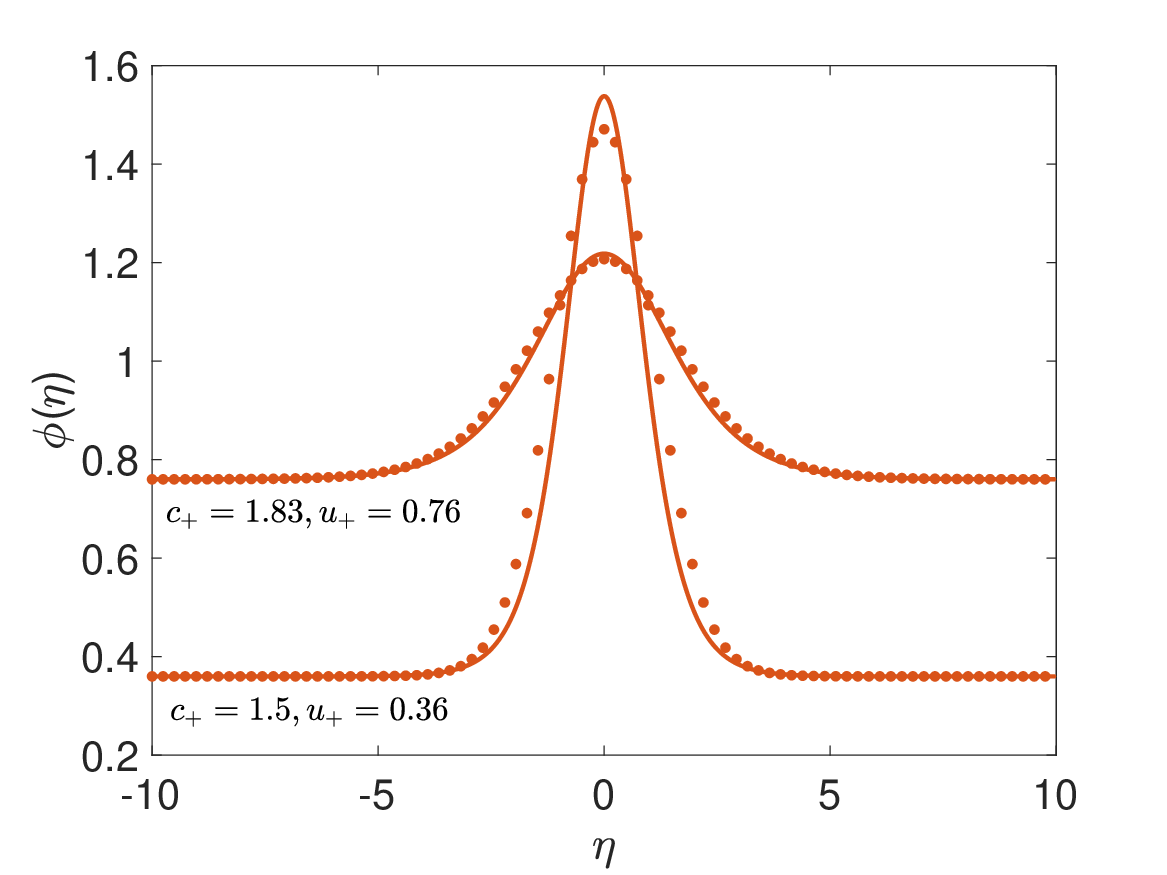}
  \end{tabular}
  }
  \caption{\textbf{(a)} Amplitude of the DSW soliton edge as a
    function of $u_+$ for the numerical simulation (open blue
    circles), the numerical solitary wave amplitude (solid red dots),
    the quasi-continuum formula $\frac{3c_+}{2} - 3 u_+$ (solid blue
    dots, with the $c_+$ obtained numerically), and the analytical prediction given by
    Eq.~\eqref{amplitude_uplus} (red line). \textbf{(b)} Same as panel
    (a), but the amplitude is shown against the wave speed
    $c=c_+$. \textbf{(c)} Comparison of the numerical solitary wave
    (solid red dots) and the quasi-continuum approximation given by
    Eq.~\eqref{qc}.  }
 \label{fig:compare_quasi}
\end{figure}

Because the quasi-continuum approximation of DSWs here and in
\cite{CHONG2022133533} performs well, we briefly contrast the Riemann
problems that result in DSWs for the lattice \eqref{eq:3} and BBM
\eqref{eq:BBM} equations.  The DSW fitting technique was applied to
the BBM Riemann problem in \cite{congy_dispersive_2021}.  In order to
compare our results for the lattice with DSWs in the BBM
equation~\eqref{eq:BBM}, we consider the initial transition 
occurring
between $u_- = 1$ and $u_+ = 1-\Delta$ for $0 < \Delta \ll 1$.  At the
DSW harmonic edge, the characteristic wavenumber
$\epsilon^2 K_-^2 = 4\Delta + 26 \Delta^2/9 + \cdots$ and speed
$C_- = 2 - 4 \Delta + 14 \Delta^2/9 + \cdots$ for BBM agree to
$\mathcal{O}(\Delta)$ with the expansions
$k_-^2 = 4 \Delta + \frac{4}{3} \Delta^2 + \cdots$ and
$c_- = 2 - 4 \Delta$ for the lattice.  Similarly, at the DSW soliton
edge, the conjugate wavenumber
$\epsilon^2 \tilde{K}_+^2 = 4 \Delta + 10 \Delta^2/9 + \cdots$ and
speed $C_+ = 2 - 2 \Delta/3 - 2 \Delta^2/27 + \cdots$ agree to
$\mathcal{O}(\Delta)$ with the lattice:
$\tilde{k}_+^2 = 4 \Delta + 8 \Delta^2/3 + \cdots$ and
$c_+ = 2 - 2 \Delta/3 - 8 \Delta^2/45 + \cdots$.  The DSW's soliton
edge amplitude in BBM is $A_+ = 2 \Delta - \Delta^2/9 + \cdots$
whereas the prediction \eqref{amplitude_uplus} for the lattice expands
as $a_+ = 2 \Delta - 4 \Delta^2/15+\cdots$.  Note that to leading
order, these predictions agree with the DSW edge characteristics of
the KdV equation $U_T + 2 UU_X + \tfrac13 \epsilon^2 U_{XXX} = 0$ for
$X = \epsilon n$, $T = \epsilon t$, $U(X,T) = u_n(t)$.  This is
expected because BBM and KdV are asymptotically equivalent to leading
order in the weakly nonlinear, long wavelength regime.

In summary, the quasi-continuum approximation of lattice DSWs by DSWs
in the BBM equation performs well for small initial jumps
$0 < \Delta \ll 1$ when the oscillation wavelengths are much larger
than the lattice spacing.  The agreement to $\mathcal{O}(\Delta)$ in
the DSW properties is expected and, in fact, is a statement of
universality of the KdV equation 
as a weakly nonlinear, long wavelength model of dispersive
hydrodynamics \cite{Mark2016}.
For sufficiently large $\Delta$, the
BBM DSW develops two-phase modulations near the trailing edge
\cite{congy_dispersive_2021}.  In contrast, for $\Delta > 1$, the
lattice DSW bifurcates into a partial DSW connected to a traveling
wave called a TDSW or exhibits blow up that we will describe in
Secs.~\ref{sec:pDSW} and \ref{sec:blow-up}, respectively.

\subsection{Dispersive shock wave  + stationary shock + rarefaction (DSW+SS+RW) }
\label{sec:dsw-ss-rw}

In this section we investigate the case where the initial step
generates two unsteady waves: a leftward moving DSW and a right moving
RW.  At the origin, there is a stationary shock (SS) joining symmetric
states at the level $u_n = \pm u_0$. A numerically computed example is
depicted in Fig.~\ref{fig:DSW_RW}. This class of solution is
empirically found for $u_+= 1$ and $u_-\in(-1,-0.26)$.  As will be
shown, the bifurcation at $u_- = -0.26$ occurs when the DSW's harmonic
edge exhibits zero velocity.

\begin{figure}[h!]
    \centering
    \includegraphics[scale = 0.2]{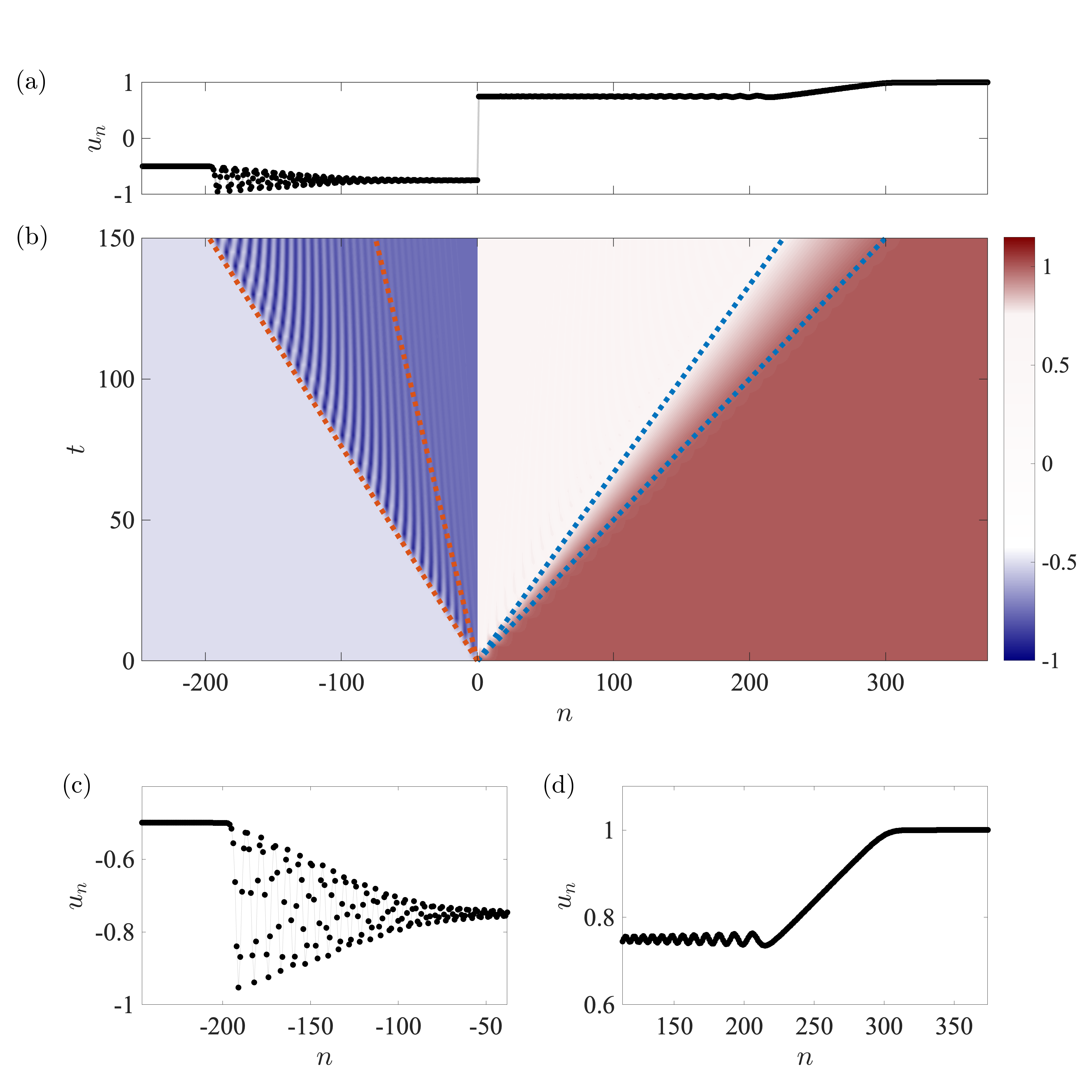} 
    \caption{(a) An example DSW + SS + RW solution of the Riemann
      problem at $t = 500$ for $u_- = -0.5$, $u_+ = 1$. (b) Solution
      contour plot in the space-time plane with the predicted edge velocities of the DSW and RW denoted by dashed lines. (c) Zoom-in of DSW.  (d)
      Zoom-in of RW.}
    \label{fig:DSW_RW}
\end{figure}

The numerical simulation shown in Fig.~\ref{fig:DSW_RW} suggests that
the solution can be approximated for large $t$ as:
\begin{align}
  \label{eq:31}
  u_n(t) = \begin{cases}
             u_- & n \leq s_- t \\ 
             u_{\rm DSW}(n,t) & s_- t < n \leq s_+ t \\ 
             -u_0 & s_+ t < n \leq 0\\
             u_0 & 0  < n \leq 2u_0 t\\
             u_{\rm RW}(n,t) & 2u_0 t < n  \leq 2u_+ t\\
             u_+ & u_+ t < n
           \end{cases} .
\end{align}
The velocities $s_- < s_+$ give the motion of the DSW's soliton and
harmonic edges, respectively.  Across all of the simulations
performed, we found the following relation for the intermediate,
symmetric states $\pm u_0$ to hold to very high precision
\begin{equation}
  \label{eq:30}
  u_0 = \frac{u_+ - u_-}{2} .
\end{equation}
This relation implies that the DSW and RW have the same jump height,
albeit with opposite polarities.  We have been unable to
mathematically justify this formula.  However, as we show in section
\ref{sec:non-uniq-riem}, the value of the intermediate state $u_0$
depends strongly on particular details of the initial data.  If the value
of $u_0(0)$ is changed, then the intermediate value $u_0$ differs from
\eqref{eq:30}.

Utilizing the formula \eqref{eq:30}, we can completely determine the
velocities that divide the approximate solution \eqref{eq:31} into
different wave patterns.  To determine the DSW edge velocities, we use
Eqs.~\eqref{eq:32} and \eqref{eq:cplus}, which were derived under the
assumption that $u_- > u_+ > 0$.  In the case of the solution
\eqref{eq:31}, the left ($u_-$) and right ($-u_0$) states are both
negative.  Since the governing equation \eqref{eq:3} is invariant
under the transformation $u_n(t) \to -u_{-n}(t)$, the DSW velocities
are mapped as follows
\begin{equation}
  \label{eq:33}
  s_-(u_-,-u_0) = -c_+(u_0,-u_-), \quad s_+(u_-,-u_0) = -c_-(u_0,-u_-) .
\end{equation}
Then, using \eqref{eq:30}, we find
\begin{equation}
  \label{eq:34}
  s_- = - \frac{2}{\mathrm{cosh}^{-1}(u_+/|u_-|)} \sqrt{u_+^2 - u_-^2}
  , \quad s_+ = 3 u_- + u_+ .
\end{equation}
Figure \ref{fig:DSW_RW}(b) shows good agreement between the
  predicted velocities of the approximate solution \eqref{eq:31} and a
  numerical simulation when $u_+ = 1$, $u_- = -0.5$.

The DSW remains detached from the stationary shock (SS) so long as the
harmonic edge velocity $s_+$ remains negative.  From
Eq.~\eqref{eq:34}, we predict that the DSW is no longer detached from
the SS when $3 u_- + u_+ = 0$, which, for $u_+ = 1$, occurs when
$u_- = -\tfrac13$.  As noted earlier, the bifurcation from the DSW +
SS + RW to the US case is empirically identified as occurring when
$u_- = -0.26$.  As shown in Fig.~\ref{fig:DSW_RW}(b), this small
discrepancy in the bifurcation value can be explained by the deviation
of the computed DSW harmonic edge velocity from the DSW fitting
prediction \eqref{eq:32}.

\subsection{Traveling dispersive shock wave (TDSW)}
\label{sec:pDSW}

In this section, we consider the case where a partial DSW connects the
level behind, $u_-$, to a periodic-to-equilibrium traveling wave
solution to the level ahead $u_+$.  Although we do not directly
compute it as a traveling wave solution of the discrete equation, it
is interpreted as a heteroclinic connection between a periodic orbit
with the constant level ahead $u_+$ based on an analysis of the
numerical simulations of the Riemann problem.  Such heteroclinic
solutions of continuum equations with higher (fifth) order dispersion
were studied in \cite{Sprenger_2020,sprenger_traveling_2023} and were
associated with so-called traveling dispersive shock waves (TDSWs)
that emerge from an associated Riemann problem.

For $u_- = 1$ and $u_+ \in (-0.724,0)$, numerical simulations show a
qualitatively similar solution pattern to that depicted in
Fig.~\ref{fig:pDSW_TW} in which $u_n(t)$ begins on the left with the
value $u_-$.  It then progresses into an oscillatory wavetrain with
increasing amplitude that resembles the leftmost portion of a DSW,
called a partial DSW, that is then connected to a periodic traveling
wave.  The periodic traveling wave is connected to the constant state
ahead $u_+$ via an abrupt transition that moves with the same speed.
Collectively, this partial DSW and traveling wave is referred
to as the
traveling DSW.

\begin{figure}[H]
    \centering
    \includegraphics{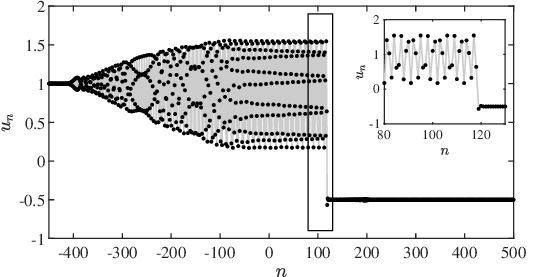}
    \caption{Example TDSW that emerges from the initial data
      \eqref{step} with $u_- = 1$ and $u_+ = -0.5$ at $t = 200$. The
      inset is a zoom-in of the boxed region, showing details of the
      periodic-to-equilibrium traveling wave at the leading edge. }
    \label{fig:pDSW_TW}
\end{figure}

The TDSW solution studied here is the discrete analogue of the TDSW
studied for the KdV5 equation \cite{Sprenger_2020}. The terminology
\textit{traveling dispersive shock wave} unites the unsteady component of the partial DSW and the steady traveling wave component (a heteroclinic periodic-to-equilibrium solution) to which it is attached in this non-classical DSW.  Consequently, the entire TDSW structure is unsteady.

\subsubsection{Approximation of the traveling wave via the quasi-continuum model}

One can describe the periodic portion of the TW (see e.g. the interval
$n\in[-90,100])$ of Fig.~\ref{fig:pDSW_TW}) using the continuum
reduction presented in section~\ref{sec:quasi}. In particular, there is a
three-parameter family of traveling wave solutions of the
quasi-continuum BBM Eq.~\eqref{eq:BBM}, given by
\begin{equation} \label{eq:quasi-TW}
    u_n(t) = r_1 + r_2 - r_3 + 2(r_3 - r_1) \dn^2\left( \sqrt{\frac{2(r_3 - r_1)}{c}}( n - c t),m\right) ,\quad m = \frac{r_2-r_1}{r_3-r_1}, \quad c = \frac{2(r_1 + r_2 + r_3)}{3}.
\end{equation} 
We treat $r_1,r_2,r_3$ as fitting parameters. After a sufficiently
long time, the traveling wave in the numerical simulation forms, as in
Fig.~\ref{fig:pDSW_TW}. We then isolate a small interval of that
traveling wave, which is then fit to
Eq.~\eqref{eq:quasi-TW}. Figure~\ref{fig:compare_qausi_TW}(a) shows a
comparison of the actual lattice dynamics at $t=480$ (blue markers)
and quasi-continuum approximation (blue lines) with the step values
$u_-=1$ and $u_+=-0.16$.
Figure~\ref{fig:compare_qausi_TW}(b) shows the trajectory in the phase
plane $(u_{400}(t),\dot{u}_{400}(t))$ (two outermost lobes) and
$(u_{140}(t),\dot{u}_{140}(t))$ (two innermost lobes) for various
values of $u_-$ for the actual lattice dynamics (markers) and
quasi-continuum approximation (lines). We show the phase plane for
different values of $n$ since the location of the traveling wave
within the lattice is moving.
The agreement is quite good throughout the interval of existence for
these structures, but is best when the jump height is smallest
(compare the blue and red trajectory of
Figure~\ref{fig:compare_qausi_TW}(b)). The comparison of the
frequency, amplitude, and mean parameters is shown in
Figure~\ref{fig:compare_qausi_TW}(c).  These are computed via the
following formulas with $n$ fixed.  For $u_-\in[-0.72,-0.28]$ the
lattice index is fixed to $n=140$, for $u_-\in(-0.28,-0.2]$ the index
is $n=400$ and for $u_-=0.08$ the index is $n=500$.  The frequency is
$f = 1/T$, where $T$ is the period (computed as the peak-to-peak time
of the trajectory); the mean is
$$\overline{u} = \frac{1}{T}\int_{I_T} u_n(t) \, dt$$ where $I_T$ is the time
interval of one oscillation period. For the computations shown here,
it is $I_T = [480-T,480]$. The amplitude is
$$a = \max_{t\in I_T}u_n(t) - \min_{t\in I_T} u_n(t).$$
Once the best-fit values of $r_1,r_2,r_3$ are obtained, the wave
parameters can be computed directly from Eq.~\eqref{eq:quasi-TW} as
$$f = \frac{\sqrt{2 c (r_3 - r_1)}}{2K(m)}, \quad \overline{u} = r_1 +
r_2 - r_3 + 2(r_3-r_1) \frac{E(m)}{K(m)}, \quad a=2(r_2 - r_1),$$
where $K(m)$ and $E(m)$ are the complete elliptic integrals of the
first and second kind, respectively. We note that while
Eq.~\eqref{eq:BBM} is able to describe the local periodic traveling
wave dynamics of the TDSW structure, it does not admit solutions
resembling the entire TDSW structure since heteroclinic
periodic-equilibrium solutions do not exist for the planar ODE
\eqref{eq:39}. Such a description may be possible by using the (1,5)
Pad\'e approximant instead of the (1,3) approximant in \eqref{eq:29}
to arrive at the 5th order model
$$U_T + (U^2)_X - \frac{\epsilon^2}{6} U_{XXT} + \frac{7
  \epsilon^4}{360} U_{XXXXT} = 0.$$ Similar models have been shown to
admit such solutions \cite{Sprenger_2020}. While this is an interesting
topic for further study, we will not pursue the identification of such
a heteroclinic orbit further herein.

\begin{figure}
  \centering
  \begin{tabular}{@{}p{0.33\linewidth}@{}p{0.33\linewidth}@{}p{0.33\linewidth}@{} }
     \rlap{\hspace*{5pt}\raisebox{\dimexpr\ht1-.1\baselineskip}{\bf (a)}}
 \includegraphics[height=4.3cm]{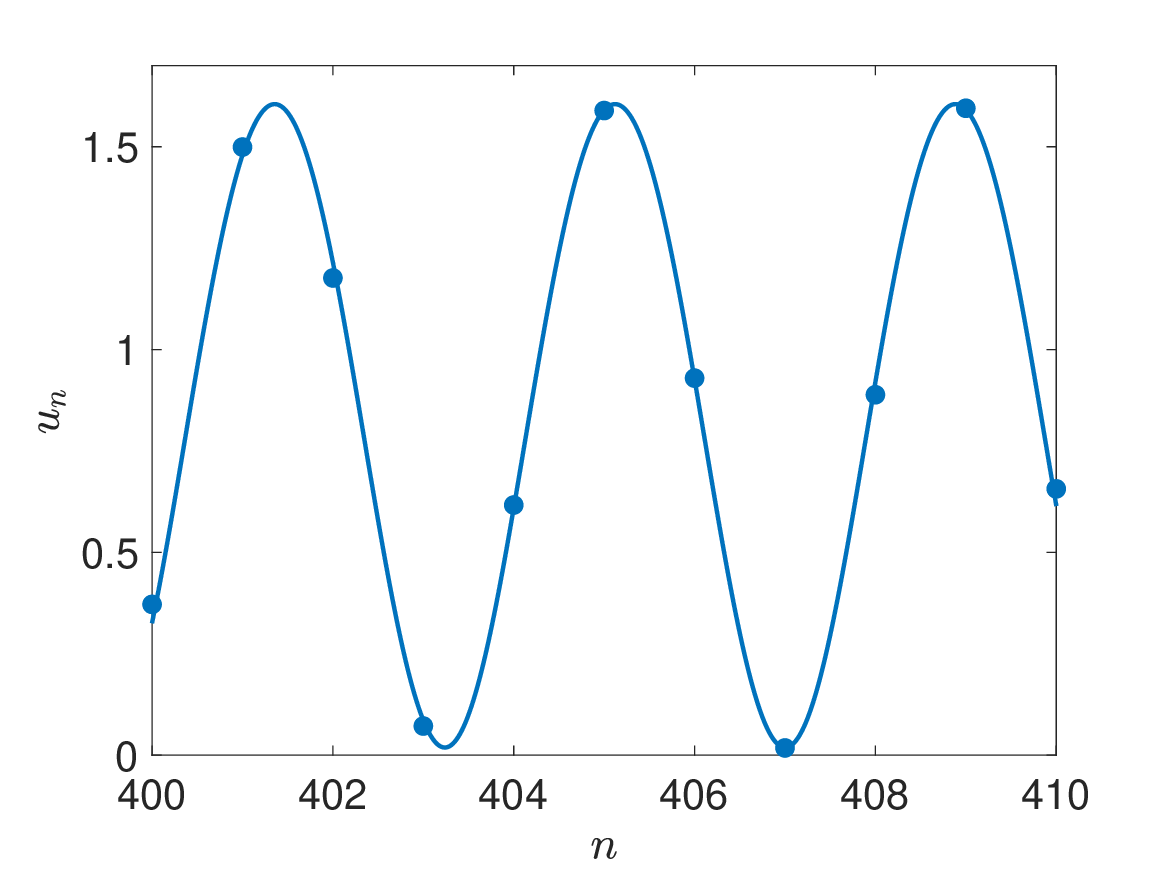} &
  \rlap{\hspace*{5pt}\raisebox{\dimexpr\ht1-.1\baselineskip}{\bf (b)}}
 \includegraphics[height=4.3cm]{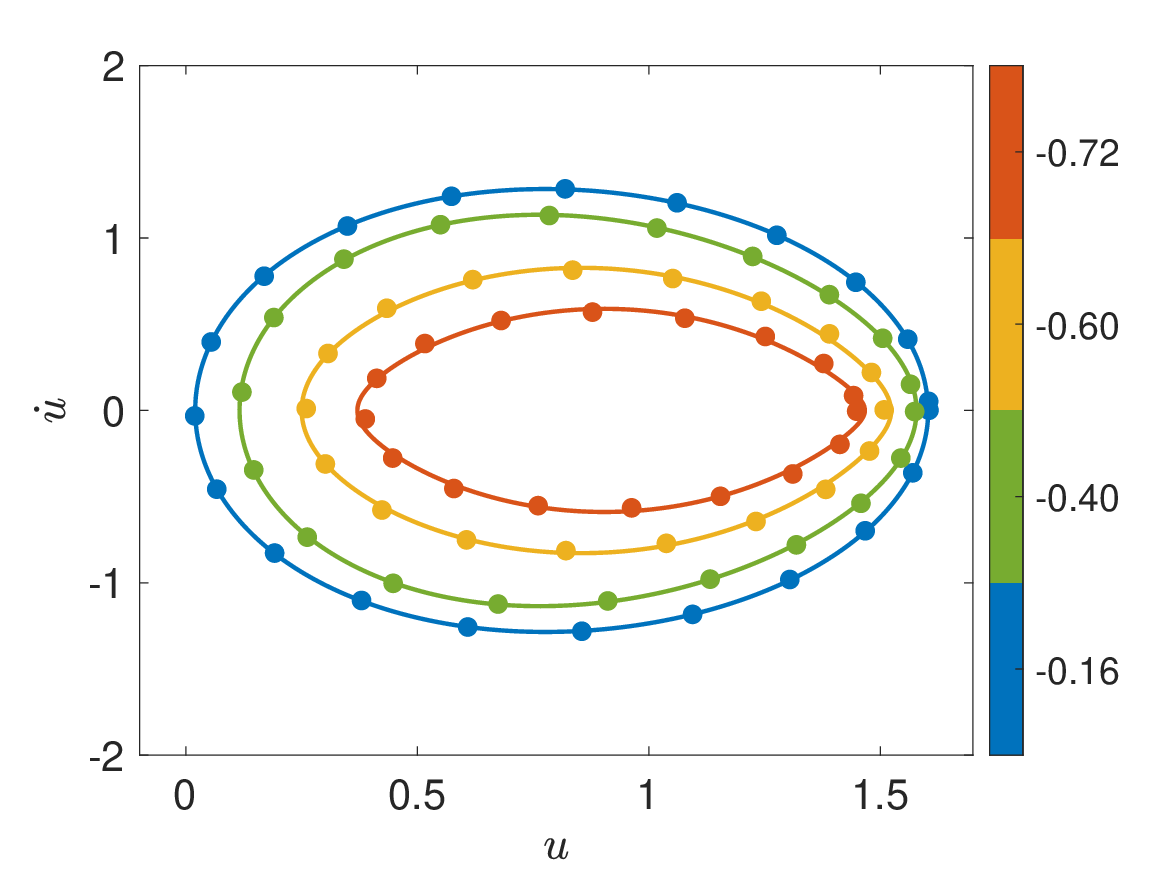} &
   \rlap{\hspace*{5pt}\raisebox{\dimexpr\ht1-.1\baselineskip}{\bf (c)}}
 \includegraphics[height=4.3cm]{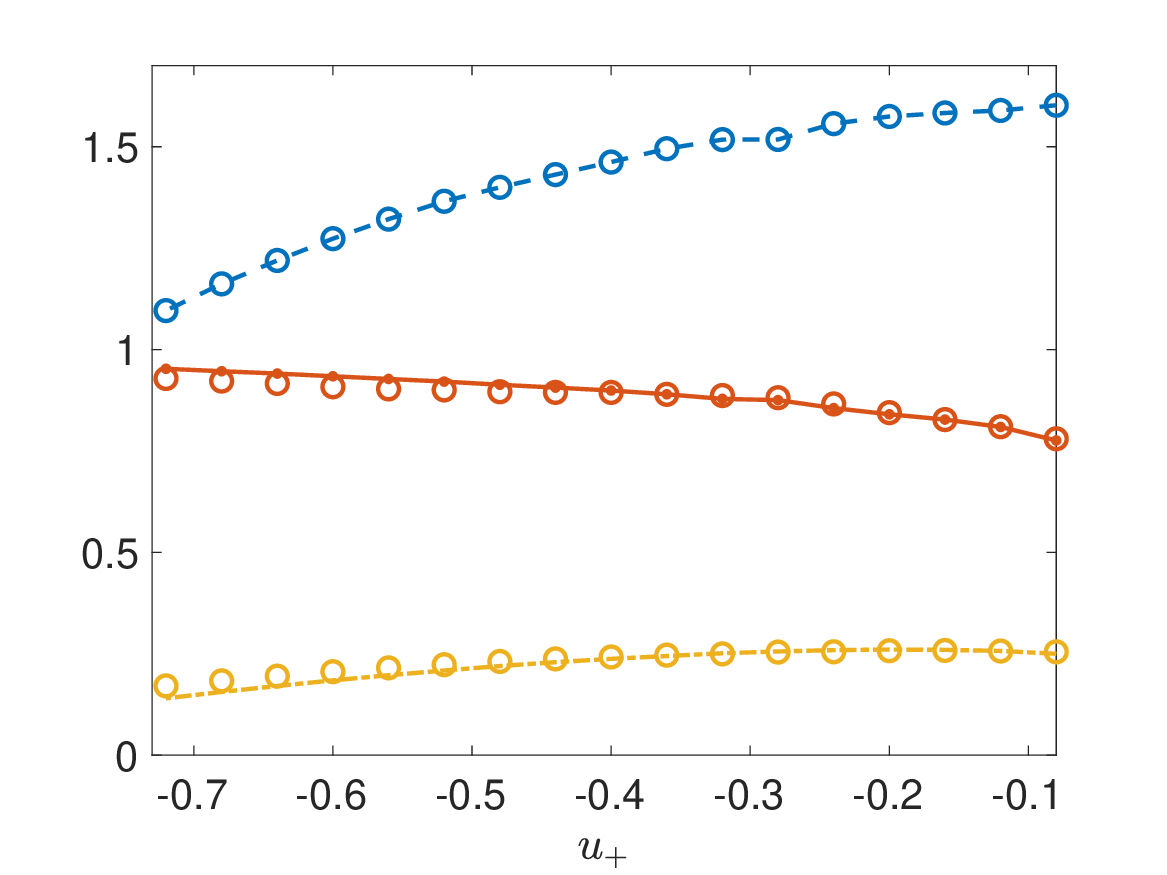} 
  \end{tabular}
  \caption{\textbf{(a)} Zoom of the periodic wave in the gray shaded
    region of Fig.~\ref{fig:pDSW_TW}(a) with $u_+ = -0.16$ (blue
    markers) and quasi-continuum approximation (solid blue
    line). \textbf{(b)} Plot of the phase plane
    $(u_{n}(t),\dot{u}_{n}(t))$ for a time interval such that the
    periodic wave has developed. The color intensity corresponds to
    the value of $u_+$.  Since the interval containing the traveling wave
    changes as $u_+$ is changed, the value of $n$ for each loop is not
    fixed.  In particular $u_+ = -0.16$ (blue, $n=400$), $u_+ = -0.40$
    (green, $n=400$), $u_+ = -0.6$ (yellow, $n=140$), $u_+ = -0.72$
    (red, $n=140$). The solid lines are the corresponding
    quasi-continuum approximations. \textbf{(c)} Plot of the mean (red
    solid lines), amplitude (blue dashed lines) and frequency $1/T$
    (yellow dashed-dot line) as a function of $u_-$. The
    quasi-continuum approximation of these wave parameters is shown
    as open circles.  For $u_+\in[-0.72,-0.28]$ the lattice index is
    fixed to $n=140$ and for $u_+\in(-0.28,-0.2]$ to $n=400$ and for
    $u_+=.08$ to $n=500$.  }
    \label{fig:compare_qausi_TW}
\end{figure}

\subsubsection{Modulation solution of the weakly nonlinear Whitham modulation equations}

In order to obtain the form of an approximate modulation solution
$\mathbf{q} = [\bar{u},a,k]^T$ of the Whitham equations
\eqref{eq:cont_mod} that describes the TDSW, we appeal to the
structure of the TDSW evident in Fig.~\ref{fig:pDSW_TW}. An
oscillatory wavetrain emerges from the left level $u_-$ with
increasing amplitude that saturates at a periodic traveling wave.  The
traveling wave then abruptly transitions to the right level $u_+$.
Guided by previous work on the traveling dispersive shock wave (TDSW)
solutions of a fifth-order Korteweg-de Vries equation
\cite{Sprenger_2020}, we make the self-similar modulation ansatz
($\xi = X/T = n/t$)
\begin{equation}
  \label{eq:40}
  \mathbf{q}(\xi) =
  \begin{cases}
    \mathbf{q}_- & \xi < \lambda_2(\mathbf{q}_-) , \\ 
    \mathbf{q}_{\rm RW}(\xi)
                 & \lambda_2(\mathbf{q}_-)  \leq \xi <
                   \lambda_2(\mathbf{q}_{\rm p}) , \\  
    \mathbf{q}_{\rm p}
                 & \lambda_2(\mathbf{q}_{\rm p}) \leq \xi <
                  \lambda_2(\mathbf{q}_+) , \\ 
    \mathbf{q}_+
                 &  \lambda_{2}(\mathbf{q}_+)  \leq \xi ,
  \end{cases}
\end{equation}
for Eq.~\eqref{eq:14} where $\lambda_2$ is the middle characteristic
velocity.  Additionally, the constant states are
\begin{align}
  \label{eq:41}
  \mathbf{q}_{-} = [u_-,0,k_-]^T, \quad \mathbf{q}_{\rm p} =
  [\overline{u}_{\rm p},a_{\rm p},k_{\rm p}]^T, \quad \mathbf{q}_+ = [u_+,a_+,0]^T,
\end{align}
and $\mathbf{q}_{\rm RW}(\xi)$ is the rarefaction solution (integral
curve) of Eq.~\eqref{eq:17} for the second characteristic field
$(\lambda_2,\mathbf{r}_2)$ that continuously connects the harmonic
edge state $\mathbf{q}_-$ and the periodic traveling wave state
$\mathbf{q}_{\rm p}$ of the TDSW.  The discontinuity from
$\mathbf{q}_{\rm p}$ to $\mathbf{q}_+$ satisfies the jump conditions
\eqref{eq:19} for the Whitham modulation equations where
$V = \lambda_2(\mathbf{q}_+)$ is simultaneously the phase speed of the
periodic traveling wave with parameters $\mathbf{q}_{\rm p}$, the
shock speed, and the phase speed of the solitary wave with parameters
$\mathbf{q}_+$, i.e., it represents the TW component of the TDSW
solution.  The modulation solution \eqref{eq:40} corresponds to a
rarefaction-shock solution of the Whitham modulation equations.

The five parameters $(k_-,\mathbf{q}_{\rm p},a_+)$ in the modulation
solution \eqref{eq:40} could, in principle, be obtained by solving the
full Whitham modulation equations \eqref{eq:cont_mod}, but we lack
explicit periodic traveling wave solutions.  Instead, we approximate
the modulation solution \eqref{eq:40} in the weakly nonlinear regime
by solving Eqs.~\eqref{eq:weakly_nl_mod}.

\begin{figure}[h]
    \centering
    \includegraphics{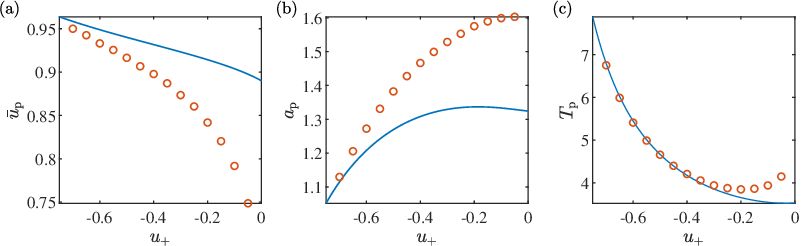}
    \caption{Traveling DSW parameters obtained from the modulation
      solution \eqref{eq:40} (curves) and numerical simulation
      (circles). 
      }
    \label{fig:TDSW-params}
\end{figure}

\begin{figure}[h]
    \centering
    \includegraphics{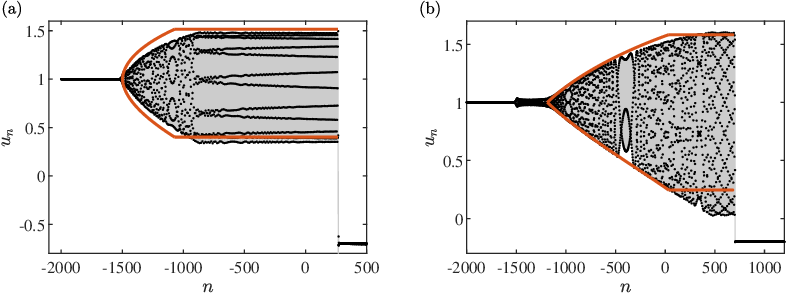}
    \caption{Two solutions at $t = 1000$ along with the envelope
      predictions from the modulation solution \eqref{eq:40}. (a)
      $u_+ = -0.7$ and (b) $u_+ = -0.2$.}
    \label{fig:TDSW-envelope}
\end{figure}

A general feature of Whitham modulation systems in the weakly
nonlinear regime is the $\mathcal{O}(a^2)$ mean induced by the finite
amplitude modulated wavetrain \cite{Whitham74}.  Absent mean changes
due to initial/boundary data, the third order weakly nonlinear
modulation system \eqref{eq:weakly_nl_mod} can be simplified, with the
same order of accuracy, to a second order modulation system.  The
procedure to do so is as follows.  The induced mean is represented by
the ansatz
\begin{equation}
  \label{eq:42}
  \overline{u}(X,T) = u_0 + a(X,T)^2 u_2(k(X,T)) + \cdots,
\end{equation}
where $u_0 \in \R$ is a constant background.  This introduces the
induced mean coefficient $u_2(k)$ that gives rise to an effective
nonlinear frequency shift $\tilde{\omega}_2$ by inserting
\eqref{eq:42} into Eq.~\eqref{eq:stokes_omega} and expanding as
\begin{equation}
  \label{eq:44}
  \begin{split}
    \omega(k,\bar{u}) &= \omega_0(k,u_0) + a^2\left ( \omega_2(k,u_0) +
                        \partial_{\overline{u}} \omega_0(k,u_0) u_2(k) \right ) +
                        o(a^2), \\
                      &\equiv \omega_0(k,u_0) + a^2 \tilde{\omega}_2(k,u_0) +
                        o(a^2) .
  \end{split}
\end{equation}
With this effective nonlinear frequency shift, weakly nonlinear wave
modulations are generically described by the simplified system
\cite{Whitham74}
\begin{subequations}
  \label{eq:2x2}
  \begin{align}
    \label{eq:2x2a}
    \left  (a^2 \right )_t + \left ( \omega_{0,k} a^2 \right )_x &= 0 ,\\
    \label{eq:2x2k}
    k_t +  \left ( \omega_0 \right )_x +  \tilde{\omega}_2 \left (a^2
    \right )_x & = 0 , 
  \end{align}
\end{subequations}
provided
\begin{equation}
  \label{eq:47}
  u_2(k) = -\frac{1}{16 u_0(1-\cos k)}, \quad \tilde{\omega}_2(k) =
  \frac{(\cos{k}-2) \cot(k/2)}{16 u_0} .
\end{equation}
Note that the additional coupling term involving $\tilde{\omega}_2$
and its derivatives contribute at higher order in $a$.  The induced
mean coefficient $u_2(k)$ in Eq.~\eqref{eq:47} is determined by
compatibility of averaged mean, energy conservation laws
\eqref{eq:a_weakly_nl_mod}, \eqref{eq:b_weakly_nl_mod} with the
induced-mean modulation system \eqref{eq:2x2}, which asymptotically
satisfies $(F(k) a^2)_T + (F(k)\omega_{0,k} a^2)_X = 0$ for any
differentiable $F$, in particular $F(k) = u_2(k)$.

Under the assumption of induced mean variation, we can analytically
obtain the rarefaction solution $\mathbf{q}_{\rm RW}(\xi)$ in
\eqref{eq:40} by solving Eq.~\eqref{eq:2x2} for a RW and then
inserting it into Eq.~\eqref{eq:42}.  For this, we express the
induced-mean modulation system in Riemann invariant form
\begin{equation}
  \label{eq:43}
  \frac{\partial r_{\pm}}{\partial t} + \lambda_\pm \frac{\partial
  r_{\pm}}{\partial x} = 0, 
\end{equation}
where 
\begin{equation}
  \label{eq:46}
  \begin{split}
    r_\pm &= a \mp \frac{1}{2}\int \left(
            \frac{\omega_{0,kk}}{\tilde{\omega}_2} \right)^{1/2} dk  \\
          &= a \mp 2\sqrt{2}\, u_0 \cos^{-1}\left(\tfrac13(-1+2\cos
            k)\right), \quad u_0 >
            0, \quad  0 < k <
            \pi, \quad a \ge 0. 
  \end{split}
\end{equation}
The restriction to positive mean and wavenumber is due to the fact
that we have selected the positive square root in \eqref{eq:46}.  The
characteristic velocities are
\begin{equation}
  \label{eq:45}
  \lambda_\pm = 2 u_0 \cos k \pm \frac{a}{2} \cos\left ( \frac{k}{2} \right )
  \sqrt{2 - \cos k},
\end{equation}
so that $\lambda_+ = \lambda_2$ and $\lambda_- = \lambda_1$ in
Eq.~\eqref{eq:48}.

We can now solve for $\mathbf{q}_{\rm RW}(\xi)$ in \eqref{eq:40} by
setting $u_0 = u_-$ and taking the fast RW solution of
Eq.~\eqref{eq:43} that satisfies $r_- = \mathrm{const}$ and $\lambda_+ = \xi$.
The constant slow Riemann invariant $r_-$ implies
\begin{subequations}
  \label{eq:49}
  \begin{equation}
    \label{eq:50}
    \cos^{-1}\left(\tfrac13(-1 + 2\cos k_-)\right) 
    = \frac{a_{\rm p}}{ 4\sqrt{2} u_-} +
    \cos^{-1}\left(\tfrac13(-1 + 2\cos k_{\rm p})\right)  , 
  \end{equation}
  and the assumption of induced mean implies
  \begin{equation}
    \label{eq:51}
    \overline{u}_{\rm p} = u_- + u_2(k_{\rm p}) a_{\rm p}^2 .
  \end{equation}
  The RW profile is obtained by inverting $\lambda_+ = \xi$ for
  $a(\xi)$ and $k(\xi)$ subject to the constraint
  \begin{equation}
    \label{eq:53}
    \cos^{-1}\left(\tfrac13(-1 + 2\cos k_-)\right) 
    = \frac{a(\xi)}{ 4\sqrt{2} u_-} +\cos^{-1}\left(\tfrac13(-1 +
      2\cos k(\xi))\right) .
  \end{equation}
  Equations \eqref{eq:50} and \eqref{eq:51} are two conditions on the
  four unknown solution parameters
  $(k_-,\overline{u}_{\rm p}, a_{\rm p}, k_{\rm p})$.  The other two
  conditions are obtained from the jump conditions \eqref{eq:19}.

  The sharp transition from the periodic traveling wave
  $\mathbf{q}_{\rm p}$ to the solitary wave ahead $\mathbf{q}_+$ is
  achieved by a shock solution of the Whitham modulation equations.
  We obtain the jump conditions from the conservation laws
  \eqref{eq:21} and \eqref{eq:22} by assuming that the periodic
  traveling wave is in the weakly nonlinear regime \eqref{eq:stokes_u}
  and the level ahead is a solitary wave where $k \to 0$, both with
  the same phase speed $V$:
  \begin{align}
    -V\left(\overline{u}_{\rm p} - u_+\right) +  \overline{u}_{\rm p}^2 +
    \frac{1}{8}a^2_{\rm p} - u_+^2
    & = 0 \label{eq:weak_jump1} \\ 
    -V\left(\frac{1}{3}\overline{u}_{\rm p}^3 +
    \frac{1}{8}\overline{u}_{\rm p} 
    a^2_{\rm p} - \frac{1}{3}u_+^3\right) + \frac{1}{2}\overline{u}_{\rm p}^4
    + \frac{1}{4}a_{\rm p}^2\left(\overline{u}_{\rm p}^2\cos(k_{\rm p}) + \frac{1}{2} u_{\rm
    p}^2\right) - \frac{1}{2}u_+^4 
    & = 0 .  \label{eq:weak_jump2}
  \end{align}
  The jump condition from the conservation of waves equation
  \eqref{eq:23} is satisfied because $\omega = k = 0$ for the solitary
  wave and $V = \omega_{\rm p}/k_{\rm p}$ for the periodic traveling
  wave
  \begin{equation}
    \label{eq:52}
    V = \frac{\omega_0(k_{\rm p},u_-) + a_{\rm p}^2 \tilde{\omega}_2(k_{\rm
        p})}{k_{\rm p}} .
  \end{equation}
\end{subequations}

Solving for $a_{\rm p}$ and $\overline{u}_{\rm p}$ from \eqref{eq:50}
and \eqref{eq:51}, then inserting them into \eqref{eq:weak_jump1},
\eqref{eq:weak_jump2} and using the phase velocity \eqref{eq:52}
determines two nonlinear equations for $k_-$ and $k_{\rm p}$.  We
solve these equations numerically using standard root finding methods
to obtain all the parameters of the shock-rarefaction modulation
solution \eqref{eq:40} and compare it with numerical simulation in
Figs.~\ref{fig:TDSW-params} and \ref{fig:TDSW-envelope}.
We observe in the figures that near the onset of the TDSW, the
solution's mean, amplitude and frequency are accurately captured by
the above theory. On the other hand, as the amplitude of the solution
increases, the approximation loses quantitative efficacy.
Nevertheless, the qualitative trend of the solution's properties are
captured by the above analysis.


\begin{figure}
    \centering
    \includegraphics{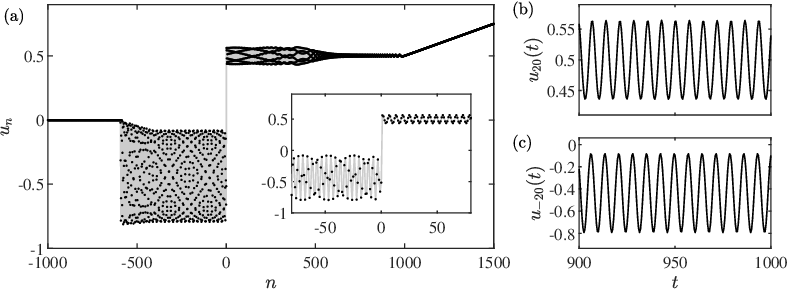}
    
    \caption{(a) Unsteady shock solution arising from step initial
      data \eqref{step} with $u_- = 0$ and $u_+ = 1$. (b,c)
      series data of the evolution of lattice sites $n = +20$ (b)
      and $n = -20$ (c) for $t \in [900,1000]$.
      } 
          \label{fig:unsteady_solution}
\end{figure}

\subsection{Unsteady Shock (US)}
\label{sec:unsteady}

In section \ref{sec:dsw-ss-rw}, we found that for $u_+ = 1$ and
$u_- \in (-1,-0.26)$, a stationary shock separated two counterpropagating
waves, one a DSW, the other a RW.  When $u_- = -0.26$, the DSW no
longer separates from the stationary shock.  Instead, the DSW+SS+RW is
numerically observed to bifurcate into the unsteady generation of
counterpropagating periodic waves that we term an unsteady shock (US)
for $u_- \in (-0.26,0.18)$.  When $u_-$ exceeds $0.18$, the US
bifurcates into a RW, described in section \ref{sec:RW}.
We now investigate the US.

A plot of the solution for $u_- = 0$ 
and $u_+ = 1$ is given in
Figure \ref{fig:unsteady_solution}. Two distinct counterpropagating
periodic waves traveling with speed $c_\pm$ emerge from the origin that then transition to the
constant level $u_-$ behind through a partial DSW and to $u_+$ ahead
via a partial DSW and a RW. 
%
While the velocities $c_\pm$ are distinct, Figure
\ref{fig:unsteady_solution}(b) depicting the time series
$u_{\pm 20}(t)$ indicates that the two periodic waves have
approximately the same temporal frequency, which we confirm below
numerically to high precision.

\begin{figure}
    \centering
    \centering
       \begin{tabular}{@{}p{0.4\linewidth}@{}p{0.4\linewidth}@{}}
     \rlap{\hspace*{5pt}\raisebox{\dimexpr\ht1-.1\baselineskip}{\bf (a)}}
 \includegraphics[height=5cm]{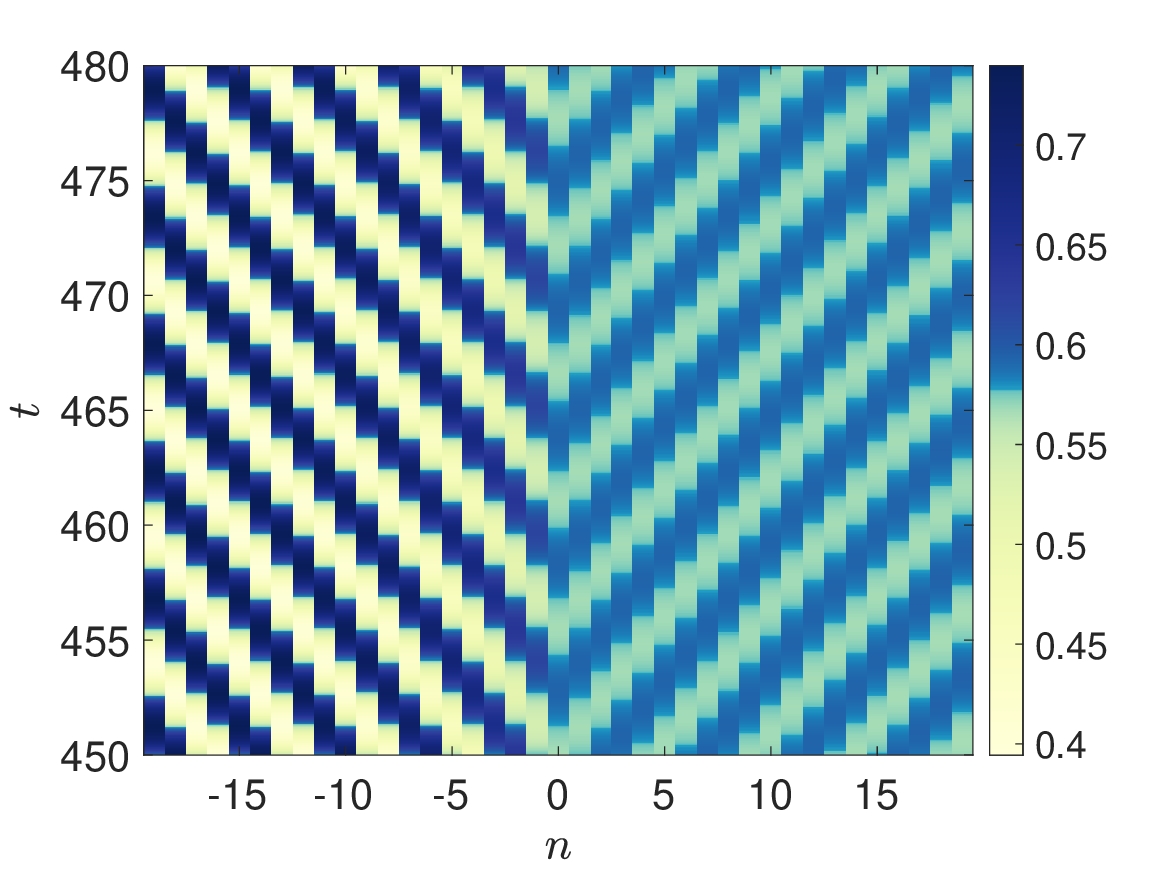} &
  \rlap{\hspace*{5pt}\raisebox{\dimexpr\ht1-.1\baselineskip}{\bf (b)}}
 \includegraphics[height=5cm]{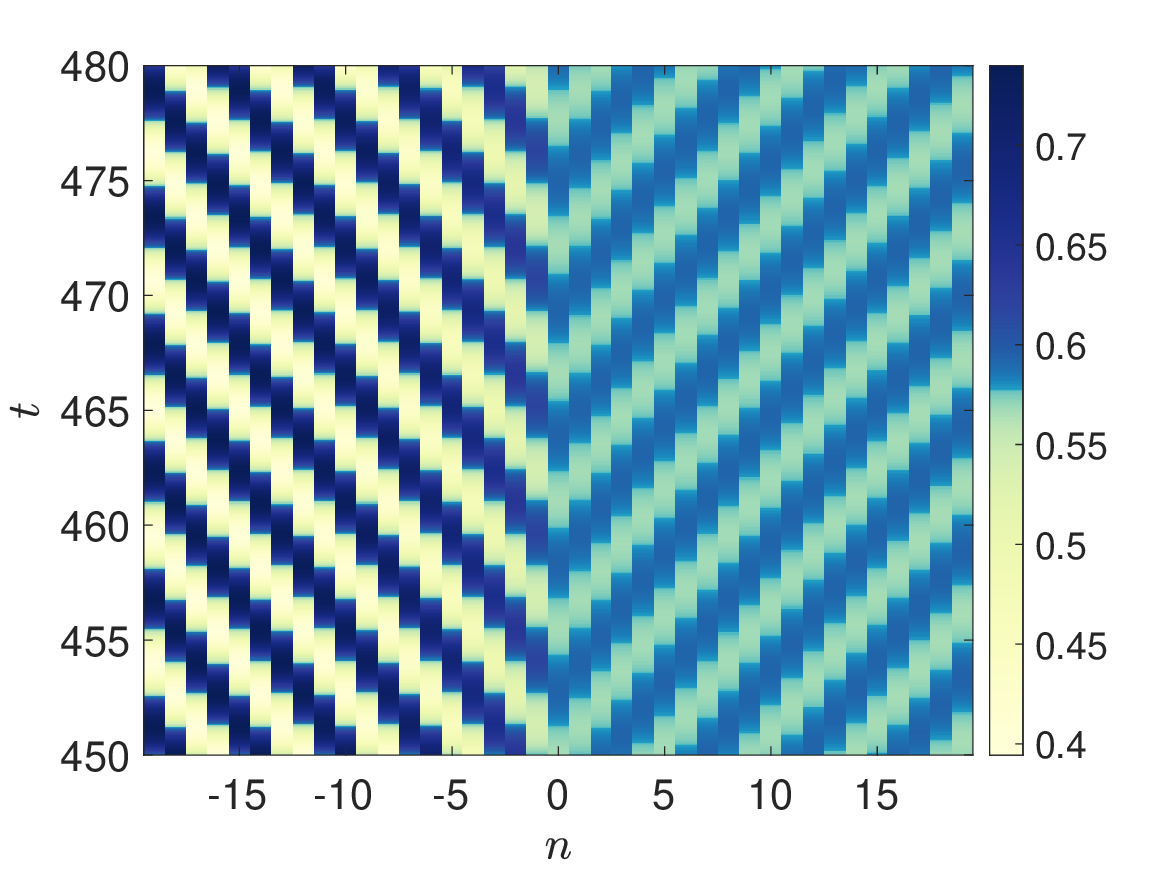} 
  \end{tabular}
         \begin{tabular}{@{}p{0.4\linewidth}@{}p{0.4\linewidth}@{}}
     \rlap{\hspace*{5pt}\raisebox{\dimexpr\ht1-.1\baselineskip}{\bf (c)}}
 \includegraphics[height=5cm]{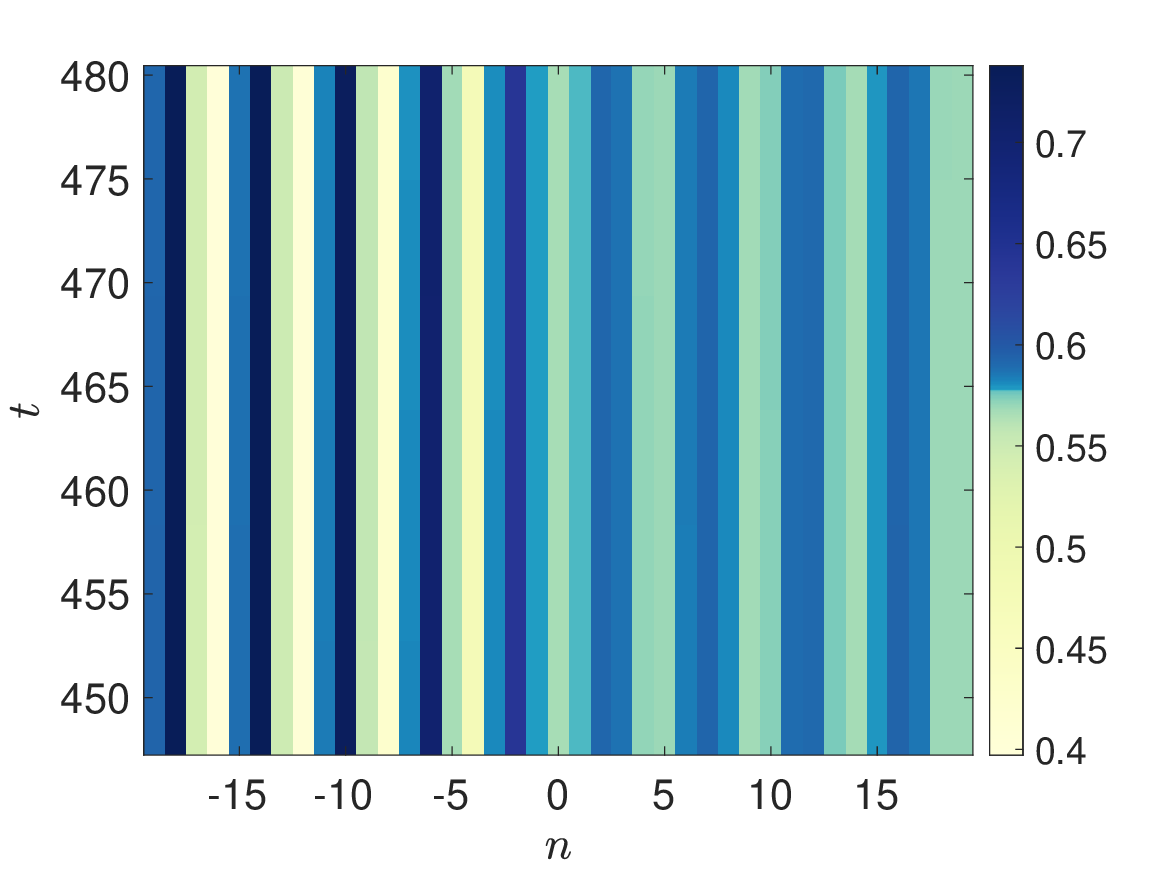} &
  \rlap{\hspace*{5pt}\raisebox{\dimexpr\ht1-.1\baselineskip}{\bf (d)}}
 \includegraphics[height=5cm]{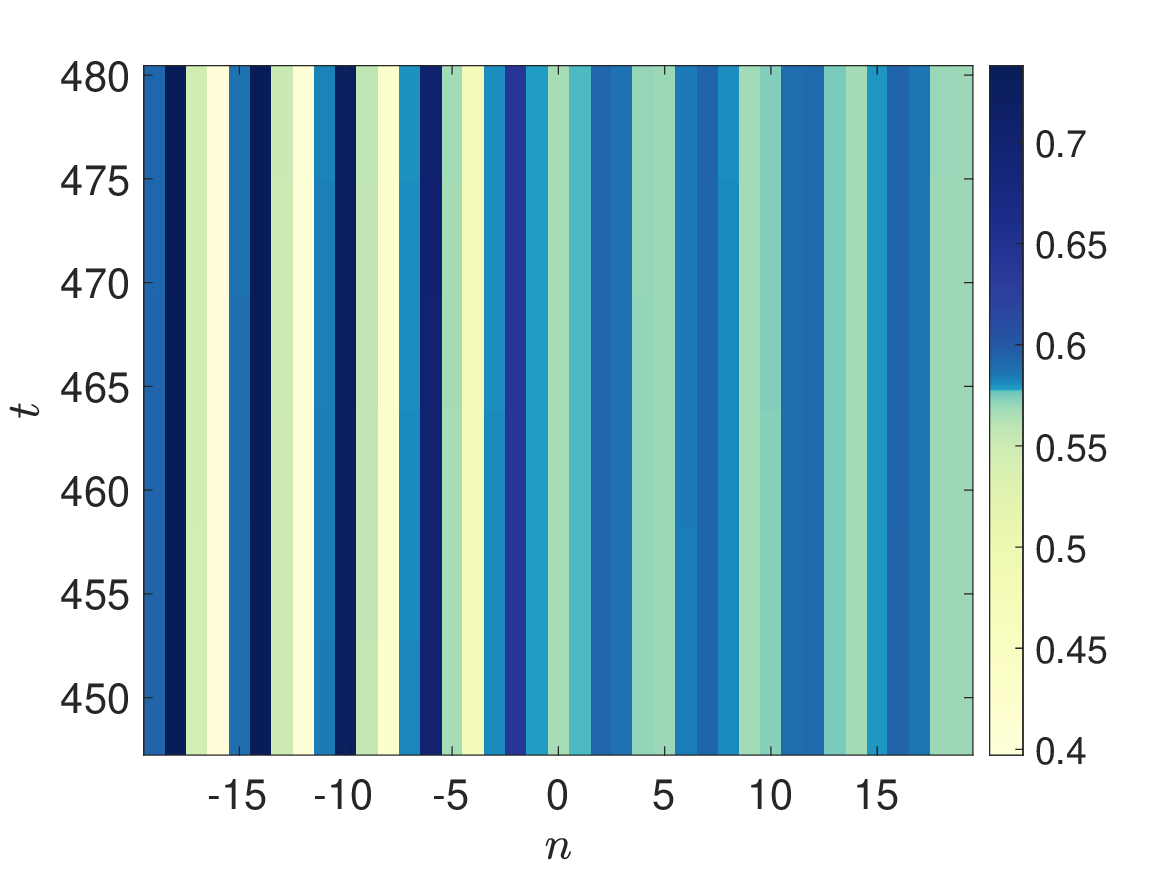} 
  \end{tabular}
  \caption{Intensity plot of the absolute value of the displacement with $u_=-0.16$ and $u_+=1$
    for the full simulation (a) and truncated simulation (b) for
    $n=-20 \ldots 20$ with forced boundaries. Panels (c) and (d) show
    the Poincar\'e map of the solution in (a) and (b), respectively,
    evaluated at every  period.}
    \label{fig:unsteady_solution_density}
\end{figure}

\subsubsection{Poincar\'e description}

A spatio-temporal intensity plot of the US with $u_- = -0.16$ near $n = 0$ is shown
in Fig.~\ref{fig:unsteady_solution_density}(a).  The
counterpropagating traveling waves can clearly be identified. While
the underlying wave parameters of the left- and right-moving waves
will generally be different, they do share the same frequency.  Thus,
the dynamics correspond to a time-periodic solution. Indeed,
inspection of Fig.~\ref{fig:unsteady_solution_density}(c) confirms
this, which shows the same intensity plot as panel (a), but with the
solution sampled every $T$ time units, where $T$ is the period of
oscillation. This is the Poincar\'e map of the dynamics.  With this
sampling size, the solution appears to be constant, suggesting that
the waveform is genuinely time-periodic. To demonstrate this further,
we simulated the equations of motion on a small lattice $n\in[-20,20]$
with boundary conditions given by the periodic solution, i.e., the
left boundary is given by $u_{-20}(t)$ and the right boundary is given
by $u_{20}(t)$. The dynamics upon initialization with the periodic
solution $u_n(450)$ are shown in
Fig.~\ref{fig:unsteady_solution_density}(b), which can be hardly
distinguished from the dynamics in (a).  The evolution remains
periodic, as can be inferred from
Fig.~\ref{fig:unsteady_solution_density}(d), which are the dynamics
sampled every $T$ seconds.  An avenue for potential further study,
prompted by these findings, is the seeking of exact time-periodic
solutions of the model and their corresponding Floquet analysis.

\subsubsection{Approximation of traveling waves via quasi-continuum model}
If considering the left-moving and right-moving waves as separate entities, we can once again apply the quasi-continuum
reduction to describe the traveling wave using formula Eq.~\eqref{eq:quasi-TW}. A comparison of the spatial profile
of the left wave moving wave with $u_+ = -0.16$ at time $t=480$ is shown in Fig.~\ref{fig:compare_qausi_unsteady}(a).
The phase plane for the left-moving waves (left lobes) and right-moving waves (right lobes) is shown in
Fig.~\ref{fig:compare_qausi_unsteady}(b) in markers, with corresponding quasi-continuum approximations shown
as solid lines. Like before, the reduction is best for smaller step heights (compare the blue and red 
orbits in panel (b)). A comparison of the mean, amplitude and frequency of the simulation and quasi-continuum
approximation as a function of $u_-$ is shown in panel (c) for both the left- and right-moving waves. We can observe that the approximation provides a very
adequate description of the relevant traveling patterns.

\begin{figure}
    \centering
       \begin{tabular}{@{}p{0.33\linewidth}@{}p{0.33\linewidth}@{}p{0.33\linewidth}@{} }
     \rlap{\hspace*{5pt}\raisebox{\dimexpr\ht1-.1\baselineskip}{\bf (a)}}
 \includegraphics[height=4.3cm]{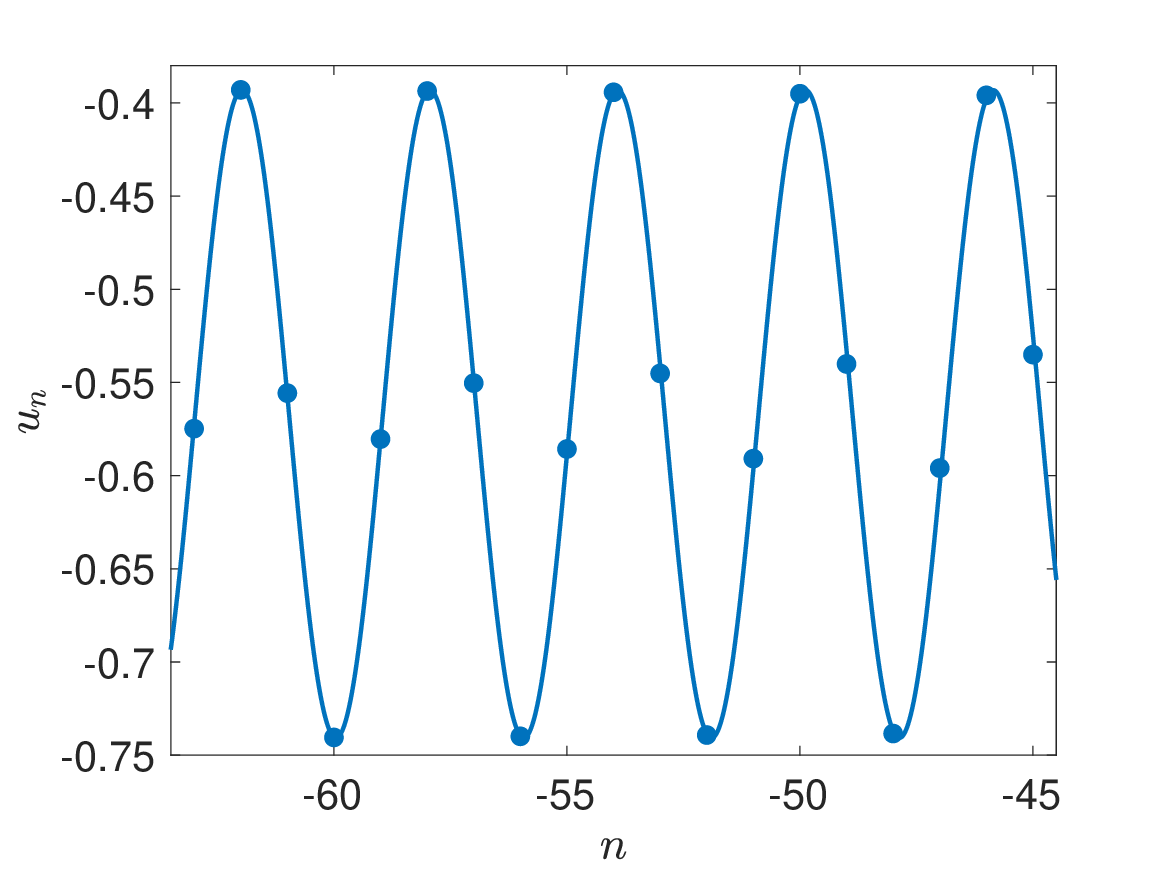} &
  \rlap{\hspace*{5pt}\raisebox{\dimexpr\ht1-.1\baselineskip}{\bf (b)}}
 \includegraphics[height=4.3cm]{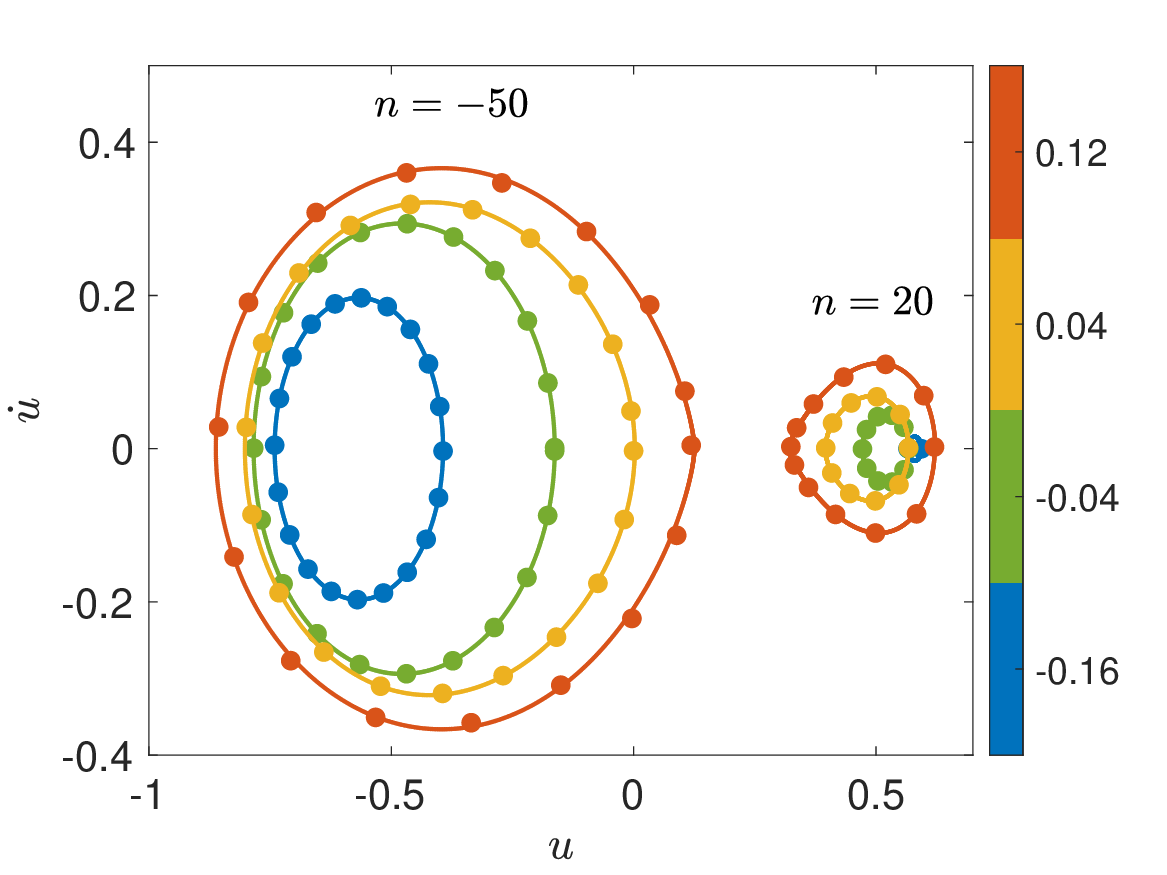} &
   \rlap{\hspace*{5pt}\raisebox{\dimexpr\ht1-.1\baselineskip}{\bf (c)}}
 \includegraphics[height=4.3cm]{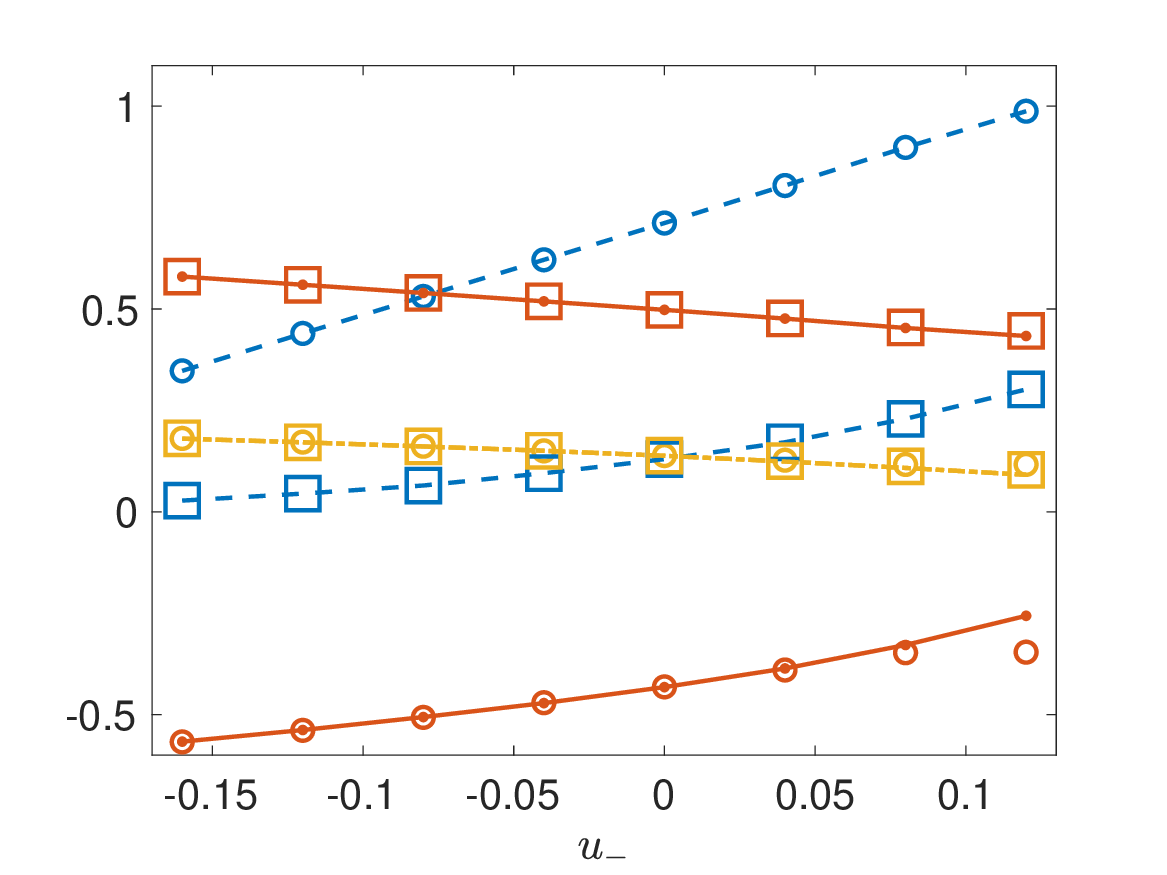} 
  \end{tabular}
    \caption{\textbf{(a)} Zoom of the left periodic wave with $u_- = -0.16$ (blue markers) and quasi-continuum approximation (solid blue line). \textbf{(b)} Plot of the phase plane $(u_{-50}(t),\dot{u}_{-50}(t))$
    (left lobes) and $(u_{20}(t),\dot{u}_{20}(t))$ (right lobes) for a time interval such that
    the periodic wave has developed. The color intensity corresponds to the value of $u_-$. In particular
    $u_- = -0.16$ (blue), $u_- = -0.04$ (green), $u_- = 0.04$ (yellow), $u_- = 0.12$ (red). The solid
    markers are the corresponding quasi-continuum approximations. \textbf{(c)} Plot
    of the mean (red solid lines), amplitude (blue dashed lines) and frequency $1/T$ (yellow dashed-dot line) as a function of $u_-$. The quasi-continuum approximation of these wave parameters
    is shown as open circles (for the left wave, $n=-50$) and open squares (for the right wave, $n=20$). }
    \label{fig:compare_qausi_unsteady}
\end{figure}

\subsubsection{Jump Conditions}

Figure \ref{fig:unsteady_solution_density} shows two
counterpropagating periodic traveling waves with a rapid transition
between them in the vicinity of $n = 0$.  This motivates the
hypothesis that these two waves satisfy the jump conditions obtained
from the Whitham modulation system's conservation laws.  We denote the
periodic traveling waves by $\varphi_\pm$ for the left ($-$) and right
($+$) periodic waves, respectively. For a discontinuous, shock
solution of the Whitham system at the origin, the corresponding jump
conditions are Eq.~\eqref{eq:19}
\begin{subequations}
  \label{eq:stat_jump}
  \begin{align}
    \left\langle \Phi'(\varphi_-) \right\rangle -  \left\langle
    \Phi'(\varphi_+) \right\rangle & = 0  \label{eq:jump1} \\ 
    \left\langle \Phi'(\varphi_-)\mathcal{S}\Phi'(\varphi_-)
    \right\rangle -  \left\langle
    \Phi'(\varphi_+)\mathcal{S}\Phi'(\varphi_+)\right\rangle
                                   & = 0  \label{eq:jump2}
    \\ 
    \omega_- - \omega_+ & = 0, \label{eq:jump3}
  \end{align}
\end{subequations}
where $S$ is the unit shift operator $SR(\eta) = R(\eta + 1)$. To
check if these jump conditions are indeed satisfied, we approximate
the above averages using the numerical simulations. In particular, we
let the structure come close to a periodic state (as in
Fig.~\ref{fig:unsteady_solution_density}) and extract one period of
motion at a particular node $n$. Let $T_n$ be the period of node $n$
and let $I_{T_n} = [\tau,\tau+T_n]$ be the corresponding time interval from
$t = \tau$.  We then make the following approximations
\begin{subequations}\label{eq:stat_jump_approx}
  \begin{align}
    \left\langle \Phi'(\varphi) \right\rangle
    &\approx \frac{1}{T_n}\int_{I_{T_n}} \Phi'(u_n(t)) dt =: f(n), \\
    \left\langle \Phi'(\varphi_-)\mathcal{S}\Phi'(\varphi_-)
    \right\rangle &\approx \frac{1}{T_n}\int_{I_{T_n}} \Phi'(u_n(t))
                    \Phi'(u_{n+1}(t)) dt =: g(n) ,
  \end{align}
\end{subequations}
for $\tau \gg 1$ (we set $\tau = 240$ in what follows).
The first jump condition, Eq.~\eqref{eq:jump1}, is checked by
comparing $f(-5)$ to $f(5)$, while the second jump condition
Eq.~\eqref{eq:jump2} is checked by comparing $g(-5)$ with $g(5)$.  The
third jump condition is simply the difference in the frequency, which
we estimate via $1/T_{-5}$ and $1/T_5$. Figure~\ref{fig:jump} shows a
plot of $f(-5)$ (open blue circles), $g(-5)$ (open blue squares) and
$1/T_{-5}$ (open blue triangles), while the red dots show the
corresponding quantities for $n=5$. Notice that each red point falls
nicely into an open blue marker, indicating that the jump conditions
are, up to some small numerical error, satisfied. The maximum residual
over the interval of $u_-$ values tested for the first jump condition
was $\displaystyle \max_{u_+} | f(-5) - f(5) | \approx 0.001$,
whereas the maximum residual for the second condition was $0.0017$ and
the maximum residual for the third was $0.0008$.

The numerical evidence is a compelling indication that the US can be
interpreted as a shock solution of the Whitham modulation equations.
While shock solutions of the Whitham equations have been constructed
previously \cite{Sprenger_2020,sprenger_traveling_2023}, their
admissibility requires the existence of traveling wave solutions of
the corresponding continuum PDE in which the phase velocities and
shock velocity all coincide.  In the present case of the US for the
lattice equation \eqref{eq:3}, all three of these velocities differ
but the frequencies are the same.  This suggests a new class of
admissible shock solutions to the Whitham equations corresponding to
time-periodic solutions of the lattice equation, an intriguing
possibility for future work.

\begin{figure}
  \centering
  \includegraphics[height=6cm]{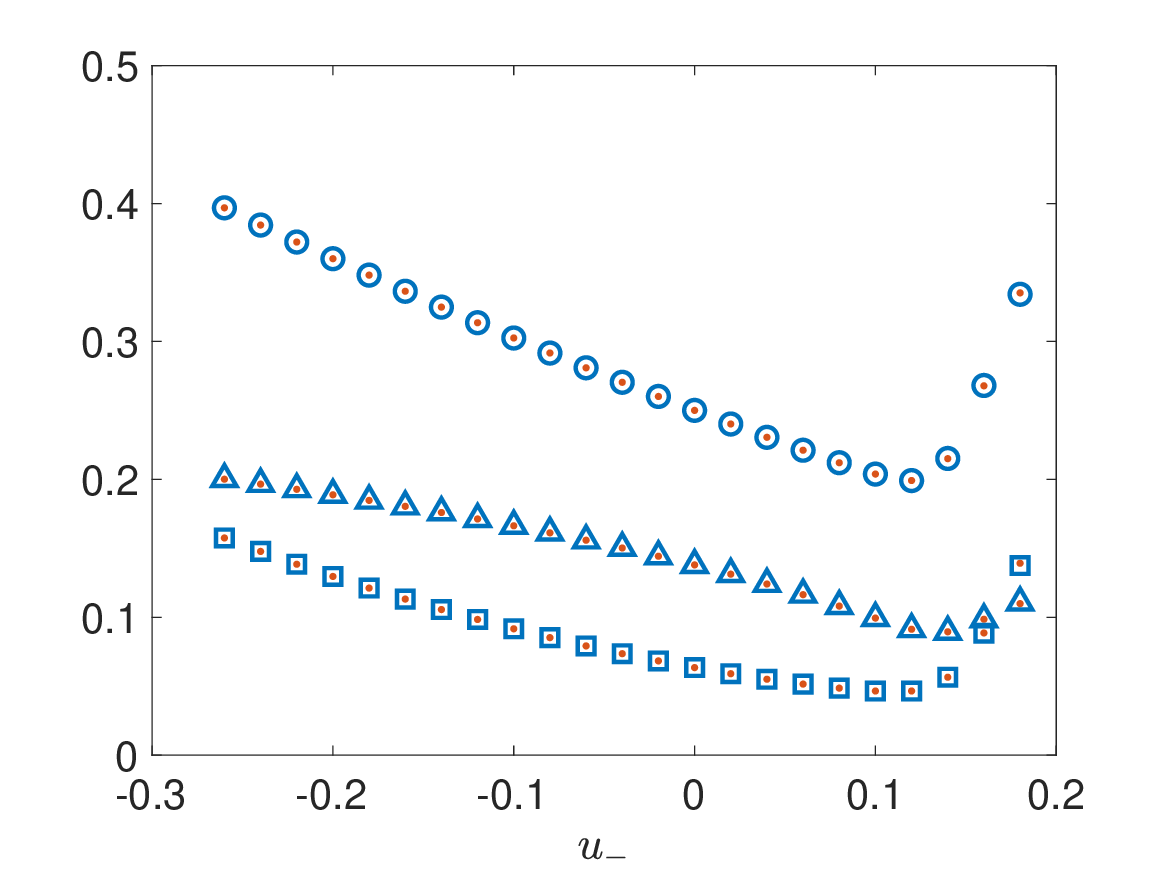} 
  \caption{Numerical verification of the jump conditions
    \eqref{eq:stat_jump} of the Whitham equations as $u_-$ is varied
    for the US solution.  Open blue circles:
    $\left\langle \Phi'(\varphi_-)|_{n=-5} \right\rangle \approx
    f(-5)$.  Open blue squares
    $\left\langle \Phi'(\varphi_-)\mathcal{S}\Phi'(\varphi_-)
    \right\rangle|_{n=-5} \approx g(-5)$.  Open blue triangles:
    $\omega_-|_{n=-5} \approx 1/T_{-5}$.  Red dots: corresponding
    quantities evaluated at $n = +5$.  Because the red dots lie inside
    the open blue markers, the jump conditions are satisfied to high
    accuracy.} 
  \label{fig:jump}
\end{figure}

We also note the similarity between the US and defect solutions of
reaction-diffusion equations
\cite{scheel}. Because the waves in the US are in-phase (see
Fig.~\ref{fig:unsteady_solution_density}), it most closely resembles a
target pattern with a source from which waves are emanating, in the
language of \cite{scheel}.  The target pattern exhibits a Hopf
bifurcation of the background state and a specific transition of the
eigenvalues associated with the linearized operator about the
background state.  It would be interesting to explore potential
connections between the underlying diffusive regularization of the
target pattern and the dispersive regularization of the US studied
here.


\subsection{Blow up}
\label{sec:blow-up}
\begin{figure}
    \centering
       \begin{tabular}{@{}p{0.33\linewidth}@{}p{0.33\linewidth}@{}p{0.33\linewidth}@{} }
  \rlap{\hspace*{5pt}\raisebox{\dimexpr\ht1-.1\baselineskip}{\bf (a)}}
 \includegraphics[height=4.3cm]{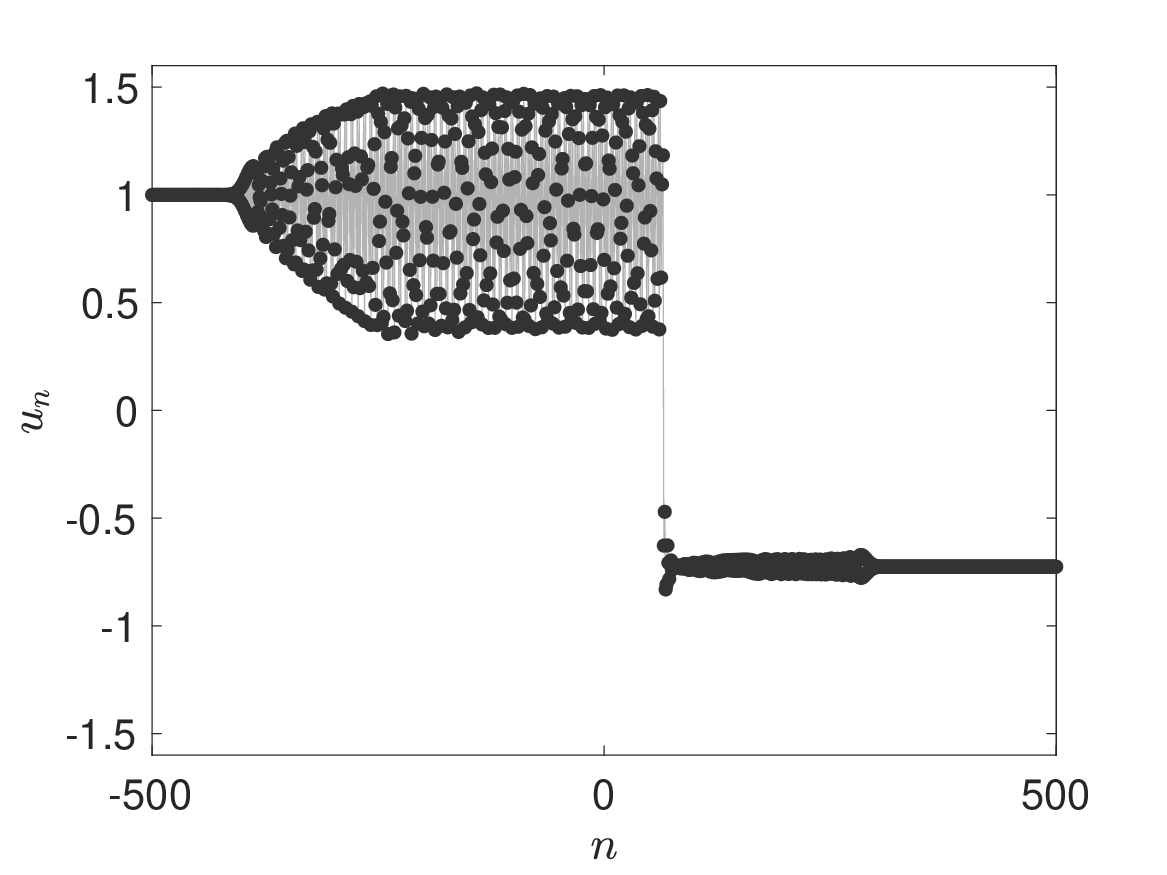} &
  \rlap{\hspace*{5pt}\raisebox{\dimexpr\ht1-.1\baselineskip}{\bf (b)}}
 \includegraphics[height=4.3cm]{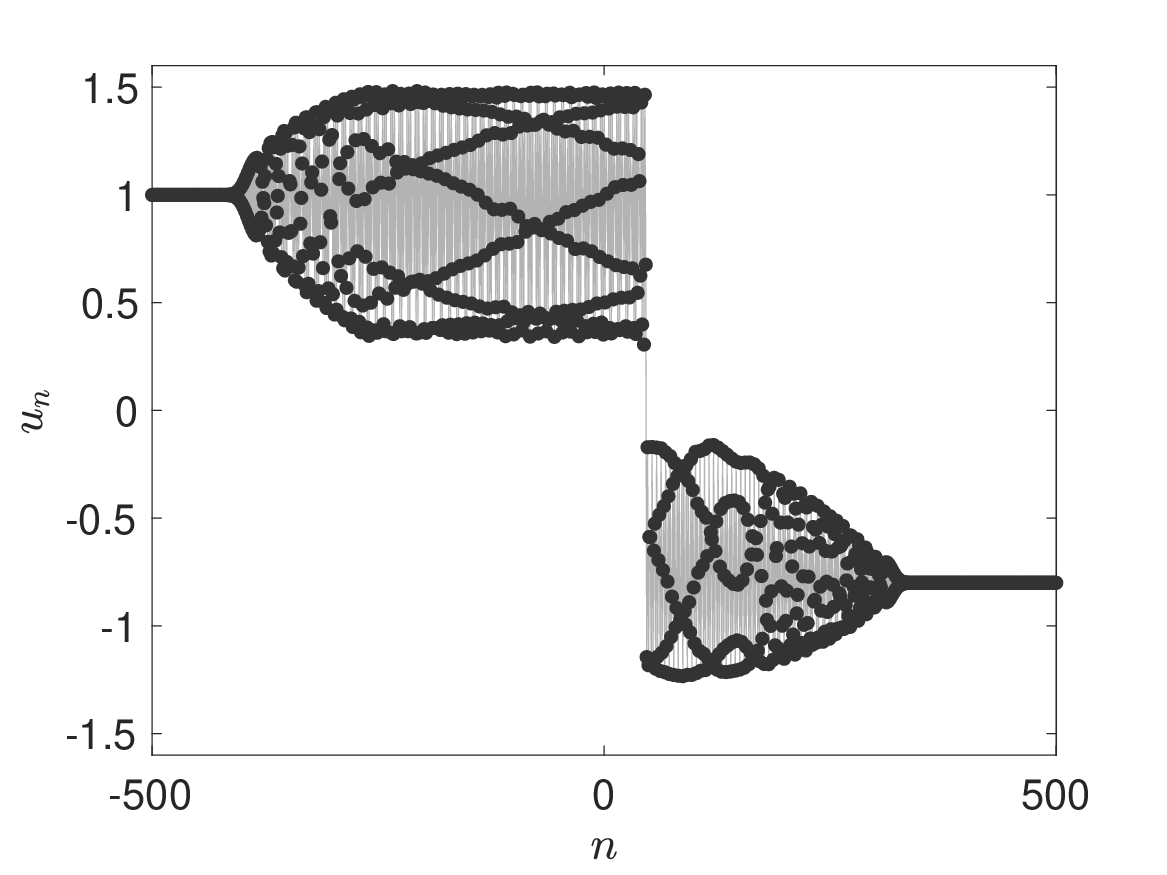} &
   \rlap{\hspace*{5pt}\raisebox{\dimexpr\ht1-.1\baselineskip}{\bf (c)}}
 \includegraphics[height=4.3cm]{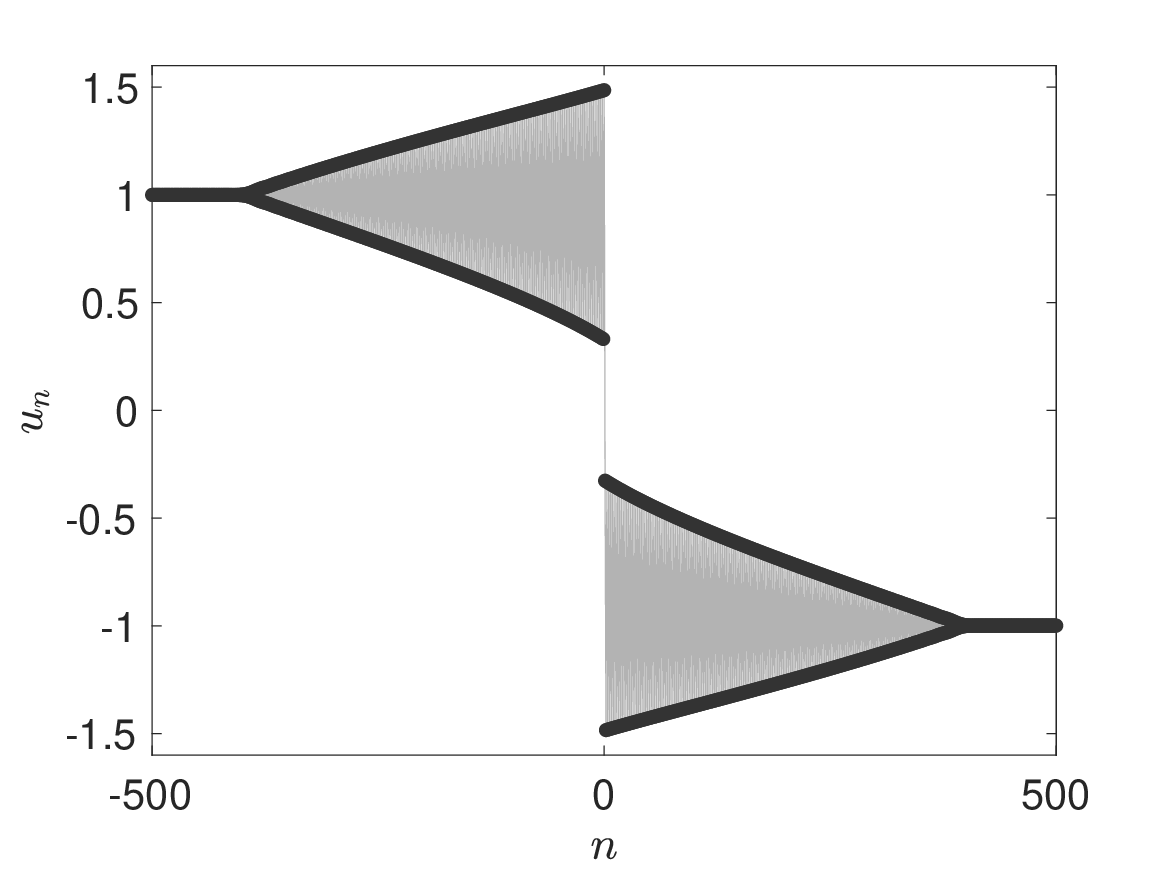} 
  \end{tabular}
    \caption{Snapshots of the spatial profile at $t=200$ for \textbf{(a)} $u_+ = -0.725$
    \textbf{(b)} $u_+ = -0.8$ and \textbf{(c)} $u_+ = -0.999$}
    \label{fig:before_blowup}
\end{figure}

In section~\ref{sec:pDSW}, we observed that a TDSW (a partial DSW
connected to a traveling wave) is generated that transitions from the
left level $u_- = 1$ to the right level $u_+ \in (-0.724,0)$. By
decreasing the value of $u_+$ below $-0.724$, a new wavetrain develops
on the right level. Figure \ref{fig:before_blowup} shows three example
profiles. When $u_+$ is close to the transition value $u_+=-0.724$,
the excitation on top of the right level $u_+$ is small, see Figure
\ref{fig:before_blowup}(a).  Decreasing $u_+$ further results in a
larger amplitude wavetrain that resembles another TDSW; see
Fig.~\ref{fig:before_blowup}(b).  Close to $u_+ =- 1$, the wavetrains
on the left and right levels approach binary oscillations, as shown in
Fig.~\ref{fig:before_blowup}(c).  Solution dynamics similar to those
shown in Fig.~\ref{fig:before_blowup}(c), but with odd initial data,
were studied extensively in \cite{wilma}.
In the work of~\cite{wilma}, it is claimed that the emergence of blow
up for the odd initial data they considered is ``almost always''
associated with the emergence of regions of binary oscillations for
which the upper and lower oscillatory envelopes have opposite sign.
Such a scenario was observed to result in the loss of hyperbolicity in
the modulation equations \eqref{eq:9} for binary oscillations and,
consequently, a dynamical instability and thus exponential growth in a
localized region of space that ultimately led to the finite time
blow up of the wave pattern.  For the Riemann data considered here
\eqref{step}, we numerically observe blow up in regions of the
solution where binary oscillations are not apparent, which we now explore.

\begin{figure}
    \centering
       \begin{tabular}{@{}p{0.4\linewidth}@{}p{0.4\linewidth}@{}}
  \rlap{\hspace*{5pt}\raisebox{\dimexpr\ht1-.1\baselineskip}{\bf (a)}}
 \includegraphics[height=4.5cm]{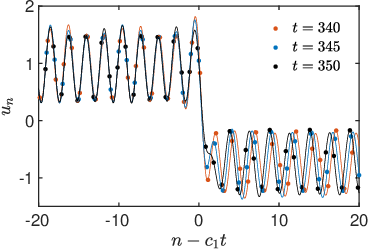} &
  \rlap{\hspace*{5pt}\raisebox{\dimexpr\ht1-.1\baselineskip}{\bf (b)}}
 \includegraphics[height=4.5cm]{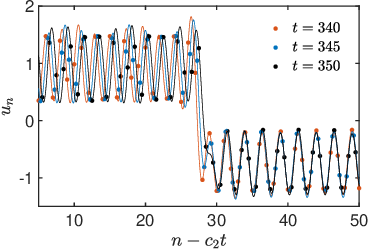}
  \end{tabular}
  \caption{Interpolated spatial profiles at different times in the
    co-moving frame for $u_+ = -0.8$ with speed \textbf{(a)}
    $c_1 = 0.229$, \textbf{(b)} $c_2 = 0.15$ demonstrating that the
    waves to the left and right of the sharp transition move at
    different speeds.}
    \label{fig:comoving}
\end{figure}
We will start with a more detailed discussion of the structure in
Fig.~\ref{fig:before_blowup}(b) with $u_+ = -0.8$, which is
representative of many of the patterns found for $u_+ \in (-1,-0.724)$
and $u_- = 1$. A zoom-in of the solution in the co-moving frame $n-ct$
near the shock interface is shown in Fig.~\ref{fig:comoving} at three
separate times ($t \in \{340,345,350\}$) represented by different
colors.  Both the solution on the lattice (dots) and its zero-padded
Fourier interpolant (curves) are shown with two different speeds.  In
Fig.~\ref{fig:comoving}(a) with speed $c = c_1 = 0.229$, the leftmost
wave at the three distinct times overlap, suggesting that it moves in
the steady frame $n-c_1t$.  Contrastingly, in
Fig.~\ref{fig:comoving}(b) with speed $c = c_2 = 0.15$, the rightmost
wave at the three distinct times overlap, suggesting it moves in the
slower steady frame $n-c_2 t$.  
The spatial profile at $t = 350$ and the Fourier transform
$\hat{u}(k)$ of a 40-site window of the leftmost (rightmost) wave are
shown in Fig.~\ref{fig:blowup}(a) in red (blue). Each wave has
wavenumber concentration, and they are distinct ($k\approx 2.76$ for
the left and $k=2.51$ for the right).  This solution is a
generalization of the unsteady shock (US) studied in
section~\ref{sec:unsteady} in that two traveling waves are connected
through a sudden jump. However, there are key differences.  First, the
shock interface itself is moving as shown in Fig.~\ref{fig:comoving}.
Second, the frequency for the leftmost and rightmost waves are not
identical, which can be seen in Fig.~\ref{fig:blowup}(c) that shows
the time series at a node located at the left wave ($n=20$) and the
right wave ($n=250$).  Here, it is clear that the oscillation
frequencies are distinct.  We conjecture that, much like the US, this
wavetrain can be modeled by a discontinuous, weak solution of the
Whitham equations satisfying the jump conditions \eqref{eq:19}.
Rather than pursue this further, we instead turn to an examination of
the large $t$ dynamics of the solution and, specifically, its eventual
blow up.

This structure is relatively coherent until about $t=600$, after which
small disturbances in the leftmost wave develop. Disturbances are
noticeable if one compares the time series for $t\in(350,400)$ (shown
in Fig.~\ref{fig:blowup}(c)) and for $t\in(750,800)$ (shown in
Fig.~\ref{fig:blowup}(d)). At about $t=816$ the solution appears to
experience finite-time blow up. Figure~\ref{fig:blowup_time}(a) shows
the blow up. The spatial profile close to (but before) the time of
blow up is shown in Fig.~\ref{fig:blowup}(b). Notice the location of
the blow up is spatially concentrated within the leftmost wave and that the
wavenumber of the traveling wave where the instability seems to
manifest is about $k\approx 2.76$, (i.e., not a binary oscillation).
\begin{figure}
    \centering
         \begin{tabular}{@{}p{0.45\linewidth}@{}p{0.45\linewidth}@{}}
     \rlap{\hspace*{5pt}\raisebox{\dimexpr\ht1-.1\baselineskip}{\bf (a)}}
 \includegraphics[height=5.5cm]{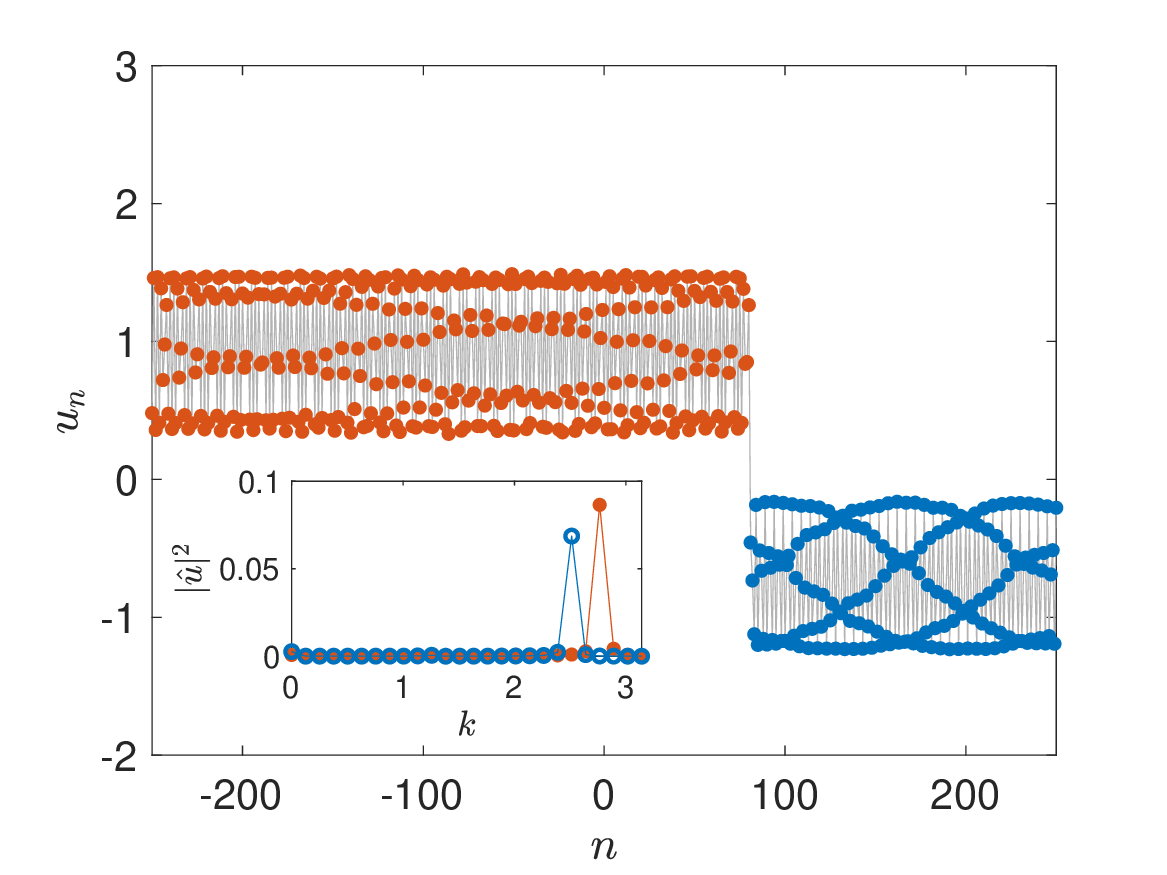} &
  \rlap{\hspace*{5pt}\raisebox{\dimexpr\ht1-.1\baselineskip}{\bf (b)}}
 \includegraphics[height=5.5cm]{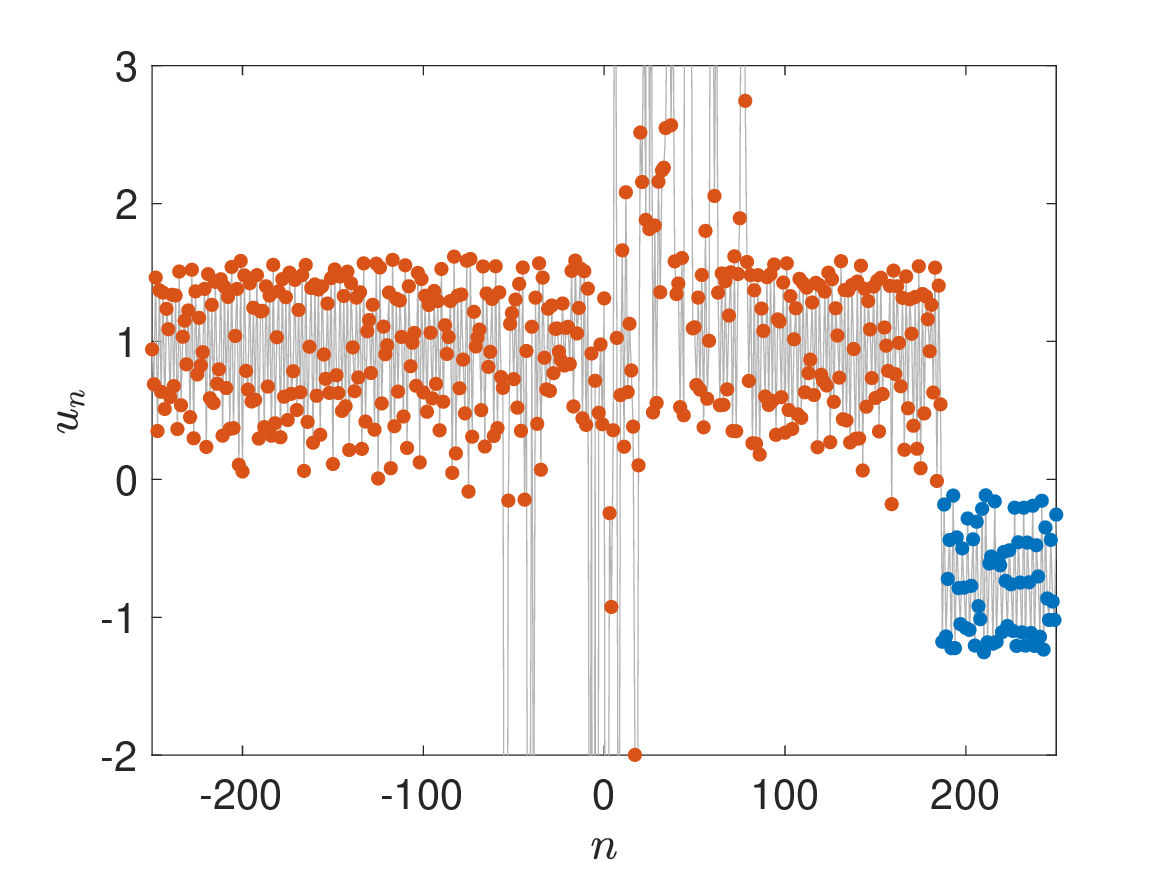} 
  \end{tabular}
       \begin{tabular}{@{}p{0.45\linewidth}@{}p{0.45\linewidth}@{}}
     \rlap{\hspace*{5pt}\raisebox{\dimexpr\ht1-.1\baselineskip}{\bf (c)}}
 \includegraphics[height=5.5cm]{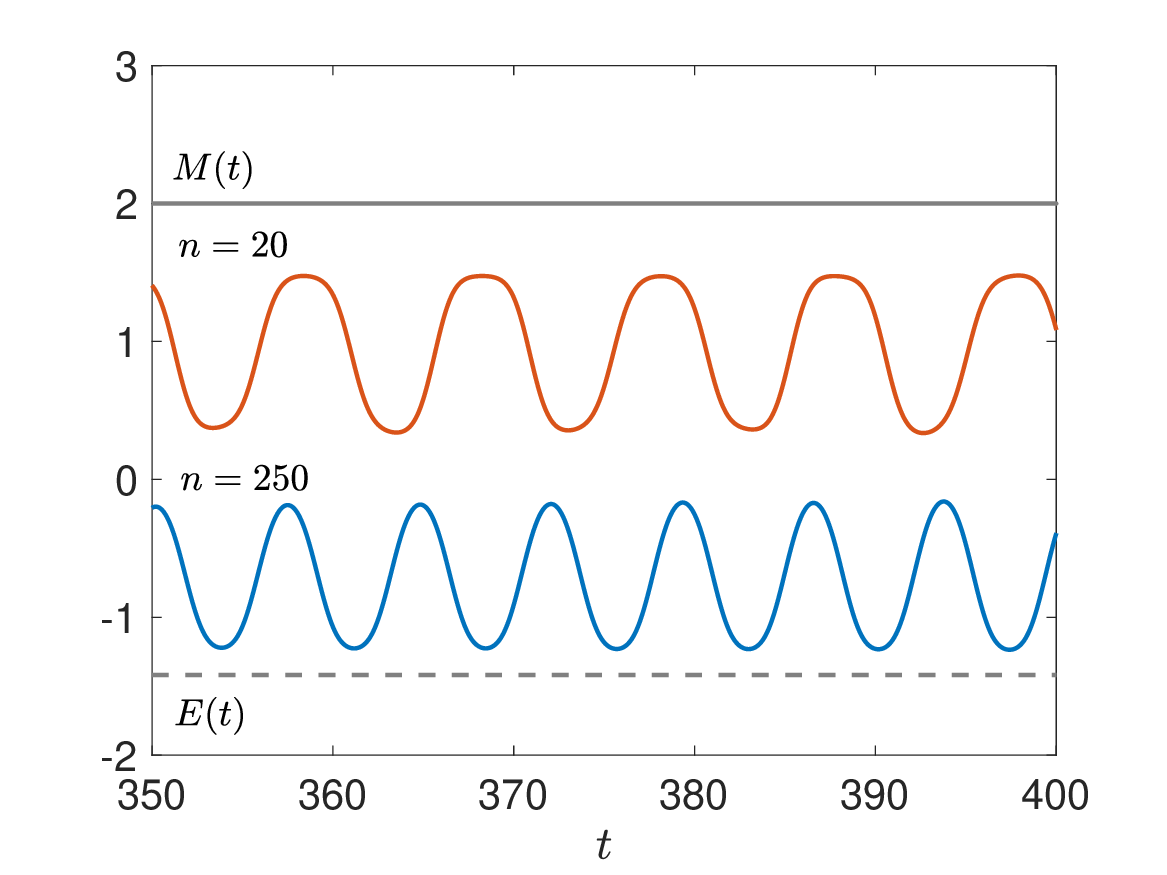} &
  \rlap{\hspace*{5pt}\raisebox{\dimexpr\ht1-.1\baselineskip}{\bf (d)}}
 \includegraphics[height=5.5cm]{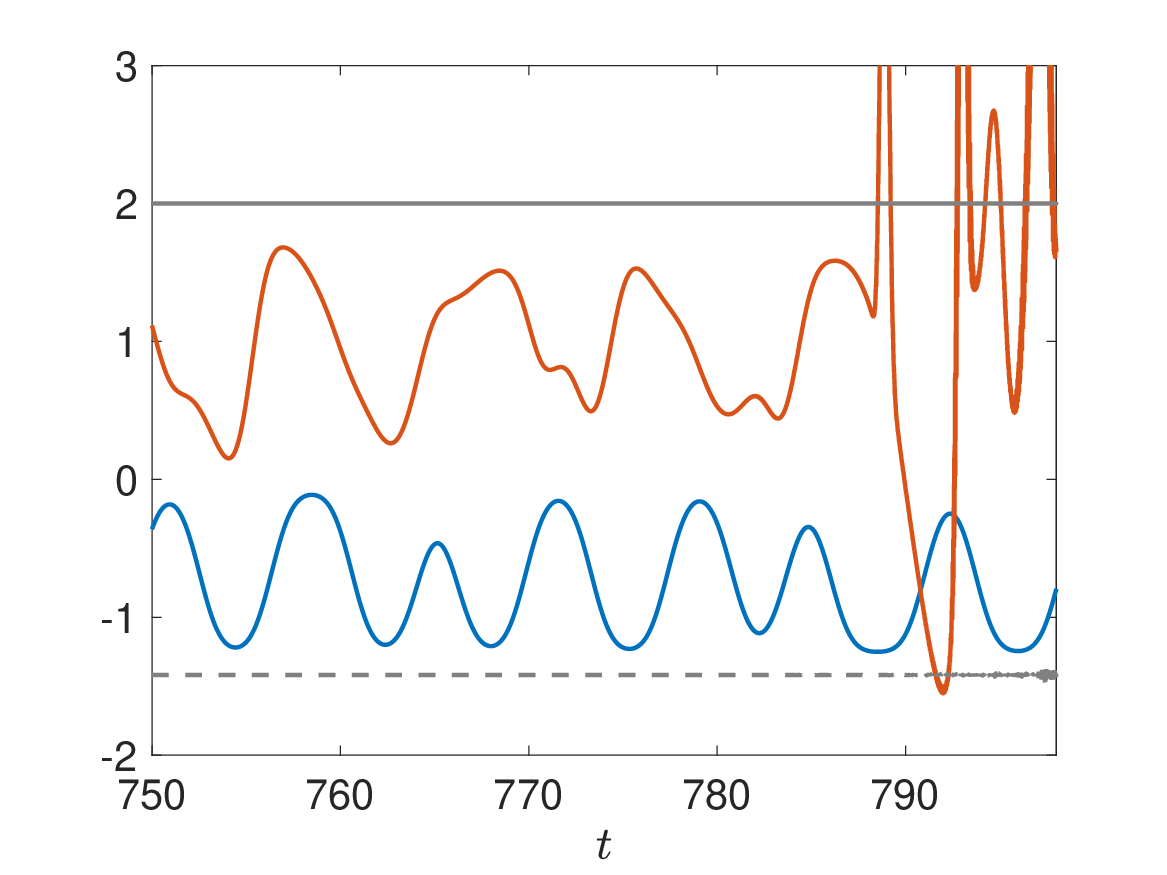} 
  \end{tabular}
  \caption{Various zooms of the solution with $u_- = 1$ and
    $u_+=-0.8$. \textbf{(a)} Zoom of spatial profile near the
    interface with $t=350$. The inset shows the spatial Fourier
    transform of the left wave (red) and right wave
    (blue). \textbf{(b)} Zoom of spatial profile near the interface
    with $t=800$, very close to the blow up time.  \textbf{(c)}
    Time-series of the $n=20$ node (red) and $n=250$ node (blue) much
    before the blow up. The mass $M(t)$ is shown as the solid gray
    line, and the energy $E(t)$ is also shown as a dashed gray line.
    Note that both are vertically displaced for visual purposes. The
    actual values are $E(0) =0.0813$ and $M(0)= 0.1$. \textbf{(d)}
    Same as (c), but for a time interval closer to the blow up
    time.}
    \label{fig:blowup}
\end{figure}

We conjecture that the observed blow up is due to an instability of
the leftmost wave with wavenumber $k\approx 2.76$, and not due to
numerical instability. A piece of evidence in this direction is that
the mass and energy are conserved until times very close to the blow
up time.  The solid gray and dashed lines of
Figs.~\ref{fig:blowup}(c,d) show the mass $M(t)$ and energy $E(t)$,
respectively. Note, in order for these quantities to be conserved, we
employ periodic boundary conditions.  
For the simulations in this
subsection, we concatenate the initial condition Eq.~\eqref{step} with
its reflection about the first site, leaving us with $2N$ total
nodes. Thus, the relevant window of space for the plots is only the
second half of the lattice. We define the lattice indices so that the
initial ($t=0$) jump from $u_-$ down to $u_+$ occurs at $n=0$. This
makes the plots consistent with those in the previous sections.  We
include the entire spatial window for the computation of $E(t)$ and
$M(t)$. Note that in Figs.~\ref{fig:blowup}(c,d) the quantities $M(t)$
and $E(t)$ are indeed conserved, even as the waveform begins to break
down, see $t\approx 790$ of panel Fig.~\ref{fig:blowup}(d). The gray
solid and dashed lines of Fig.~\ref{fig:blowup_time}(a) show plots of
$M(t)$ and $E(t)$, respectively, for times leading to the blow up
itself. While the energy $E(t)$ remains constant after the strong
onset of instability at about $t=785$, the energy conservation breaks
down for $t>800$, while $M(t)$ remains conserved. The conservation of
mass in the numerical scheme is not surprising, since by direct
computation one sees that the variational integrator applied to
Eq.~\eqref{eq:disc_mod} with $\Phi'(u) = u^2$ conserves the mass
exactly \cite{Herrmann_Scalar}. The near
conservation of energy for the variational integrator relies on the boundedness of
the underlying numerical solution \cite{Hairer}. This will be clearly
violated for solutions exhibiting 
collapse-type phenomena
and thus it is not surprising
that the energy is not conserved in the numerical scheme close to the time of blow up.
Thus, it seems the initial collapse
is due to instability of the wave (energy remains conserved for $t<790$), but after experiencing sustained large amplitude oscillations, the numerical scheme may begin to exhibit additional numerical instabilities (since energy is not conserved for $t>790$). The blow up time
found here is an approximation that depends on the particulars of the numerical scheme.

For $u_+ \in (-1,-0.724)$ and $u_- = 1$ we observe a similar blow up of the solutions,
with the blow up time varying roughly between $t=700$ and $t=2500$, see Fig.~\ref{fig:blowup_time}(b).
We define the solution as blown-up once $\max_n |u_n|$ exceeds a large threshold. We practically used $1000$ as the threshold
(the actual threshold makes little difference in Fig.~\ref{fig:blowup_time}(b)
since the blow up occurs so quickly). For this figure panel, we used a lattice size of $2N=100,000$
and simulated until $t=10,000$. We observe blow up with $u_+ = -0.725$ (at about $t=1500$) but no blow up
with $u_+ = -0.724$, even when simulating until $t=10,000$. This sharp transition between finite-time blow up
and no blow up is further evidence that the blow up is due to an underlying instability of the waveform
and not due to  numerical instability.

The question of stability of traveling waves of this lattice is thus
an important open question, meriting further investigation.  Based on
these findings, it appears that traveling waves with wavenumbers
other than $k=\pi$
can lead to instabilities and finite-time blow up, a
generalization of the findings in \cite{wilma} where binary
oscillations with $k = \pi$ were identified as a primary instability
mechanism.

\begin{figure}
    \centering
           \begin{tabular}{@{}p{0.4\linewidth}@{}p{0.4\linewidth}@{}}
     \rlap{\hspace*{5pt}\raisebox{\dimexpr\ht1-.1\baselineskip}{\bf (a)}}
 \includegraphics[height=5cm]{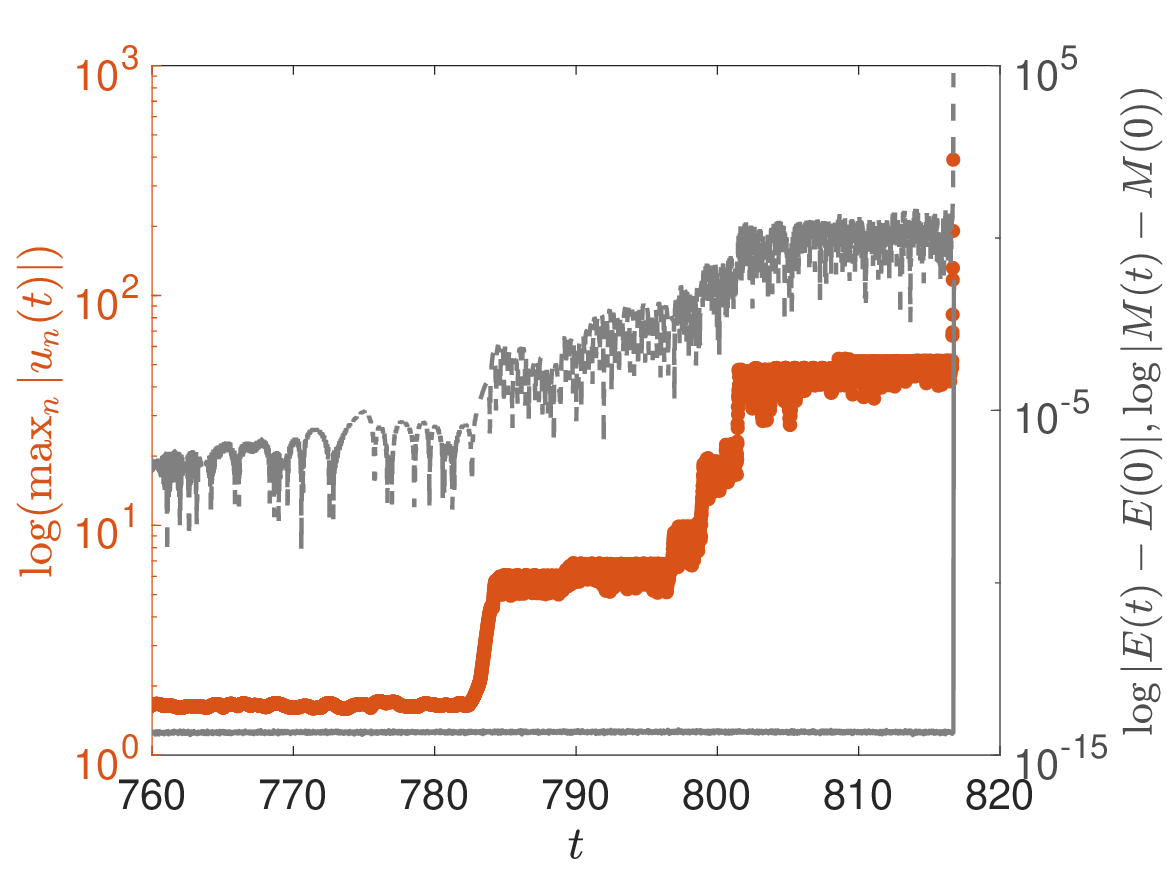} &
  \rlap{\hspace*{5pt}\raisebox{\dimexpr\ht1-.1\baselineskip}{\bf (b)}}
 \includegraphics[height=5cm]{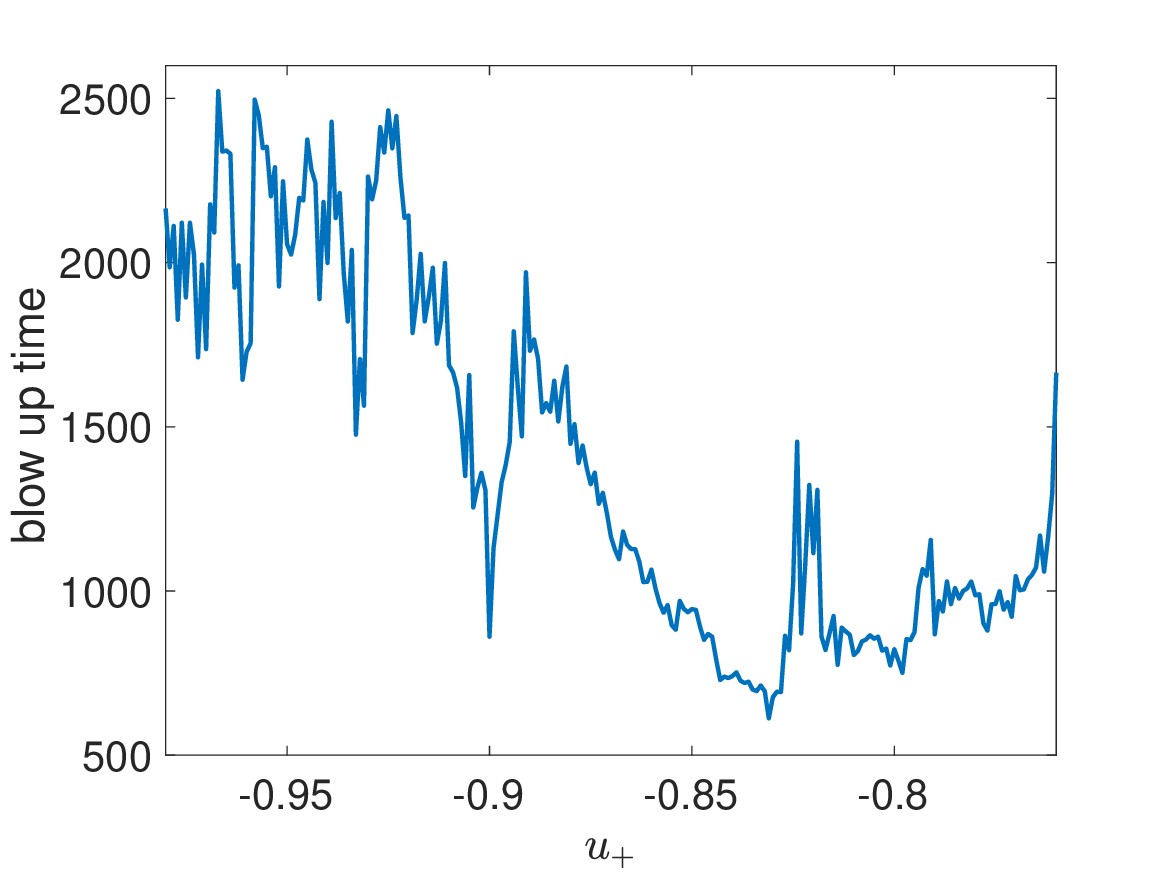} 
  \end{tabular}
  \caption{\textbf{(a)} Semilog plot of $\max_{n}|u_n(t)|$ vs. time
    leading to blow up (red). The semilog plot of $|E(t)-E(0)|$
    (dashed gray line) and $|M(t)-M(0)|$ (solid gray line) are also shown.
    \textbf{(b)} Observed blow up time for
    various values of $u_+$ with $u_-=1$ fixed. The lattice size is
    $2N=100,000$ and we simulate until $t=10,000$. For
    $u_- \geq -0.725$ we observed no blow up. 
    }
  \label{fig:blowup_time}
\end{figure}

\section{Non-uniqueness of Riemann problems}
\label{sec:non-uniq-riem}

As already mentioned in the introduction, the solutions to the Whitham system \eqref{eq:cont_mod} are non-unique for certain classes of initial data. To discuss this in greater detail, we start with some preliminary comments concerning the inviscid Burgers' PDE \eqref{eq:1}, which describes the dispersionless, hyperbolic scaling limit of the lattice \eqref{eq:3} in the case of no oscillations and is a subsystem of \eqref{eq:cont_mod} that governs the dynamics in the case of zero wave amplitude or zero wavenumber. To obtain an elementary example for non-uniqueness, we impose the odd initial condition
\begin{align}
\label{Nonuni2}
u\left(x,0\right)=\mathrm{sign}\left(x\right)
\end{align} 
and notice that the rarefaction wave
\begin{align}
\label{Nonuni5}
u\left(x,t\right)=U\left(x/t\right)\qquad \text{with}\qquad
U\left(\xi\right)=\left\{\begin{array}{lcc}
-1&\text{for}&-\infty<\xi<-2\\
\xi/2&\text{for}&-2<\xi<+2\\
+1&\text{for}&+2<\xi<+\infty\\
\end{array}\right.
\end{align}
is the unique self-similar solution according to the classical theory of hyperbolic conservation laws, see for instance \cite{LeF02}. The latter complements the PDE with an additional selection rules to exclude solutions that are considered to be unphysical. The most prominent example is the Lax condition for shocks but without such an admissibility criterion there exists a plethora of possibilities to fulfill the initial value problem \eqref{eq:1}+\eqref{Nonuni2}. For instance, the formulas
\begin{align}
\label{Nonuni3}
U\left(\xi\right)=\left\{\begin{array}{lcr}
-1&\text{for}&-\infty<\xi<-2\\
\xi/2&\text{for}&-2<\xi<-2\,\mu\\
-\mu&\text{for}&-2\,\mu<\xi<0\\
+\mu&\text{for}&0<\xi<+2\,\mu\\
\xi/2&\text{for}&+2\,\mu<\xi<+2\\\
+1&\text{for}&+2<\xi<+\infty\\
\end{array}\right.
\qquad\text{and}\qquad
U\left(\xi\right)=\left\{\begin{array}{lcr}
-1&\text{for}&-\infty<\xi<-1-\mu\\
-\mu&\text{for}&-1-\mu<\xi<0\\
+\mu&\text{for}&0<\xi<+\mu+1\\
+1&\text{for}&+\mu+1<+\infty\\
\end{array}\right.
\end{align}
provide two families of further self-similar solutions in dependence of the real parameter $0<\mu<1$ and $\mu>1$, respectively, and contain \eqref{Nonuni5} as limiting case for $\mu=0$. The corresponding  profile functions $U$ are illustrated in Figure \ref{FigNonuni1} and combine a steady discontinuity at $x=0$ (which violates the Lax condition on both sides) with either two rarefaction waves or two Lax shocks. 
\begin{figure}[h!]
\centering{\includegraphics{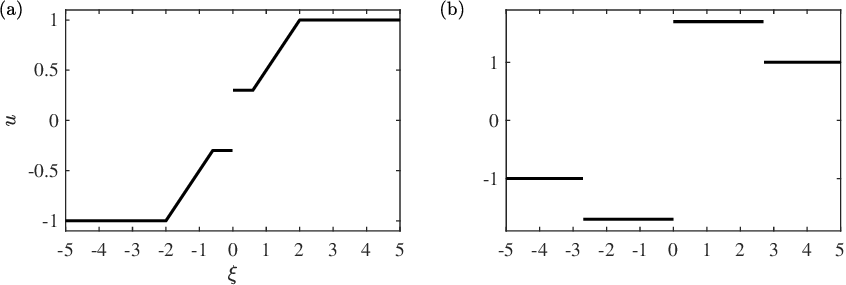}}
\caption{Solutions from \eqref{Nonuni3} with (a) $\mu = 0.3$ and (b)
  $\mu = 1.7$ to the inviscid Burgers' equation \eqref{eq:1} with initial data \eqref{Nonuni2}. 
}
\label{FigNonuni1}
\end{figure}
In numerical simulations of the lattice \eqref{eq:3} we find a related family of solutions that differ only in microscopic details of the imposed initial conditions. As a typical but still elementary example we study the lattice initial data 
\begin{align}
\label{Nonuni4}
u_n\left(0\right)=\left\{\begin{array}{lcr}
-1&\text{for}&n<0\\
\eta&\text{for}&n=0\\
+1&\text{for}&n>0\\
\end{array}\right.
\end{align}
where the free parameter $\eta$ reflects that the macroscopic Riemann data \eqref{Nonuni2} can be realized in many different ways on the microscopic scale. Figure \ref{FigNonuni2} reveals that the numerical solutions for different choices of $\eta$ correspond to distinct waves on the macroscopic scale. We always find a steady discontinuity at $x=0$ but its jump height as well as the other waves depend crucially on the value of microscopic parameter. In analogy to Figure \ref{FigNonuni1} we observe two rarefaction waves in the case of $0<\eta<1$ but for $\eta>1$  we have to replace the Lax shocks by dispersive shock waves.
\par
It remains a challenging task to understand the macroscopic impact of small-scale fluctuations in the initial data. For instance, it seems that the parameter $\eta$ in \eqref{Nonuni4} determines the jump height of the emerging steady discontinuity in Figure \ref{FigNonuni2} but we are not aware of any heuristic or even rigorous explanation thereof. Numerical simulations also indicate that microscopic details might {not} be relevant for Riemann data that are either positive or negative on both sides but also this must be investigated more thoroughly. Both issues are also intimately related to the properties of the Whitham system \eqref{eqn:WhithamAbstract} which can be regarded as an extension of the Burgers' equation \eqref{eq:1}. At least for sign changing Riemann problems,  the lattice \eqref{eq:3} is able to produce an entire family  of Whitham solutions and it is very natural to investigate and classify them in terms of selection criteria and entropy inequalities. Moreover, other nonlinearities might produce further effects due to linearly degenerated states in the scalar first order PDE corresponding to \eqref{deq}; see \cite{HR10,HR10b} for a related problem in FPUT chains whose force functions is increasing with inflection point.
\begin{figure}[H]
\centering{\includegraphics{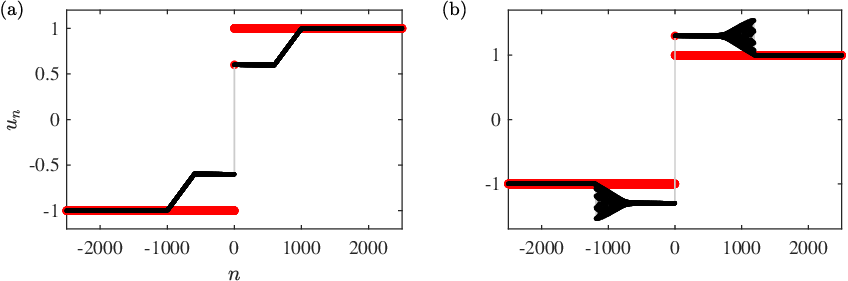}}
\caption{Lattice solution with initial data \eqref{Nonuni4} and two choices of the microscopic parameter (a) $\eta = 0.5$ and (b) $\eta = 1.3$ at $t = 500$. The initial data is shown in red. 
}
\label{FigNonuni2}
\end{figure}

\section{Conclusions and Future Challenges}
\label{sec:concl-future-chall}

In the present work, we revisit the first order nonlinear dynamical
lattice of Eq.~\eqref{deq} through the lens of ``lattice hydrodynamics'' by
providing a systematic analysis of the solutions to the canonical
lattice Riemann problem \eqref{step} for the case of quadratic flux
\eqref{eq:3}. Building on earlier works
of~\cite{wilma,Herrmann_Scalar,CHONG2022133533}, we have characterized
dispersive shock waves and rarefaction waves, familiar from continuum
dispersive hydrodynamics.  But we also discover a variety of
non-classical hydrodynamic-like solutions.  In addition to finite time
blow up, recognized earlier in the work of~\cite{wilma},
we identify three additional dynamical regimes that are interpreted
using a combination of numerical simulation, quasi-continuum
approximation, and Whitham modulation theory. These regimes include
the generation of a counterpropagating DSW and RW pair separated by a
stationary, abrupt shock. Additionally, an abrupt, unsteady transition
between two counterpropagating periodic waves with the same frequency
is interpreted as a shock solution of the Whitham modulation
equations.
Finally, the phase diagram \ref{fig:example_sols} of solutions to the
lattice Riemann problem is rounded out by a traveling DSW (TDSW),
consisting of an unsteady partial DSW, connected to a
heteroclinic periodic-to-equilibrium traveling wave.

These elements of the lattice model's phenomenology provide an
opportunity to develop different aspects of the mathematical analysis
of lattice hydrodynamics. For instance, we adapt the DSW fitting
method of~\cite{El2005} to lattice equations in order to characterize
the expansion and edge properties of lattice DSWs.  We also follow up
the work of~\cite{CHONG2022133533}, providing abstract modulation
equations for a quasi-continuum analogue of the lattice dynamical
system. We connect these to the genuinely discrete modulation
equations, made explicit in the weakly nonlinear regime by a
Poincar{\'e}-Lindstedt expansion of periodic traveling wave solution
profiles and their frequencies.  Hyperbolicity and self-similar
solutions of the modulation equations are used to characterize lattice
hydrodynamics. The quasi-continuum theory was also leveraged
elsewhere, such as in characterizing the leading edge amplitude of the
DSW and the periodic/traveling wave solutions of the discrete model.
A truncation of the dynamics in wavenumber space allows for a
linearized, large $t$ analysis of small amplitude oscillations that
accompany the RW and DSW solutions.

Beyond the binary oscillations that have been studied previously, this
work has identified new lattice hydrodynamic features that, so far,
appear not to have a continuum, dispersive hydrodynamic parallel.  Of
particular interest is the rapid, unsteady transition between two
in-phase counterpropagating periodic traveling waves with the same
frequency (US in Fig.~\ref{fig:example_sols}), identified with a shock
solution of the Whitham modulation equations.  This solution is born
out of the bifurcation of another lattice solution particular
to the lattice, the
counterpropagating DSW, RW pair separated by a stationary shock (DSW+SS+RW
in Fig.~\ref{fig:example_sols}) when the DSW merges with the stationary
shock.  The US represents an unsteady generalization of the steady
periodic-to-equilibrium heteroclinic solutions of higher order
continuum dispersive equations, themselves interpreted as admissible
shock solutions of the corresponding Whitham modulation equations
\cite{Sprenger_2020,sprenger_traveling_2023}.  This work points to a
new class of admissible ``Whitham shocks'' in the lattice context.  It
will be interesting to see if continuum solutions of this type also
exist.

The asymptotic and semi-analytical tools used in this work constitute
an effective framework in which to investigate the hydrodynamics of
other lattice dynamical systems for which dispersive shock wave
phenomena may arise.  For example, the Riemann problem for FPUT chains
generalizes to two wave families (second order in time) the problem of
a single wave family studied here and has been shown to exhibit a
variety of lattice hydrodynamic solutions including steady transition fronts \cite{gorbushin_transition_2022}
and DSWs \cite{chong_integrable_2024}.

A number of interesting open questions remain. Here, we list some of
these. While the existence of periodic traveling waves has been proven
\cite{Herrmann_Scalar}, their stability and the existence, as well as stability
of a more general class of traveling waves consisting of
periodic-to-periodic heteroclinic solutions are important problems to
better understand lattice hydrodynamics, with implications for finite
time blow up and the construction of the TDSW solution. In the
DSW+SS+RW solution, we have not been able to provide a mechanism that
leads to the particular selection of the intermediate constant
$u_0=(u_+-u_-)/2$.  We suspect that this is due to microscopic details
of the lattice equation that have been neglected in the analysis.  Of
similar, unknown origin is the selection mechanism
for shock solutions of the binary modulation
equations \eqref{eq:9} investigated in detail in \cite{wilma}.  On the
other hand, the unsteady shock constitutes a genuine time-periodic
orbit. Such periodic orbits are worthy of exploration in their own
right, including from the perspective of dynamical systems by, for
example, generalizing the spatial dynamics in dispersive
\cite{sprenger_traveling_2023} and dissipative \cite{scheel} systems.
The existence, stability and bifurcation analysis of such states are
intriguing problems for future exploration.  
As discussed in~\cite{wilma}, the model with quadratic flux
\eqref{eq:3} analyzed herein is one among several
possible discretizations of the inviscid Burgers' equation. A
comparative study of different discretizations of the more general
conservation law $u_t + \Phi'(u)_x = 0$
could lead to other lattice hydrodynamics and, in the case of
integrable discretizations, be amenable to deeper mathematical
analysis. These topics constitute some of the open problems
within the theme of DSWs in one spatial dimension.  The study of
higher-dimensional models appears to be wide open.
Some of these topics are currently under consideration and will be
reported in future publications.


\section*{Acknowledgments}
The authors would like to thank the Isaac Newton Institute for
Mathematical Sciences for support and hospitality during the programme
Dispersive Hydrodynamics when work on this paper was undertaken (EPSRC
Grant Number EP/R014604/1). This material is also based upon work
supported by the US National Science Foundation under Grants
DMS-2107945 (C.C. and E.O.), DMS-2204702 (P.G.K.) and PHY-2110030
(P.G.K), DMS-2306319 (M.A.H.). This work began during P.S.'s time as INI-Simons Post Doctoral Research Fellow, and this author would like to thank the INI for support and hospitality during this fellowship, which was supported by Simons Foundation (Award ID 316017).

\bibliographystyle{unsrt}
\bibliography{Chong}

\begin{thebibliography}{10}

\bibitem{Whitham74}
G.B. Whitham.
\newblock {\em Linear and Nonlinear Waves}.
\newblock Wiley, New York, 1974.

\bibitem{LeF02}
P.~G. LeFloch.
\newblock {\em Hyperbolic systems of conservation laws}.
\newblock Lectures in Mathematics ETH Z\"{u}rich. Birkh\"{a}user Verlag, Basel,
  2002.
\newblock The theory of classical and nonclassical shock waves.

\bibitem{lax_small_1983}
P.~D. Lax and C.~D. Levermore.
\newblock The small dispersion limit of the {Korteweg-de Vries} equation: 1-3.
\newblock {\em Comm. Pure Appl. Math.}, 36(3,5,6):253--290; 571--593; 809--830,
  1983.

\bibitem{GP73}
A.~V. {Gurevich} and L.~P. {Pitaevskii}.
\newblock {Nonstationary structure of a collisionless shock wave}.
\newblock {\em J. Exp. Theor. Phys.}, 65:590--604, 1973.

\bibitem{Mark2016}
G.A. El and M.A. Hoefer.
\newblock Dispersive shock waves and modulation theory.
\newblock {\em Physica D}, 333:11, 2016.

\bibitem{Sprenger_2020}
P.~Sprenger and M.~A. Hoefer.
\newblock Discontinuous shock solutions of the whitham modulation equations as
  zero dispersion limits of traveling waves.
\newblock {\em Nonlinearity}, 33(7):3268, may 2020.

\bibitem{congy_dispersive_2021}
T.~Congy, G.~A. El, M.~A. Hoefer, and M.~Shearer.
\newblock Dispersive {{Riemann}} problems for the {Benjamin--Bona--Mahony}
  equation.
\newblock {\em Stud. Appl. Math.}, 147(3):1089--1145, 2021.

\bibitem{hou_dispersive_1991}
T.~Y. Hou and P.~D. Lax.
\newblock Dispersive approximations in fluid dynamics.
\newblock {\em Commun. Pure Appl. Math.}, 44(1):1{\textendash}40, 1991.

\bibitem{ehrnstrom_existence_2012}
M.~Ehrnstr{\"o}m, M.~D. Groves, and E.~Wahl{\'e}n.
\newblock On the existence and stability of solitary-wave solutions to a class
  of evolution equations of {{Whitham}} type.
\newblock {\em Nonlinearity}, 25(10):2903--2936, 2012.

\bibitem{hur_modulational_2015}
V.~Y. Hur and M.~A. Johnson.
\newblock Modulational instability in the {{Whitham}} equation with surface
  tension and vorticity.
\newblock {\em Nonlinear Anal. Thoer.}, 129:104--118, 2015.

\bibitem{binswanger_whitham_2021}
Adam~L. Binswanger, Mark~A. Hoefer, Boaz Ilan, and Patrick Sprenger.
\newblock Whitham modulation theory for generalized {{Whitham}} equations and a
  general criterion for modulational instability.
\newblock {\em Stud. Appl. Math.}, 147(2):724--751, 2021.

\bibitem{benjamin_model_1972}
T.~B. Benjamin, J.~L. Bona, and J.~J. Mahony.
\newblock Model {{Equations}} for {{Long Waves}} in {{Nonlinear Dispersive
  Systems}}.
\newblock {\em Phil. Trans. Roy. Soc. Lond. Ser. A}, 272(1220):47--78, 1972.

\bibitem{rosenau2}
P.~Rosenau.
\newblock Hamiltonian dynamics of dense chains and lattices: {O}r how to
  correct the continuum.
\newblock {\em Phys. Lett. A}, 311:39, 2003.

\bibitem{rosenau1}
P.~Rosenau.
\newblock Dynamics of nonlinear mass-spring chains near the continuum limit.
\newblock {\em Phys. Lett. A}, 118:222, 1986.

\bibitem{Nester2001}
V.F. Nesterenko.
\newblock {\em Dynamics of Heterogeneous Materials}.
\newblock Springer-Verlag, New York, 2001.

\bibitem{wilma}
C.~V. Turner and R.~R. Rosales.
\newblock The small dispersion limit for a nonlinear semidiscrete system of
  equations.
\newblock {\em Stud. Appl. Math.}, 99:205, 1997.

\bibitem{dsw}
G.A. El, M.A. Hoefer, and M.~Shearer.
\newblock Dispersive and diffusive-dispersive shock waves for nonconvex
  conservation laws.
\newblock {\em SIAM Review}, 59(1):3, 2017.

\bibitem{Hascoet2000}
E.~Hascoet and H.~J. Herrmann.
\newblock Shocks in non-loaded bead chains with impurities.
\newblock {\em Eur. Phys. J. B}, 14:183, 2000.

\bibitem{Herbold07}
E.~B. Herbold and V.~F. Nesterenko.
\newblock Solitary and shock waves in discrete strongly nonlinear double
  power-law materials.
\newblock {\em Appl. Phys. Lett.}, 90(26):261902, 2007.

\bibitem{Molinari2009}
A.~Molinari and C.~Daraio.
\newblock Stationary shocks in periodic highly nonlinear granular chains.
\newblock {\em Phys. Rev. E}, 80:056602, 2009.

\bibitem{HEC_DSW}
H.~Kim, E.~Kim, C.~Chong, P.~G. Kevrekidis, and J.~Yang.
\newblock Demonstration of dispersive rarefaction shocks in hollow elliptical
  cylinder chains.
\newblock {\em Phys. Rev. Lett.}, 120:194101, 2018.

\bibitem{fleischer2}
S.~Jia, W.~Wan, and J.~W. Fleischer.
\newblock Dispersive shock waves in nonlinear arrays.
\newblock {\em Phys. Rev. Lett.}, 99:223901, Nov 2007.

\bibitem{talcohen}
Jian Li, S~Chockalingam, and Tal Cohen.
\newblock Observation of ultraslow shock waves in a tunable magnetic lattice.
\newblock {\em Phys. Rev. Lett.}, 127:014302, Jun 2021.

\bibitem{holian_atomistic_1995}
B.~L. Holian.
\newblock Atomistic computer simulations of shock waves.
\newblock {\em Shock Waves}, 5(3):149--157, 1995.

\bibitem{first_DSW}
D.~H. Tsai and C.~W. Beckett.
\newblock Shock wave propagation in cubic lattices.
\newblock {\em J. of Geophys. Res.}, 71(10):2601, 1966.

\bibitem{duvall_steady_1969}
George~E. Duvall, R.~Manvi, and Sherman~C. Lowell.
\newblock Steady {{Shock Profile}} in a {{One}}-{{Dimensional Lattice}}.
\newblock {\em Journal of Applied Physics}, 40(9):3771--3775, 1969.

\bibitem{gurevich_expanding_1984}
A.~V. Gurevich and A.~P. Meshcherkin.
\newblock Expanding self-similar discontinuities and shock waves in dispersive
  hydrodynamics.
\newblock {\em Zhurnal Eksp. Noi Teor. Fiz.}, 87:1277--1292, 1984.

\bibitem{holian_molecular_1979}
B.~L. Holian and G.~K. Straub.
\newblock Molecular {{Dynamics}} of {{Shock Waves}} in {{Three-Dimensional
  Solids}}: {{Transition}} from {{Nonsteady}} to {{Steady Waves}} in {{Perfect
  Crystals}} and {{Implications}} for the {{Rankine-Hugoniot Conditions}}.
\newblock {\em Phys. Rev. Lett.}, 43(21):1598--1600, 1979.

\bibitem{mossman_dissipative_2018}
M.~E. Mossman, M.~A. Hoefer, K.~Julien, P.~G. Kevrekidis, and P.~Engels.
\newblock Dissipative shock waves generated by a quantum-mechanical piston.
\newblock {\em Nat. Commun.}, 9(1):4665, 2018.

\bibitem{CHONG2022133533}
C.~Chong, M.~Herrmann, and P.G. Kevrekidis.
\newblock Dispersive shock waves in lattices: A dimension reduction approach.
\newblock {\em Physica D}, 442:133533, 2022.

\bibitem{Herrmann_Scalar}
M.~Herrmann.
\newblock Oscillatory waves in discrete scalar conservation laws.
\newblock {\em Math. Mod. Meth. Appl. S.}, 22(01):1150002, 2012.

\bibitem{FPU55}
E.~Fermi, J.~Pasta, and S.~Ulam.
\newblock {Studies of Nonlinear Problems. I.}
\newblock {\em Tech. Rep.}, (Los Alamos National Laboratory, Los Alamos, NM,
  USA):LA--1940, 1955.

\bibitem{Whitham90}
G.~B. Whitham.
\newblock Exact solutions for a discrete system arising in traffic flow.
\newblock {\em P. Roy. Soc. A-Math-Phy.}, 428(1874):49--69, 1990.

\bibitem{LAX86}
Peter~D. Lax.
\newblock On dispersive difference schemes.
\newblock {\em Physica D}, 18(1):250--254, 1986.

\bibitem{goodman_dispersive_1988}
J.~Goodman and P.~D. Lax.
\newblock On dispersive difference schemes. {{I}}.
\newblock {\em Commun. Pure Appl. Math.}, 41(5):591{\textendash}613, 1988.

\bibitem{KAC1975160}
M~Kac and Pierre {van Moerbeke}.
\newblock On an explicitly soluble system of nonlinear differential equations
  related to certain toda lattices.
\newblock {\em Adv. Math.}, 16(2):160--169, 1975.

\bibitem{toda}
M.~Toda.
\newblock {\em Theory of nonlinear lattices}.
\newblock Springer-Verlag, Berlin, 1989.

\bibitem{Holian81}
B.L. Holian, H.~Flaschka, and D.~W. McLaughlin.
\newblock Shock waves in the {T}oda lattice: Analysis.
\newblock {\em Phys. Rev. A}, 24:2595, 1981.

\bibitem{kodama_solutions_1990}
Y.~Kodama.
\newblock Solutions of the dispersionless {{Toda}} equation.
\newblock {\em Phys. Lett. A}, 147(8):477--482, 1990.

\bibitem{Bloch_Toda92}
A.~Bloch and Y.~Kodama.
\newblock Dispersive regularization of the {W}hitham equation for the {T}oda
  lattice.
\newblock {\em SIAM J. on Appl. Math.}, 52(4):909--928, 1992.

\bibitem{venakides_toda_1991}
S.~Venakides, P.~Deift, and R.~Oba.
\newblock The toda shock problem.
\newblock {\em Comm. Pure Appl. Math.}, 44(8-9):1171--1242, 1991.

\bibitem{biondini_whitham_2023}
G.~Biondini, C.~Chong, and P.~Kevrekidis.
\newblock On the {{Whitham}} modulation equations for the {{Toda}} lattice and
  the quantitative characterization of its dispersive shocks.
\newblock {\em arXiv:2312.10755}, 2023.

\bibitem{benzoni-gavage_slow_2014}
S.~{Benzoni-Gavage}, P.~Noble, and L.~M. Rodrigues.
\newblock Slow {{Modulations}} of {{Periodic Waves}} in {{Hamiltonian PDEs}},
  with {{Application}} to {{Capillary Fluids}}.
\newblock {\em J. Nonlinear Sci.}, 24(4):711--768, 2014.

\bibitem{Venakides99}
A.~M. Filip and S.~Venakides.
\newblock Existence and modulation of traveling waves in particles chains.
\newblock {\em Comm. Pure and Appl. Math.}, 52(6):693, 1999.

\bibitem{DHM06}
W.~Dreyer, M.~Herrmann, and A.~Mielke.
\newblock Micro-macro transition in the atomic chain via {W}hitham's modulation
  equation.
\newblock {\em Nonlinearity}, 19(2):471, 2005.

\bibitem{DH08}
W.~Dreyer and M.~Herrmann.
\newblock Numerical experiments on the modulation theory for the nonlinear
  atomic chain.
\newblock {\em Physica D}, 237(2):255, 2008.

\bibitem{HR10}
M.~Herrmann and J.~D.~M. Rademacher.
\newblock Riemann solvers and undercompressive shocks of convex {FPU} chains.
\newblock {\em Nonlinearity}, 23:277, 2010.

\bibitem{dafermos_hyperbolic_2016}
C.~M. Dafermos.
\newblock {\em Hyperbolic {{Conservation Laws}} in {{Continuum Physics}}}.
\newblock {Springer}, {Berlin}, 4th ed. edition, 2016.

\bibitem{El2005}
G.~A. El.
\newblock {Resolution of a shock in hyperbolic systems modified by weak
  dispersion}.
\newblock {\em Chaos}, 15(3):037103, 10 2005.

\bibitem{sprenger_traveling_2023}
P.~Sprenger, T.~J. Bridges, and M.~Shearer.
\newblock Traveling {{Wave Solutions}} of the {{Kawahara Equation Joining
  Distinct Periodic Waves}}.
\newblock {\em J Nonlinear Sci}, 33(5):79, 2023.

\bibitem{MielkePatz2017}
Alexander Mielke and C.~Patz.
\newblock Uniform asymptotic expansions for the fundamental solution of
  infinite harmonic chains.
\newblock {\em Z. Anal. Anwend.}, 36:437--475, 2017.

\bibitem{Kamchatnov}
A~M Kamchatnov.
\newblock {\em Nonlinear Periodic Waves and Their Modulations}.
\newblock World Scientific, 2000.

\bibitem{scheel}
B.~Sandstede and A.~Scheel.
\newblock Defects in oscillatory media: Toward a classification.
\newblock {\em SIAM J. Appl. Dyn. Syst.}, 3(1):1--68, 2004.

\bibitem{Hairer}
E.~Hairer, S.P. N{\o }rsett, and G.~Wanner.
\newblock {\em Solving Ordinary Differential Equations I}.
\newblock Springer-Verlag, Berlin, Germany, 1993.

\bibitem{HR10b}
M.~Herrmann and J.~D.~M. Rademacher.
\newblock Heteroclinic travelling waves in convex {FPU}-type chains.

\bibitem{gorbushin_transition_2022}
N.~Gorbushin, A.~Vainchtein, and L.~Truskinovsky.
\newblock Transition fronts and their universality classes.
\newblock {\em Phys. Rev. E}, 106(2):024210, 2022.

\bibitem{chong_integrable_2024}
C.~Chong, A.~Geisler, P.~G. Kevrekidis, and G.~Biondini.
\newblock Integrable {{Approximations}} of {{Dispersive Shock Waves}} of the
  {{Granular Chain}}.
\newblock {\em arXiv:2402.08218}, 2024.

\end{thebibliography}

\end{document}